\documentclass[12pt]{article} % Comment this line out
\usepackage{graphicx}
\usepackage{verbatim}
\usepackage{natbib}
\usepackage[table,usenames,dvipsnames]{xcolor}
\usepackage{longtable}
\usepackage{afterpage}
\usepackage[toc,page]{appendix}
\usepackage[hypertexnames=false,colorlinks=true]{hyperref}
\usepackage[toc,nonumberlist]{glossaries}
\usepackage{authblk}
\usepackage{amsmath,amsthm}
\usepackage{amsfonts}
\usepackage{caption}
\usepackage{rotating}
\usepackage[nottoc,numbib]{tocbibind}

\DeclareMathOperator*{\argmax}{arg\,max}
\DeclareMathOperator*{\argmin}{arg\,min}

\LTcapwidth=\textwidth % for longtable caption to fit the page width

\newtheorem{upperBoundAvgDMP}{Proposition}
\newtheorem{minRidgeEquivalentMAP}[upperBoundAvgDMP]{Proposition}

\title{A new foreperiod effect on single-trial phase coherence.\\
Part I: existence and relevance}

\author[1,2]{Joaquin Rapela\thanks{rapela@ucsd.edu}} 
\author[1,3]{Marissa Westerfield}
\author[3]{Jeanne Townsend} 
\author[1]{Scott Makeig}
\affil[1]{Swartz Center for Computational Neuroscience; University of
California San Diego}
\affil[2]{Instituto de Investigaci\'{o}n en Luz, Ambiente y Visi\'{o}n; Universidad
Nacional de Tucum\'{a}n--Consejo Nacional de Investigaciones Cient\'{i}ficas y
T\'{e}cnicas, Argentina}
\affil[3]{Research on Autism and Development Laboratory; University of
California San Diego}

\newglossaryentry{trial} {
    name={trial},
    description={a 500~ms-long interval following the presentation of a
standardard stimulus}
}

\newglossaryentry{epoch} {
    name={epoch},
    description={a 700~ms-long interval following the presentation of a
standardard stimulus}
}

\newglossaryentry{svSubBlock} {
    name={switch-to-vision subblock},
    description={all segments of a \gls{switchBlock} between a LOOK and a
HEAR cue.  These segments correspond to brief periods of time where subjects
had oriented their attention to the visual modality, while they were quickly
switching their attention between the visual and auditory modality in a
longer \gls{switchBlock} (Section~\ref{sec:experimentalDesign})}
}

\newglossaryentry{saSubBlock} {
    name={switch-to-audition subblock},
    description={all segments of a \gls{switchBlock} between a HEAR and a LOOK
cue.  These segments correspond to brief periods of time where subjects had
oriented their attention to the auditory modality, while they were quickly
switching their attention between the visual and auditory modality in a
longer \gls{switchBlock} (Section~\ref{sec:experimentalDesign})}
}

\newglossaryentry{standardModality} {
    name={standard modality},
    plural={standard modalities},
    description={the modality of the standard stimuli (visual or auditory) 
used to align a set of epochs (Section~\ref{sec:characterizedEpochs})}
}

\newglossaryentry{attendedModality} {
    name={attended modality},
    plural={attended modalities},
    description={the attended modality (visual or auditory) corresponding to a
set of epochs 
(Section~\ref{sec:characterizedEpochs})}
}

\newglossaryentry{meanPhase} {
    name={mean phase},
    description={mean direction (Section~\ref{sec:circularStats}) of a set of phases}
}

\newglossaryentry{DMP} {
    name={DMP},
    description={(deviation from the mean phase) single-trial measure of
\gls{ITPC} (Section~\ref{sec:dmp})}
}

\newglossaryentry{warningSignal} {
    name={warning signal},
    description={stimulus initiating a period of expectancy for a forthcoming
impendent stimulus. The LOOK and HEAR attention-shifting cues are the warning
signals in this study (Section~\ref{sec:experimentalDesign})}
}

\newglossaryentry{ITPC} {
    name={ITPC},
    description={(inter trial phase coherence) degree of phase alignment across multiple trials. Here we use the \gls{DMP} and the \gls{ITC} as measures of \gls{ITPC} for single trials and for groups of trials, respectively (Section~\ref{sec:introduction})}
}

\newglossaryentry{SFP} {
    name={SFP},
    description={(standard foreperiod) interval between the presentation of a
\gls{warningSignal} and a standard stimulus (Section~\ref{sec:introduction})}
}

\newglossaryentry{SFPD} {
    name={SFPD},
    description={(standard foreperiod duration) duration of a \gls{SFP} (Section~\ref{sec:introduction})}
}

\newglossaryentry{IC} {
    name={IC},
    description={(independent component) component from an ICA decomposition (Section~\ref{sec:icaPreprocessing})}
}

\newglossaryentry{ITC} {
    name={ITC},
    description={(inter trial coherence) a measure of \gls{ITPC} in a
group of trials (Section~\ref{sec:itcAndPeakITCFrequency})}
}

\makeglossaries

\begin{document}

% \tableofcontents

\pagebreak

\setcounter{page}{1}

\maketitle

\abstract{

Expecting events in time leads to more efficient behavior. A remarkable early
finding in the study of temporal expectancy is the foreperiod effect on
reaction times; i.e., the influence or reaction time of the time period between
a warning signal and an imperative stimulus to which subjects are instructed to
respond as quickly as possible. Recently it has been shown that the phase of
oscillatory activity preceding stimulus presentation is related to behavior. 
Here we connect both of these findings by reporting a novel foreperiod effect
on the inter-trial phase coherence of the electroencephalogram (EEG) triggered
by stimuli to which subjects are instructed not to respond. 
Inter-trial phase coherence has been used to describe regularities in phases of
groups of trials time locked to an event of interest. We propose a single-trial
measure of inter-trial phase coherence and prove its soundness.
Equipped with this measure, and using a multivariate decoding method, we
demonstrate that the foreperiod duration in and audiovisual attention-shifting
task modulates single-trial phase coherence.
In principle, this modulation could be an artifact of the decoding method used
to detect it. We show that this is not the case, since the modulation can also
be observed using a simple averaging method.
We show that the strength of this modulation correlates with subject behavior
(both error rates and mean-reaction times).
We anticipate that the new foreperiod effect on inter-trial phase coherence,
and the decoding method used here to detect it, will be important tools to
understand cognition at the single-trial level. In Part II of this manuscript,
we support this claim, by showing that changes in attention modulate the
strength of the new foreperiod effect on a trial-by-trial basis.

}

\vspace{1.0\baselineskip}


\section{Introduction}
\label{sec:introduction}

The scientific study of temporal expectation begun at the birth of experimental
psychology when \citet{wundt1874} and \citet{woodrow14} discovered the effect
on reaction times of the duration of the interval between a \gls{warningSignal}
and a subsequent imperative stimulus to which subjects were instructed to
respond as quickly as possible. This delay is termed the foreperiod and this
effect the foreperiod effect on reaction time. Here we report a new foreperiod
effect on the inter-trial phase coherence (\gls{ITPC}, the degree of alignment
in phases from multiple trials) of electroencephalographic (EEG) signals evoked
by standard stimuli to which subjects are instructed not to respond, following
a \gls{warningSignal}.  We call this effect the standard foreperiod (\gls{SFP})
effect on \gls{ITPC}.

When foreperiod durations are fixed across
a block of trials, reaction times increase with increasing foreperiod
duration~\citep{wundt1874, woodrow14}. This effect is
commonly explained as arising from a reduced precision of temporal expectancy
with increasing time,
due to the
deterioration of accuracy in time estimation for longer time
intervals~\citep{grondin01}. When foreperiod durations vary from
trial to trial, however, reaction times are commonly found to decrease with increasing
foreperiod durations~\citep{naatanen70, woodrow14}. This effect is often
explained in terms of conditional probabilities, i.e., the probability that
the imperative stimulus will occur in a next small time window given that
it has not occurred yet. As time passes during a given foreperiod without
the impending stimulus having been presented, the conditional probability of its
occurrence increases in the reminder of the foreperiod. The brain
presumably learns this changing conditional probability early in the
experimental session and exploits it to heighten temporal expectation. A
review on foreperiod effects appears in~\citet{niemiAndNaatanen81}.

% Review of CNV
% References from tecce72, rockstrohEtAl83, mento13
Electrophysiological correlates of the foreperiod effect have mainly been
investigated through the Contingent Negative
Variation~\citep[CNV,][]{walterEtAl64}. The CNV is a slow negative shift in
the base line of the EEG that develops between the \gls{warningSignal} and
the imperative stimulus.
\citet{walterEtAl64} showed that the CNV does not reflect sensory
activity, but the contingency of the
\gls{warningSignal} and
the imperative stimulus; i.e., the probability that imperative stimulus follows
the \gls{warningSignal}. If the imperative
stimulus is suddenly omitted the CNV amplitude is reduced, and the CNV
disappears after 20-50 trials without the imperative stimulus~\citep{lowEtAl66}.
Its amplitude is significantly elevated when a motor response is required to
the imperative stimulus, compared to when no motor response is
needed~\citep{walterEtAl64}.
The CNV is not discernible in raw EEG records and requires trial averaging (6-12
trials in normal adults).
% It has most often been recorded in the vertex (Cz electrode in the 10-20
% system). 
% The appearence of the CNV requires that the \gls{warningSignal} signal is followed by the
% imperative stimulus. 
% The amplitude of the CNV is increased for low-intensity imperative
% stimuli~\citep{lowEtAl67}, suggesting that attending to the imperative
% stimulus is relevant for the development of the CNV.
% The CNV is typically
% recorded in experiments with a fixed foreperiod of length between 1 and 1.5
% seconds. With longer fixed foreperiod lengths, 8 and 6 seconds, respectively,
% \citet{weertsAndLang73} and \citet{lovelessAndSanford75} reported that the
% CNV is composed of an orienting wave, or O wave, that succeeds the
% \gls{warningSignal} by approximately one second, and the expectancy wave, or E wave, which
% rises in anticipation to the imperative stimulus. 
The CNV is generated in the
cortex, requiring inputs from the basal ganglia, and involving cerebral
networks where the thalamus plays a critical
role~\citep{bruniaAndVanBoxtel01}. 
% In blocks of trials with fixed foreperiod duration, the effect of the
% foreperiod duration on the CNV was studied by \citet{mcadamEtAl69}.  Across
% blocks of increasing foreperiod duration (800ms.\, 1500ms.\ and 4800ms.),
% \citet{mcadamEtAl69} confirmed the reaction time increased previously
% reported by \citet{walterEtAl64}, and showed a concurrent decrease of CNV
% amplitude.
% The effect of the foreperiod duration on the CNV in blocks of variable
% foreperiod duration was investigated by \citet{loveless73}. With five
% equiprobable foreperiod durations, ranging from 500 to 900 msec in steps of
% 100 msec. CNV amplitude was found to increase with foreperiod duration
% concomitant with a decrease in reaction times.
The investigation of the cognitive
functions that are reflected on the CNV has a long and rich history. Walter
and colleagues~[\citeyear{walterEtAl64}] first discussed the CNV as a sign of
expectancy.  \citet{tecce72} postulated attention as a correlate of the CNV.
\citet{gaillard77, gaillard78} proposed that the late E wave of CNV may
reflect preparation for the motor response.
% The
% CNV has also been associated with the ``preparatory set;'' i.e., a mental
% state where a subject is prepared to perform a (motor or non-motor)
% action~\citep[][p. 781]{lowEtAl66}. Information
% processing~\citep[e.g.,][]{fordEtAl79, marshEtAl76} and
% motivation~\citep{irwinEtAl66, rebertEtAl67} have been connected to the CNV.
More recent experiments have linked the CNV with
explicit~\citep{macarAndVidal03, pfeutyEtAl03, pfeutyEtAl05} and
implicit~\citep{praamstraEtAl06} time perception. 
% Given the similarities
% between morphological features of the CNV and characteristics of the
% influencial pacemaker-accumulator (PA) model of interval timing
% \citep{treisman63,gibbon77,gibbon84}, it has been proposed that the CNV is
% the signature of a neural substrate of the PA model \citep{macarAndVidal09}.
% However, a recent study failed to observe expected predictions of the PA
% model in the CNV~\citep{kononowichAndVanRijn11}, and alternative hypothesis
% have emarged that interpreted the CNV as an index of temporal
% preparation~\citep{ngEtAl11} and/or decision~\citep{nononowiczAndVanRijn11}.
% A fundamental unresolved question is whether the CNV reflects expectation,
% motor preparation, or both types of functions. The close association between
% the CNV and properties of the motor response have suggested tha the E wave of
% the CNV is the readiness potential, a measure of activity in the motor cortex
% and supplementary mortor area leading to voluntary muscle
% movement~\citep{kornhuberAndDeecke65}. However, the CNV has been found on a
% recent oddball perception tasks requiring no motor response or decision
% making~\citep{mentoEtAl13}, highlighting the relevance of cognitive expectancy
% effects on the generation of the CNV.

% Early research
That the phase of brain oscillation plays an essential role for understanding
human behavior is not a recent finding.  Early EEG investigations reported that
the instant when a subject performs a voluntary hand movement coincides with a
specific phase of the alpha oscillation~\citep{bates51}.
However, in recent years we have witnessed exceptional new findings on the
relation between phase and behavior.
% Reaction times and perception Stimuli appearing at suitable phases of EEG
% oscillations trigger faster reaction times~\citep{}, and facilitate
% visual~\citep{} and auditory~\citep{} perception. 
%
% Entrainment
External visual~\citep{regan66} and auditory \citep{galambosEtAl81} rhythmic
stimuli entrain brain oscillations (i.e., drag brain oscillations to follow the
rhythm of the stimuli), and this entrainment is modulated by attention in such
a way that the occurrence of attended stimuli coincides with the phase of brain
oscillations of maximal excitability
\citep{lakatosEtAl05,lakatosEtAl08,lakatosEtAl13,oconnellEtAl11,besleEtAl11,gomezRamirezEtAl11,zionGolumbicEtAl13,cravoEtAl13,mathewsonEtAl09,mathewsonEtAl11,mathewsonEtAl12,spaakEtAl14,grayEtAl15}.
% Perception, motor control and cognition
Alignment of the phase of brain oscillations to optimize performance has been
observed in perception, motor control, and cognition. 
% Perception: visual, auditory, speech, cross-modal effects
In perception, this
alignment has been related to vision
\citep{valeraEtAl81,buschEtAl09,mathewsonEtAl09,mathewsonEtAl12,cravoEtAl13,hanslmayrEtAl13,spaakEtAl14,cravoEtAl15,grayEtAl15,miltonAndPleydellPearce16},
audition
\citep{riceAndHagstrom89,lakatosEtAl08,stefanicsEtAl10,henryAndObleser12,ngEtAl12,neulingEtAl12,lakatosEtAl13,hickokEtAl15}
, and speech
\citep{ahissarEtAl01,luoAndPoeppel07,howardAndPoeppel10,coganAndPoeppel11,howardAndPoeppel12,giraudAndPoeppel12,morillonEtAl12,grossEtAl13,peelleEtAl13,luoEtAl13,xiangEtAl13,doellingEtAl14,vanRullenEtAl14,parkEtAl15,zoefelAndVanRullen15a,zoefelAndVanRullen15b,zoefelAndVanRullen16}.
In addition, the alignment of phases of oscillations in the visual brain
regions can be triggered crossmodally by auditory stimuli, and vice versa
\citep{thorneEtAl11,kayserEtAl08,lakatosEtAl09,romeiEtAl12}. 
% Motor control and reaction speed
In motor control, optimal alignment of brain oscillations has been found in eye
movements
\citep{bosmanEtAl09,hammEtAl10,salehEtAl10,drewesAndVanRullen11,mclellandEtAl16}, and
cortico-spinal excitability \citep{vanElswijkEtAl10,bergerEtAl14}.
It has also been ascribed to reaction speed
\citep{lansing57,callawayAndYeager60,stefanicsEtAl10,drewesAndVanRullen11,thorneEtAl11}.
% Cognition
In cognition, the
alignment of brain oscillations has been linked to attention
\citep{yamagishiEtAl08,lakatosEtAl08,lakatosEtAl09,lakatosEtAl13,sausengEtAl08,capotostoEtAl09,vanRullenEtAl11,gomezRamirezEtAl11,zumerEtAl14,liuEtAl14,mathewsonEtAl14,grayEtAl15,vanDiepenEtAl15},
working memory
\citep{bonnefondAndJensen12,myersEtAl14}, causality judgment
\citep{cravoEtAl15}, stimuli coincidence \citep{miltonAndPleydellPearce16}, temporal predictions \citep{samahaEtAl15,fujiokaEtAl12}, 
and
to the transmission of prior information to visual cortex
\citep{shermanEtAl16}.
% ITC
Specially relevant for this manuscript are reports on effects of stimulus
expectancy on phase alignment
\citep{cravoEtAl11,cravoEtAl13,stefanicsEtAl10,zionGolumbicEtAl13,mathewsonEtAl12,miltonAndPleydellPearce16}. 
% TMS Relations between phase and perception are not only correlative, but
% causal, as demonstrated by studies using non-invasive electromagnetic brain
% stimulation
% \citep{capotostoEtAl09,zaehleEtAl10,dugueEtAl11,thutEtAl11,romeiEtAl12,neulingEtAl12,jaegleAndRo14,helfrichEtAl14}.
% ITC, inter-electrode phase coherence, cross-frequency phase coupling
The alignment of phase in groups of trials at one brain region, as studied in
this manuscript, has been related to the latency of saccadic eye movements
\citep{drewesAndVanRullen11}, visual contrast sensitivity \citep{cravoEtAl13},
visual perception \citep{buschEtAl09,mathewsonEtAl12,spaakEtAl14}, visual
attention \citep{grayEtAl15}, visual stimuli timing
\citep{miltonAndPleydellPearce16}, top-down visual attention
\citep{yamagishiEtAl08}, speech perception \citep{zionGolumbicEtAl13},
cross-modal interactions \citep{thorneEtAl11,fiebelkornEtAl11}, as well as to
temporal expectation \citep{cravoEtAl11,stefanicsEtAl10}. Also, the alignment
of phase across brain regions \citep{hanslmayrEtAl07}, the coupling of phase
across frequencies and brain regions \citep{fiebelkornEtAl13, sausengEtAl08,
handelAndHaarmeier09}, and the coupling between phase and amplitude
\citep{voytekEtAl10, bonnefondAndJensen15, cravoEtAl11} is predictive of some
aspects of behavior. 
% Infra-slow oscillations
% The previous studies investigated anticipatory phase at higher frequencies
% (\textgreater0.5 Hz), but the phase of infra-slow oscillations (0.01-0.1 Hz)
% during stimulus anticipation is also predictive of behavior \citep{montoEtAl08,
% vanhataloEtAl04}.
% Phase and information gating Reviews It has been hypothesized that the phase
% of occipital alpha oscillations serves to gate the flow of visual information
% from perceptual to higher-order brain regions~\citep{bonnefondAndJensen15,
% mathewsonEtAl11, zumerEtAl14} and that these oscillations are controlled by
% downstream brain regions~\citep{yamagishiEtAl08,samahaEtAl15,
% mathewsonEtAl14,bonnefondAndJensen15,liuEtAl14}. It has been proposed that
% perception is a discrete process, which is regulated by the phase of neural
% oscillations~\citep{valeraEtAl81,buschAndVanRullen10,fiebelkornEtAl11,chakravarthiAndVanRullen12,cravoEtAl15}.
% Reviews
Reviews of the relation between prestimulus phase, perception, behavior, and
cognition appear in \citet{schroederAndLakatos09, mathewsonEtAl11,
vanRullenEtAl11, thutEtAl12, vanRullenEtAl14}.
% In sum
In sum, almost every aspect of behavior, perception, and cognition critically
depends on the phase of brain oscillations.  

Considering that both the foreperiod duration and the phase of EEG oscillations
are related to anticipatory attention, we conjectured that the former could be
modulating the latter. In addition, because influences of attention on EEG
phase activity can be observed in single
trials \citep[e.g.,][]{buschAndVanRullen10}, we hypothesized that modulations
of phase activity by the foreperiod duration could also be seen in single
trials. 
In a previous investigation \citep{makeigEtAl02}, we showed that \gls{ITPC} can be a mechanism responsible for the appearance
of event-related potential features in trial-averaged measures of EEG activity.
Continuing with this line of investigation, here we study effects of
expectancy, as induced by the foreperiod duration, on the \gls{ITPC}. For this,
we developed a single-trial multivariate decoding method.  We reasoned that if
the \gls{SFP} modulates the \gls{ITPC} then we should be able to reliably
decode the \gls{SFP} duration (\gls{SFPD}) from the \gls{ITPC} evoked by a
standard stimulus.
In Section~\ref{sec:measuringSFPeITPC} we show that there exist a significant
correlation between the \gls{SFPD} and the \gls{ITPC} evoked by a standard in
a single trial.
In principle, this correlation could be an artifact of the decoding method. In
Section~\ref{sec:directEvidence} we show that this is not the case, since the
effect of the \gls{SFPD} on \gls{ITPC} can be observed by just comparing the
average \gls{ITPC} of trials closer and further away from the
\gls{warningSignal}.
It could be argued that the \gls{SFP} effect on \gls{ITPC} is a nuisance in EEG
recordings. We argue against this possibility in
Section~\ref{sec:correlationsWithBehavior}, by showing significant correlations
between the strength of the \gls{SFP} effect and behavioral measures (detection
errors and reaction speed). In Section~\ref{sec:timing} we study the timing of
the \gls{SFP} effect on \gls{ITPC}. We conclude with a discussion in
Section~\ref{sec:discussion}.

\section{Results}
\label{sec:results}

We characterized the foreperiod effect on \gls{ITPC} in the audio-visual
oddball detection task schematized in Figure~\ref{fig:experiment}. Subjects had
to detect visual or auditory deviants among visual or auditory standards.
Attention-shifting cues (LOOK and HEAR) were interspersed among the standards
and deviants. After a LOOK (HEAR) cue subjects had to detect visual (auditory)
deviants.  Thus, the LOOK (HEAR) cue directed attention to the visual
(auditory) \gls{attendedModality}. These cues initiated a period of expectancy
for the next deviant stimulus and and acted as a \gls{warningSignal}. Further
details on the experiment appear in Section~\ref{sec:experimentalInformation}. 

\begin{figure}
\begin{center}
\includegraphics{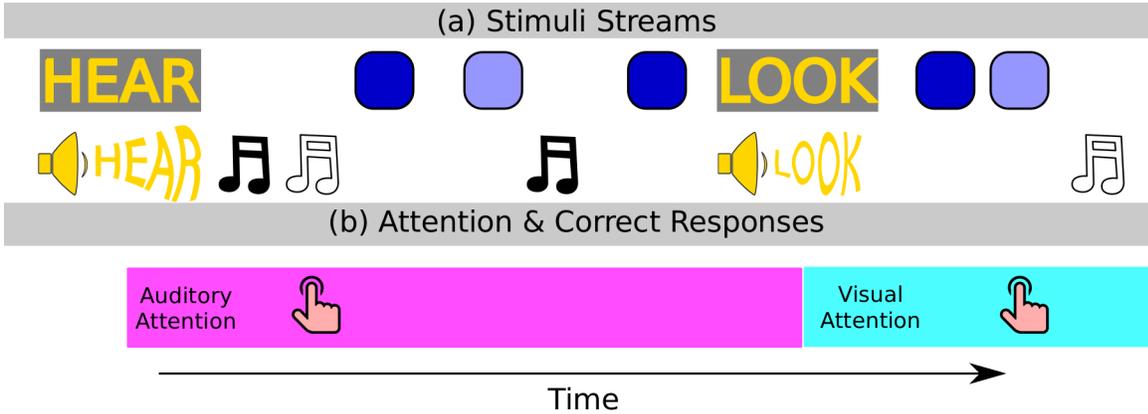}

\caption{Audio-visual attention-shift experiment. (a) Audio-visual LOOK and
HEAR attention-shifting cues, visual standards (dark-blue squares) and deviants
(light-blue squares), and auditory standards (low-pitch sounds, filled notes)
and deviants (high-pitch sounds, open notes) were presented sequentially in
pseudo-random order. (b) After seeing and hearing
a LOOK (HEAR) cue, subjects had to indicate with a button press when they saw
a visual (heard an auditory) deviant.}

\label{fig:experiment}
\end{center}
\end{figure}

All analysis were performed at the level of independent components, \glspl{IC},
obtained from an Independent component analysis~\citep[ICA,][]{makeigEtAl96} of
the recorded EEG data.  Figures~\ref{fig:exampleSingleTrialAnalysis}a
and~\ref{fig:exampleSingleTrialAnalysis}b illustrate the main finding of this
study. The thin curves are the cosine of the phase (at the peak \gls{ITC}
frequency, Section~\ref{sec:itcAndPeakITCFrequency}; 8.15~Hz in this figure) of
of the 20 epochs with the shortest
(Figure~\ref{fig:exampleSingleTrialAnalysis}a) and longest
(Figure~\ref{fig:exampleSingleTrialAnalysis}b) \gls{SFPD}.  These epochs were
aligned at time zero to the presentation of attended visual standards from
\gls{IC} 13 of subject av124a.  The thick red curve is the
cosine of the average phase across all epochs
(Section~\ref{sec:circularStats}).  Between 150 and 250~ms the phase of the ten
epochs with shortest \gls{SFPD} (thin lines in
Figure~\ref{fig:exampleSingleTrialAnalysis}a) are more similar to the average
phase (thick line in Figure~\ref{fig:exampleSingleTrialAnalysis}a or in
Figure~\ref{fig:exampleSingleTrialAnalysis}b) than the phase of the ten epochs
with longest \gls{SFPD} (thin lines in
Figure~\ref{fig:exampleSingleTrialAnalysis}b).  This suggests the existence of
a foreperiod effect on \gls{ITPC} where the \gls{SFPD} modulates the similarity
between the phase of a given epoch and the average phase across all epochs. The
rest of this paper studies the validity of this effect and investigates its
relation with subjects' behavior (detection errors and reaction times).

\subsection{Behavioral results}
\label{sec:behavioralResults}

We call deviant foreperiod duration to the delay between the presentation of a
\gls{warningSignal} and the next deviant
stimulus. We detected significant foreperiod effects on reaction times (i.e.;
significant correlations between deviant foreperiod durations and subjects'
reaction times to the corresponding deviant stimuli) in ten combinations of
subject and \gls{attendedModality} (26\% out of a total of 38 combinations of
19 subjects and two \glspl{attendedModality}) . Five of these combinations
corresponded to the visual \gls{attendedModality}. These correlations were all
positive for the visual \gls{attendedModality} (indicating that the longer the deviant foreperiod duration the longer
the subject reaction time), while for the
auditory \gls{attendedModality} there was a mixture of two positive and three
negative significant correlations.
Figure~\ref{fig:exampleForeperiodEffectsOnRTAndDetection}a plots deviant
foreperiod durations as a function of reaction times for an example subject
and \gls{attendedModality}. We found significant foreperiod effects on
detectability (i.e.; significant differences between the median deviant
foreperiod duration for hits and misses) in other 10 combinations of
subject and \gls{attendedModality} (26\%). All of these combinations corresponded to
the auditory \gls{attendedModality} and in all of them subjects detected more
easily later deviants (i.e.; the median deviant foreperiod was significantly
larger for hits than for misses).
Figure~\ref{fig:exampleForeperiodEffectsOnRTAndDetection}b plots deviant
foreperiod durations for hits and misses for an example subject and
\gls{attendedModality}.

Visual deviants were detected more reliably and faster than auditory
ones. 
Mean error rates were 0.22 and 0.09 for detecting auditory and visual stimuli,
respectively. A repeated-measures ANOVA, with error rate (in the detection of
deviants of the attended modality) as dependent variable and attended
modality as independent factor, showed a significant main effect of attended
modality (F(1, 17)=39.35, p\textless0.0001). A posthoc analysis revealed that
error rates were smaller when detecting visual than auditory deviants 
(z=6.45, p=5.59e-11).
Mean response times to auditory and visual stimuli were 407 and 377~ms,
respectively.  A repeated-measures ANOVA, with mean response time (to deviant
stimuli of the attended modality) as dependent variable and attended modality
as independent factor, showed a significant main effect of attended modality
(F(1, 18)=40.29, p\textless0.0001). A posthoc analysis revealed that mean
response time was shorter for visual than auditory deviants (z=6.52,
p=3.49e-11).

\subsection{Measuring a new foreperiod effect on the ITPC elicited by 
standards}
\label{sec:measuringSFPeITPC}

Since we were interested in studying expectancy effects unrelated to motor
preparation, we only characterized the foreperiod effect on \gls{ITPC} for
standards, to which subjects were instructed no to respond.
Figure~\ref{fig:exampleSingleTrialAnalysis}c shows phase values, at the peak
\gls{ITC} frequency (Section~\ref{sec:itcAndPeakITCFrequency}), in single
trials from \gls{IC} 13 of subject av124a, and for epochs aligned to the
presentation of unattended visual standards. Trials are sorted from bottom to
top by increasing \gls{SFPD}. The solid line represents the presentation time
of the \gls{warningSignal}, scaled so that the maximum \gls{SFPD} fits in the
negative time period (the maximum, median, and minimum \gls{SFPD} in this panel
were 2,424~ms, 1,160~ms, and 328~ms, respectively). The top and bottom insets
show magnified views of the phase of the trials with longest and shortest
\glspl{SFPD}, respectively. Panels (a) and (b) show the 20 bottommost and
topmost trials from panel (c). Apparently, standards with
shorter \gls{SFPD} elicit stronger phase coherence from 150 to 250~ms after the
presentation of the \gls{warningSignal} than standards with longer
\gls{SFPD}.  Below we use decoding models to assess the validity of this
observation.

% \begin{figure}
\begin{center}
\includegraphics[width=5in]{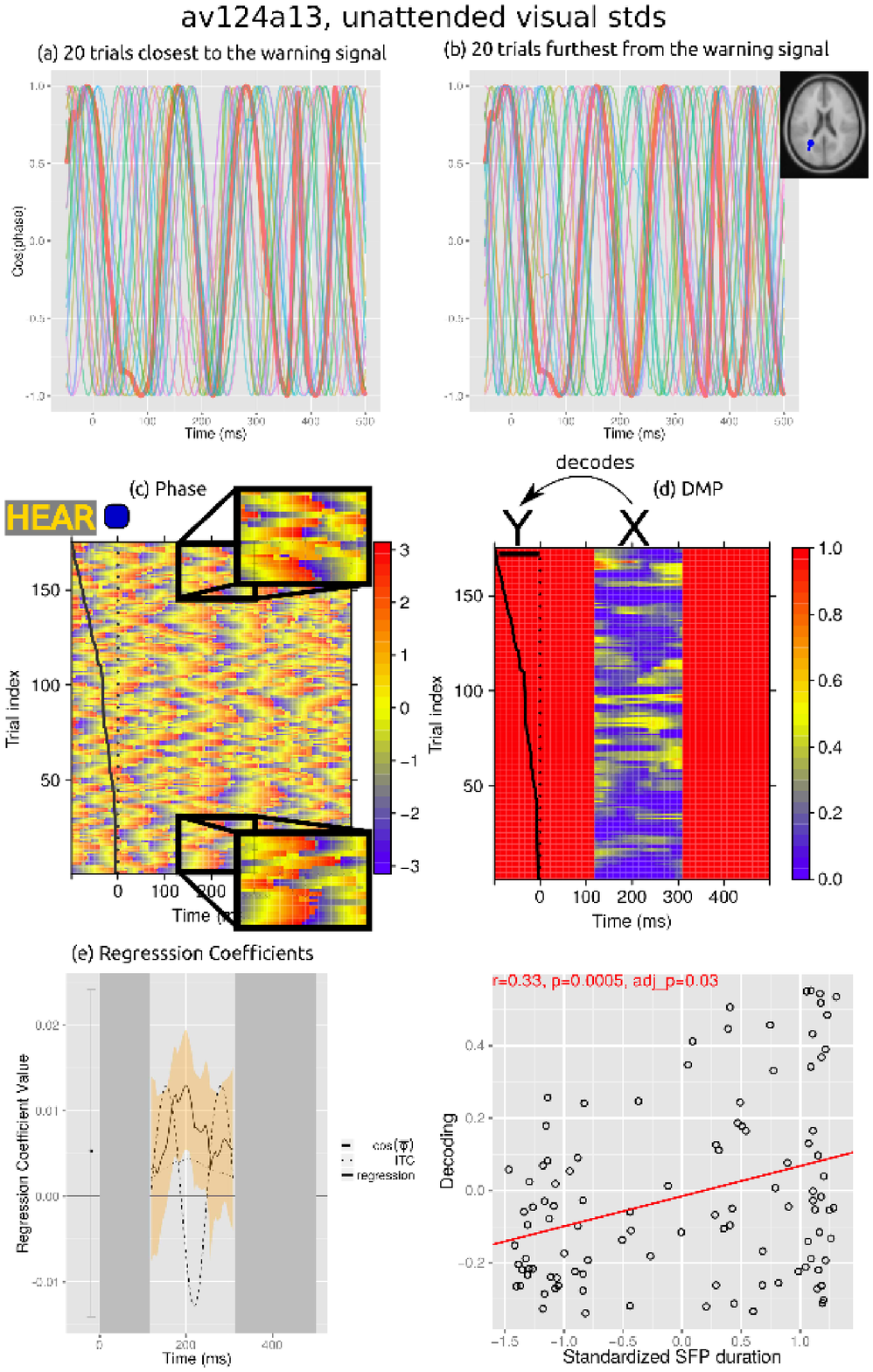}

\captionof{figure}{Data used for, and parameters obtained from, the analysis of
trials from \gls{IC} 13 of subject av124a, for epochs aligned to the
presentation of unattended visual standards.  
(a, b) Cosine of the phase of the twenty epochs with the shortest (a) and the
longest (b) \gls{SFPD}. Thin curves are the cosines of the phases of individual
epochs, and the thick curve is the cosine of the mean phase (i.e., mean
direction, Eq.~\ref{eq:meanDirection}) across all epochs. Between 150 and
250~ms, the phases of epochs with shortest \glspl{SFPD} in (a) are more similar
to the mean phase than the phases of epochs with longest \glspl{SFPD} in (b).
(c) Phase erpimage. Each horizontal line shows the phase of a single trial. For
visualization clarity, the plot has been smoothed across trials using a running
median with a window of size three.  Time zero (dotted vertical line)
corresponds to the presentation of standard visual stimuli. The solid black
line indicates the (scaled) times of presentation of the \glspl{warningSignal}
immediately preceding the standards at time zero. Trials are sorted from bottom
to top by increasing \gls{SFPD}. Insets magnify phases closest (bottom) and
furthest (top) from the \gls{warningSignal}.
(d) \gls{DMP} erpimage obtained from the phases in (b), smoothed with a
running median across trials of size 5.  Solid red regions correspond to time
points where the distribution of phases across trials was not significantly
different from the uniform distribution (p\textgreater0.01, Rayleigh test).
(e) The solid line plots the coefficients of the linear regression model used
to decode the \gls{SFPD} (y in panel (b)) from the corresponding \gls{DMP}
values (x in panel (b)) of single trials. The orange band gives 95\% bootstrap
confidence intervals of the estimated coefficients
(Section~\ref{sec:additionalStatInfo}). The dashed line at time $t$ plots the
mean phase at time $t$ across all trials. The dotted line is the \gls{ITC}
curve (Section~\ref{sec:itcAndPeakITCFrequency}).
(f) Scatter plot of model decodings versus standardized \glspl{SFPD}.}

\label{fig:exampleSingleTrialAnalysis}

\end{center}
% \end{figure}

To study if the \gls{SFPD} modulates the \gls{ITPC} evoked by standards, we
attempted to decode the former from the latter.  Reliable decodings would
indicate that the \gls{SFPD} is correlated with the \gls{ITPC} evoked by
standards.  We used the Deviation from the Mean Phase (\gls{DMP}) to quantify
\gls{ITPC} in single trials (Section~\ref{sec:dmp}).
Figure~\ref{fig:exampleSingleTrialAnalysis}d shows \gls{DMP} values
computed from the phase values in Figure~\ref{fig:exampleSingleTrialAnalysis}c.
Time points at which there was not enough statistical evidence to conclude that
the phase distribution across different trials was statistically different from
the uniform distribution (p\textgreater 0.01, Rayleigh test) are masked in red.
We did not use these points to decode the \gls{SFPD}, for a reason explained in
Section~\ref{sec:dmp}.
For each epoch, we decoded the \gls{SFPD} of the standard presented at time
zero from a 500 ms-long window of \gls{DMP} values following the presentation
of this standard.  To avoid possible EEG movement artifacts induced by
responses to deviants, we excluded from the analysis trials with deviant
stimuli in this 500-ms-long window.  We used a linear regression model for
decoding (Section~\ref{sec:linearRegressionModel}), and estimated its
parameters using a variational-Bayes method
(Section~\ref{sec:lmEstimationMethod}). To each estimated model corresponds a
\gls{standardModality} (i.e., the modality of the standards used to align the
epochs from which the model parameters were estimated; visual
\gls{standardModality} in
Figure~\ref{sec:robustCorrelationCoefficientsAndPValues}) and
an \gls{attendedModality} (i.e., the modality that was attended when
the epochs used to estimate the model parameters were recorded; auditory
\gls{attendedModality} in
Figure~\ref{sec:robustCorrelationCoefficientsAndPValues}).

If the delay between the \gls{warningSignal} and the following deviant
is longer than a threshold, that depends on the distribution of deviants,
the foreperiod effect on reaction times
disappears~\citep{botwinickAndBrinley62}. 
Similarly, if we estimate decoding models including a large proportion of
standards presented long after the \gls{warningSignal}, decodings become
non-significant (i.e., the \gls{SFP} effect on \gls{ITPC} disappears).
To fit decoding models we used data from standards that were presented
before a maximum \gls{SFP} duration after the \gls{warningSignal}. The
selection of this maximum is described in Section~\ref{sec:selectionMaxSFPD}.

The coefficients of the linear regression model estimated from the \glspl{SFPD}
and \gls{DMP} values in Figure~\ref{fig:exampleSingleTrialAnalysis}d are shown
in Figure~\ref{fig:exampleSingleTrialAnalysis}e.  The leave-one-cross-validated
decodings from the linear regression model are plotted as a function of
\gls{SFPD} in Figure~\ref{fig:exampleSingleTrialAnalysis}f. To assess the
accuracy of these decodings we calculated their robust correlation coefficient
with the \glspl{SFPD} (Section~\ref{sec:robustCorrelationCoefficientsAndPValues}).  We
obtained a correlation coefficient r=0.33, which was significantly different
from zero (p=0.0005, permutation test for skipped Pearson correlation
coefficient; adjusted p-value for multiple comparisons adj\_p=0.03,
Section~\ref{sec:multipleComparisons}).  This significance, and the fact that
the estimated regression model was significantly different from the
intercept-only model (p\textless0.01, likelihood-ratio permutation test,
Section~\ref{sec:additionalStatInfo}), suggests that the \gls{SFPD} is
significantly correlated with the \gls{ITPC} activity evoked by unattended
visual standards in \gls{IC} 13 of subject av124a.

To characterize the \gls{SFP} effect on \gls{ITPC} across subjects we grouped
all \glspl{IC} of all subjects using a clustering algorithm
(Section~\ref{sec:clusteringOfICs}).  Figure~\ref{fig:clusters} shows the
obtained clusters and Table~\ref{table:clustersInfo} provides additional
information about them. To quantify the strength of the standard foreperiod
effect on \gls{ITPC} at a given cluster, \gls{standardModality}, and
\gls{attendedModality}, we computed the proportion of models with decodings
significantly correlated with experimental \glspl{SFPD}, as in
Figure~\ref{fig:exampleSingleTrialAnalysis}f. Dots below the image of a cluster
in Figure~\ref{fig:clusters} indicate that the decodings of more than 40\% of
the models, estimated from phase activity in that cluster and from the
\gls{standardModality} and \gls{attendedModality} given by the color of the
dot, were significantly correlated with experimental \glspl{SFPD}.
For each cluster, \gls{standardModality}, and \gls{attendedModality}, the
number of models with significant correlations with the \glspl{SFPD}, as well
as the proportion of such models over the total number of models estimated for
the cluster, \gls{standardModality}, and \gls{attendedModality}, are given in
Table~\ref{table:nModelsSignCorPredictionsVsSFPDs}.
The larger number of clusters with more than 40\% significant correlations for
the visual than the auditory modality (red vs.\ blue points) and for attended
than unattended standards (darker vs.\ lighter colors) hints that the \gls{SFP}
effect on \gls{ITPC} is stronger for the visual than the auditory modality, and
for attended than unattended standards.  For the visual \gls{standardModality},
clusters where the decoding of more than 40\% of models were significantly
correlated with \glspl{SFPD} are localized in parieto-occipital regions, while
for the auditory \gls{standardModality} these clusters are found on
fronto-central brain regions, suggesting that the \gls{SFP} effect on
\gls{ITPC} is somatotopically organized.

From the successful decodings of \glspl{SFPD} from \gls{ITPC} activity triggered
by standards, we infer that the \gls{SFPD} is modulating the \gls{ITPC}.
However, other events, such at the \gls{warningSignal}, could be modulating
this \gls{ITPC}. We argue against this possibility in
Section~\ref{sec:controls}.

\begin{figure}
\begin{center}
\includegraphics{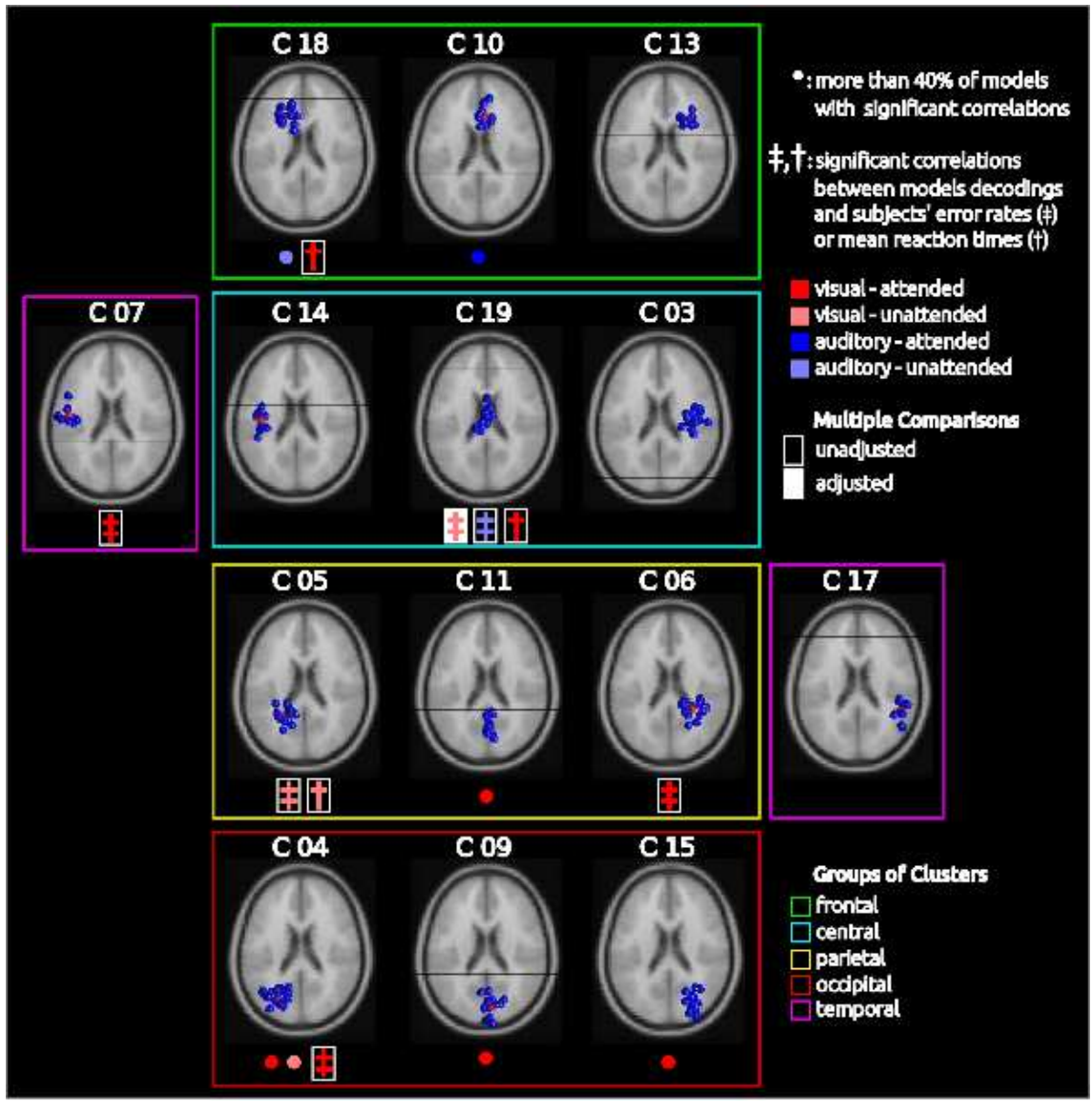}

\caption{Clusters of \glspl{IC}, large proportions of significant models,
correlations with behavior, and groups of clusters.
A blue ball inside a brain slice represents an \gls{IC} from one subject. Each
brain slice displays one cluster of \glspl{IC}.
A colored dot below the image of a cluster indicates that in more than 40\% of
the models estimated from data from that cluster, and from the
\gls{standardModality} and the \gls{attendedModality} given by the color of the
dot, the correlation coefficient between models' decoding and \glspl{SFPD} was
significantly different from zero (p\_adj\textless0.05,
Section~\ref{sec:multipleComparisons}).
A dagger (double dagger) signals a significant correlation between models'
decoding power and subjects' mean reaction times (error rates).  Filled
(unfilled) rectangles behind daggers indicate that the significance of the
corresponding correlation test was corrected (uncorrected) for multiple
comparisons.
Colored boxes mark groups of clusters used in Section~\ref{sec:timing} to study
the timing of the \gls{SFP} effect on \gls{ITPC}.
The clusters with a large proportion of significant models suggest
that the \gls{SFP} effect on \gls{ITPC} is stronger for visual than auditory
standards and for attended than unattended standards, and that the effect is
somatotopically organized (see text). The significant correlations between the
strength of the \gls{SFP} effect on \gls{ITPC} and subjects' behavioral
performance show that the effect is behaviorally relevant.
}

\label{fig:clusters}
\end{center}
\end{figure}

\subsection{Direct evidence for the SFP effect on ITPC}
\label{sec:directEvidence}

In the previous section we used linear regression models to decode \glspl{SFPD}
from \gls{ITPC} activity following the presentation of standards.  Since these
decodings were reliable, we inferred that the \gls{SFPD} is correlated with the
\gls{ITPC} activity triggered by standards. Here we use a simple method to gain
more direct evidence of the existence of this correlation.  

We compared the averaged phase decoherence of trials with the shortest and the
longest \glspl{SFPD}.  The existence of a significant difference in these
values would be direct evidence for the \gls{SFP} effect on \gls{ITPC}.  To
measure phase decoherence in a group of trials, we computed the mean \gls{DMP}
of the trials in the group.  Section~\ref{sec:proofAveragedDMPApproachesITC}
proves that, at times of high phase coherence, the mean \gls{DMP} is related to
the \gls{ITC}, a standard measure of averaged \gls{ITPC}. The red and blue
lines in the top row of Figure~\ref{fig:avgDMPsAndCoefs} plot the mean
\gls{DMP} of the 20\% trials furthest from and closest to the
\gls{warningSignal}, respectively.  Panels in the middle row plot the
difference between the red and blue lines in the corresponding panels of
the top row. Plots in different columns correspond to 
phase activity extracted from different subjects, \glspl{IC} in the left
parieto-occipital cluster 04, and attended visual stimuli.
For subject av124a and \gls{IC} 07 (left column), the mean \gls{DMP} was
significantly larger for trials with the longest \gls{SFPD} between 0 and
100~ms after the presentation of standard stimuli, but smaller between 200 and
320~ms (Figures~\ref{fig:avgDMPsAndCoefs}a and~\ref{fig:avgDMPsAndCoefs}d).
Differently, for subject av115a and \gls{IC} 4 (central column), the mean phase
decoherence was significantly larger for trials with the longest \gls{SFPD}
between 200 and 400~ms (Figures~\ref{fig:avgDMPsAndCoefs}b
and~\ref{fig:avgDMPsAndCoefs}e).
And for subject av113a and \gls{IC} 17 (right column), the mean phase
decoherence was significantly larger for trials with the longest \gls{SFPD}
between 0 and 250~ms (Figures~\ref{fig:avgDMPsAndCoefs}c
and~\ref{fig:avgDMPsAndCoefs}f).

The linear regression method is statistically more powerful than the the simple
method for testing the association between the \gls{SFPD} and the \gls{ITPC}
activity triggered by standards, since it uses all trials in a statistically
optimal way, and not just extreme trials in and ad-hoc manner.  Another
advantage of the linear regression method, compared with more general function
approximation methods, such as artificial neural networks, is that the
regression coefficients are readily interpretable.  A positive regression
coefficient indicates that trials with longer \gls{SFPD} correspond to larger
values of \gls{DMP}, or larger decoherence, at the time of the regression
coefficient, compared to trials with shorter \glspl{SFPD}. This is because to
decode a long \gls{SFPD} a positive regression coefficient needs to be
multiplied by a large \gls{DMP} value (assuming that the reminding regression
coefficients are kept fixed).  Similarly, a negative regression coefficient
implies that trials with longer \gls{SFPD} correspond to smaller phase
decoherence, at the time of the regression coefficient.
To validate this interpretation, and the soundness of the decoding methodology,
the bottom row in Figure~\ref{fig:avgDMPsAndCoefs} plots the regression
coefficients of the decoding models corresponding to the top rows.
Figure~\ref{fig:avgDMPsAndCoefs}g shows significantly positive regression
coefficients between 100 and 200~ms and significantly negative coefficients
between 240 and 300~ms. Therefore, according to the previous interpretation,
trials with longer \glspl{SFPD} should correspond to more decoherent activity
between 100-200~ms and to less decoherent activity between 240 and 300~ms. And
this is what we observed in Figures~\ref{fig:avgDMPsAndCoefs}a
and~\ref{fig:avgDMPsAndCoefs}d. Similar consistencies hold for the other
columns.  The normalized crosscorrelation
(Section~\ref{sec:additionalStatInfo}) between the difference in averaged
\gls{DMP} (middle row in Figure~\ref{fig:avgDMPsAndCoefs}) and the
corresponding regression coefficients (bottom row of
Figure~\ref{fig:avgDMPsAndCoefs}) was 0.80 for (d)-(g), 0.39 for (e)-(h), and
0.38 for (f)-(i).  Across all models significantly different from the
intercept-only model, the first, second (median), and third quartiles of the
normalized crosscorrelation distribution were 0.40, 0.71 and 0.87,
respectively.  These results indicate that on average regression coefficients
were similar to differences in average \gls{DMP}, validate the previous
interpretation of regression coefficients, and support the inference that
reliable decodings by the regression models indicate modulations of \gls{ITPC}
by the \gls{SFPD}.

\begin{figure}
\begin{center}
\includegraphics[width=5.5in]{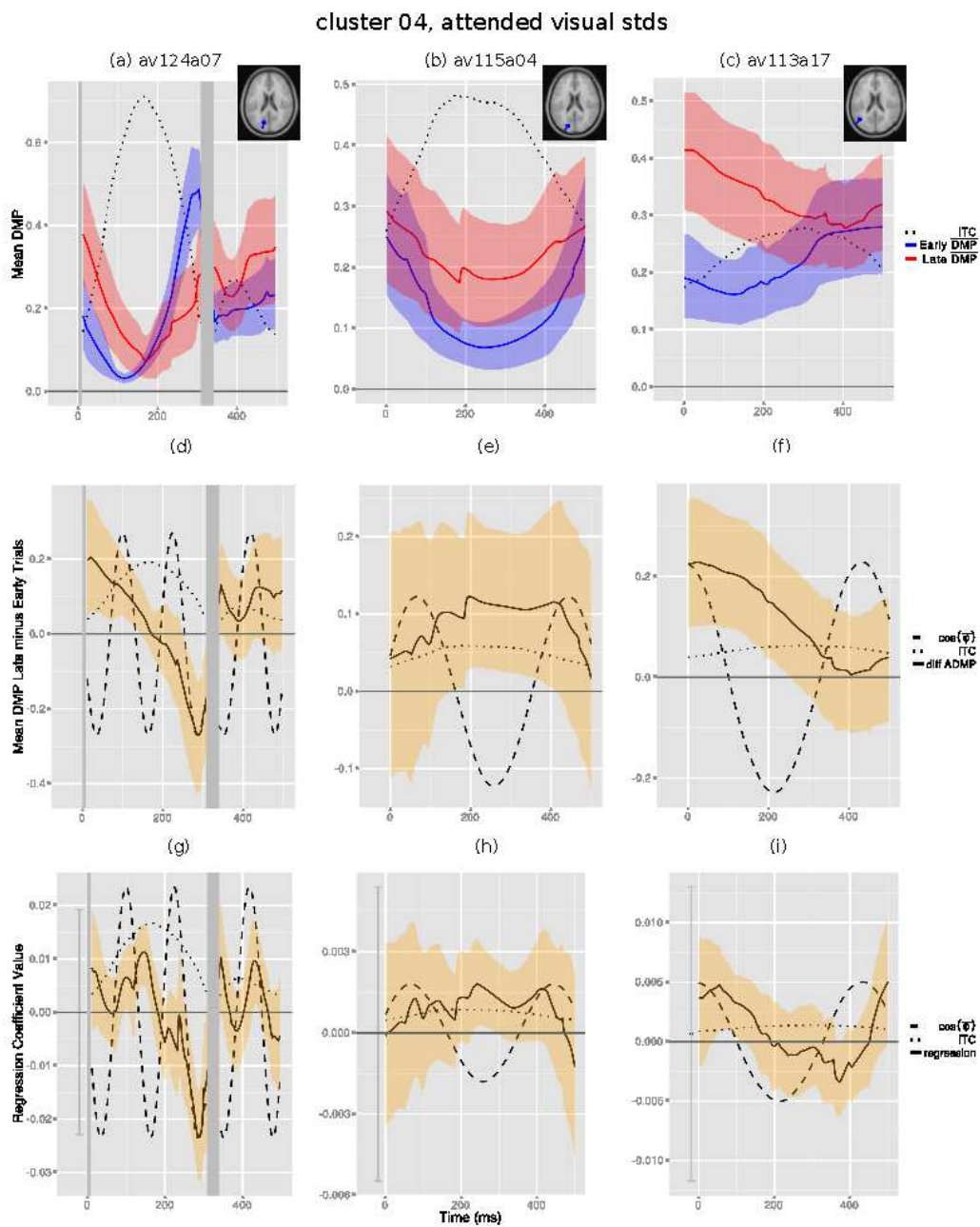}

\caption{Relation between differences in phase coherence in trials furthest
from and closest to the \gls{warningSignal} and regression models'
coefficients.
Top row: averaged \gls{DMP} evoked by the 20\% standards closest to (blue
curve) and furthest from (red curve) the \gls{warningSignal}. The dotted line
plots the \gls{ITC} computed from all trials.
Middle row: the solid line plots the difference in averaged \gls{DMP} evoked by
trials furthest from minus closest to the \gls{warningSignal} (i.e., the red
minus the blue curves in the top row).
Bottom row: standardized coefficients of regression models.
Colored bands in all panels represent 95\% confidence intervals. The dotted and
dashed lines in the panels of the middle and bottom rows plot the arbitrarily
scaled \gls{ITC} and the cosine of the mean phase from all trials,
respectively. The top and middle rows demonstrate that \gls{ITPC} depends on
the \gls{SFPD}, and the similarity between the middle and bottom rows show that
the coefficients of a regression model indicate whether standards further from
the \gls{warningSignal} evoke more coherence or decoherent activity than
standards closer to it.}

\label{fig:avgDMPsAndCoefs}
\end{center}
\end{figure}

\subsection{The SFP effect on ITPC is correlated with subjects' behavior}
\label{sec:correlationsWithBehavior}

In the previous sections we presented evidence suggesting that the \gls{SFPD}
influences the \gls{ITPC} evoked by standards. Here we show
that this is a relevant effect, since its strength is correlated with subjects'
error rates and mean reaction times.

The ordinate in Figure~\ref{fig:correlationsWithBehavior}a gives the
correlation coefficient between decodings from a model and experimental
\glspl{SFPD}, for all subjects with an \gls{IC} in the central-midline cluster
19 and unattended visual standards.  The abscissa provides subjects' error
rates. The significant negative correlation between models' decoding accuracies
and subjects' error rates (r=-0.91, p=0.0004, adj\_p=0.05) shows that the
higher is the association between \gls{SFPD} and \gls{ITPC} in a subject, the
better is his or her behavioral performance (error rate).
Figure~\ref{fig:correlationsWithBehavior}b is as
Figure~\ref{fig:correlationsWithBehavior}a, but for attended visual standards
and the left parieto-occipital cluster 04. We also found significant
correlations between the decoding accuracy of models and subjects mean
reaction times, as illustrated in Figures~\ref{fig:correlationsWithBehavior}c
and~\ref{fig:correlationsWithBehavior}d. Points marked in green in these figures
represent outliers detected in the calculation of robust correlation
coefficients (Section~\ref{sec:robustCorrelationCoefficientsAndPValues}). Note that the
decoding models were optimized to decode \glspl{SFPD}. These significant
correlations 
resulted without fitting the decoding models to behavioral data.

\begin{figure}
\begin{center}
\includegraphics[width=5.5in]{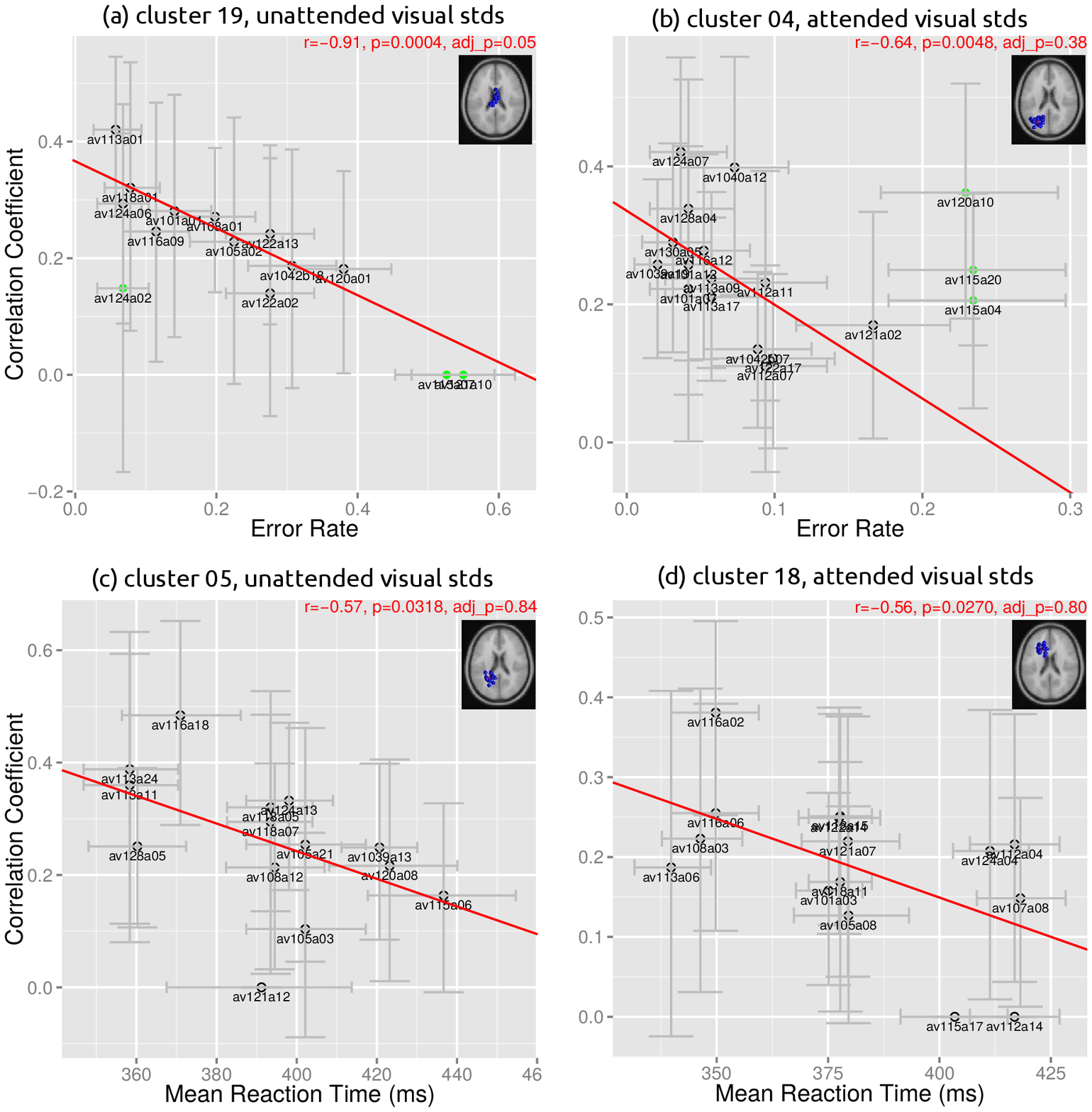}

\caption{The strength of \gls{SFP} effect on \gls{ITPC} (i.e., the accuracy of
the decoding models) is related to subjects' behavior (error rates in (a,b) and
mean reaction times in (c,d)), attesting the behavioral relevance of the
\gls{SFP} effect on \gls{ITPC}. Larger decoding accuracy (i.e., larger
correlation coefficients between models decodings and \glspl{SFPD}; ordinate in
all panels) corresponds to lower error rates (abscissa in (a,b)), as indicated
by negative correlation coefficients, $r$, in (a,b), and this correspondence is
stronger in models estimated from phase activity triggered by unattended
standards (e.g., $|r|$ is larger in (b) than in (a)).  Larger decoding accuracy
seems to correspond to faster (slower) mean reaction times (abscissa in (c,d))
in models estimated from phase activity triggered by attended (unattended)
standards (e.g., $r$ in in (c) and (d) is negative and positive, respectively),
but this correspondence is not as conclusive as for error rates (see text).}

\label{fig:correlationsWithBehavior}
\end{center}
\end{figure}

A colored double and simple dagger below the image of a cluster in
Figure~\ref{fig:clusters} indicates a significant correlation between the
decoding accuracy of models estimated from phase activity in the cluster and
subjects error rates and mean reaction times, respectively. The color of the
dagger corresponds to the \gls{standardModality} and the \gls{attendedModality}
of the data used to estimate the models' parameters.
Table~\ref{table:statsErrorRates} shows the correlation coefficients, and
corresponding p-values, between the decoding accuracy of models and subjects error
rates.  Table~\ref{table:statsMeanRTs} is as Table~\ref{table:statsErrorRates},
but for mean reaction times.
All significant correlations between model's decoding accuracies and subject's
error rates were negative (Table~\ref{table:statsErrorRates}), indicating that the larger the detection performance of a
subject, the larger the strength of the \gls{SFP} effect on \gls{ITPC}.  The
absolute value of these correlations was larger for attended than unattended
standards (p\textless 1e-04; permutation test).
For mean reaction times, the three significant correlations were also negative
(Table~\ref{table:statsMeanRTs}),
showing that subjects with strongest \gls{SFP} effect on \gls{ITPC} reacted the
slowest. However, differently from error rates, positive correlation for mean
reaction times almost reach significance (e.g., cluster 06 and unattended
auditory standards, or cluster 13 and unattended visual standards,
Table~\ref{table:statsMeanRTs}).
Correlations (both significant and non-significant) were stronger for error
rates than for mean reaction times and stronger for the visual than the
auditory \gls{standardModality}.  An ANOVA with the absolute value of the
correlation coefficient  as dependent variable showed significant main effects
of behavioral measure type (i.e., error rate or mean reaction time; F(1,
109)=7.26, p=0.0082) and for \gls{standardModality} (F(1, 109)=8.22, p=0.005).
A posthoc analysis indicated that the mean absolute value of the correlation
coefficient was larger for error rates than for mean reaction times (p=0.005;
Tukey test) and larger for the visual than the auditory \gls{standardModality}
(p=0.0082; Tukey test).  Further information on this ANOVA appears in
Section~\ref{sec:anovaFurtherInfo}.

The above correlations with error rates, although significantly different from
zero at the 0.05 confidence level for univariate hypothesis tests, were found
after computing a total of 56 correlations (14 clusters x 2
\glspl{standardModality} x 2 \glspl{attendedModality}).  Due to the
multiple-comparison problem~\citep{westfallAndYoung93} the probability that at
least one of these correlations was found significant while there was no
association between subjects' error rates and models' decoding accuracies is
very large ($p\simeq0.94$). Only the correlation for the central cluster 19 and
unattended visual standards (Figure~\ref{fig:correlationsWithBehavior}a)
remained significant after adjusting for multiple comparisons using the
resampling procedure described in Section~\ref{sec:multipleComparisons}. It is
thus highly probable that some of the five unadjusted significant correlations
with error rates that did not survive the multiple comparison test occurred by
chance.  However, note that, as expected, all of these correlations were
negative, which is a very rare event under the assumption that all of these
correlations occurred by chance (i.e., the probability of finding five negative
correlations assuming equal probability for positive and negative correlations
is $p=0.5^5=0.03$). Thus, not all of the five unadjusted significant
correlations that did not survive the multiple comparison adjustment may have
occurred by chance.
None of the correlations with mean reaction times survived the multiple
comparisons adjustment.  Nonetheless, since the brain regions were these
unadjusted significant correlations occurred have been previously connected to
foreperiod effect on reaction times, as we discuss below, it is possible that
not all these correlations occurred by chance.

\subsection{Timing of the SFP effect on ITPC}
\label{sec:timing}

In Section~\ref{sec:directEvidence} we advanced an interpretation of the coefficients of
the decoding model, and used a simple trial averaging procedure to verify that
this interpretation was sound. As shown in
Figures~\ref{fig:exampleSingleTrialAnalysis}e and
\ref{fig:avgDMPsAndCoefs}d-f, modulations of \gls{ITPC} by the \gls{SFPD},
as represented by the coefficients of decoding models, are not constant in
time, but fluctuate in an oscillatory manner. 
In the 500~ms window following the presentation of a standard, these
coefficients displayed one or more peaks, as in
Figures~\ref{fig:exampleSingleTrialAnalysis}e
and~\ref{fig:avgDMPsAndCoefs}g-h, respectively. Details on the identification
of these peaks are given in Section~\ref{sec:identificationOfCoefsPeaks}. 
The time of the largest peak corresponds to the latency after the presentation
of standards where modulations of \gls{ITPC} by the \gls{SFPD} are strongest.
In this section we use this time to represent the timing of the \gls{SFP}
effect on \gls{ITPC}, and study how this timing varies across brain regions,
\glspl{standardModality}, and \glspl{attendedModality}.

The median time of the largest peak of the decoding model coefficients was
292~ms, with a 95\% confidence interval of [284, 300]~ms. We grouped the
clusters of \glspl{IC} into five groups, as indicated in
Figure~\ref{fig:clusters}, and examined the mean time of the largest peak for
each group of clusters and \gls{attendedModality} (Figure~\ref{fig:timing}a),
and for each group of clusters and and \gls{standardModality}
(Figure~\ref{fig:timing}b).
We found that peaks occurred earlier when attention was oriented to the auditory
than to the visual modality (Figure~\ref{fig:timing}a).
This modulation by attention was strongest at the central group of clusters,
where the peak of the models' coefficients occurred earlier for auditory than
visual standards (Figure~\ref{fig:timing}b).  Therefore, the timing of the
\gls{SFP} effect on \gls{ITPC} was most strongly modulated, by the
\gls{attendedModality} and by the \gls{standardModality}, at central brain
regions.
In addition, for visual standards, the \gls{SFP} effect on \gls{ITPC} occurred
earlier at occipital than central brain regions (Figure~\ref{fig:timing}b).

A first ANOVA, using data from all models significantly different from the
intercept only model (p\textless 0.01; likelihood-ratio permutation test,
Section~\ref{sec:additionalStatInfo}), with time of the largest peak of the
coefficients as dependent variable, found a significant main effect of
\gls{attendedModality} (F(1,230)=4.259, p=0.0402). A
posthoc test showed that the peak occurred
earlier for the auditory than the visual \glspl{attendedModality} (z=1842,
p=0.0327; black asterisk to the left of Figure~\ref{fig:timing}a).
A second ANOVA restricted to models corresponding to the visual
\gls{standardModality} found a significant main effects of group of clusters
(F(4,135)=4.4073). A posthoc analysis revealed that the peak was earlier for the
occipital than the central group of clusters (z=4.163, p=3.14e-05; red
asterisks in Figure~\ref{fig:timing}b).
A third ANOVA restricted to models corresponding to the central group of
clusters found significant main effects of \gls{attendedModality}
(F(1,43)=4.5061, p=0.0396) and of \gls{standardModality} (F(1,43)=14.2173,
p=0.0005). A posthoc analysis showed that the peak was earlier
when attention was oriented to the auditory than visual modality (z=2.178,
p=0.029360; black asterisk next to \emph{Central} in Figure~\ref{fig:timing}a)
and earlier for auditory than visual standards (z=3.406, p=0.000659; black
asterisks in Figure~\ref{fig:timing}b).

\begin{figure}
\begin{center}
\includegraphics[width=5.5in]{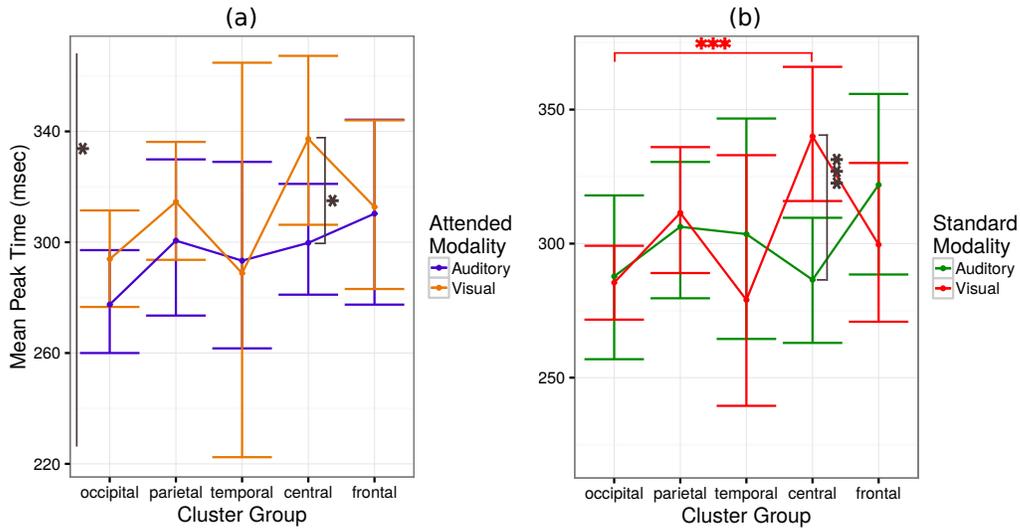}

\caption{Timing of the \gls{SFP} effect on \gls{ITPC}.
(a) Mean coefficients peak times for models corresponding to the auditory (blue
                                                                            points)
and visual (orange points) \gls{attendedModality}, as a function of group of
clusters. On average, when attending to the audition, the \gls{SFP} effect on
\gls{ITPC} occurred earlier than when attending to vision. 
(b) Mean coefficients peak times for models corresponding to the auditory
(green points) and visual (red points) \gls{standardModality}, as a function of
group of clusters. On average, for visual standards, the \gls{SFP} effect on
\gls{ITPC} occurred earlier in the occipital than in central brain regions. 
The effect was most strongly modulated over the central brain region by both
the \gls{attendedModality} (panel (a)) and by the \gls{standardModality} (panel
                                                                           (b)).
}

\label{fig:timing}
\end{center}
\end{figure}

\section{Discussion}
\label{sec:discussion}

Here we reported a new foreperiod effect on the \gls{ITPC} activity triggered
by standards. We demonstrated that, when visual or auditory standards are
preceded by a warning signal (Figure~\ref{fig:experiment}), \gls{ITPC} is
modulated by the delay between the warning signal and the standards
(Figures~\ref{fig:exampleSingleTrialAnalysis}, clusters with a colored dot in
Figure~\ref{fig:clusters}, and blue entries in
Table~\ref{table:nModelsSignCorPredictionsVsSFPDs}). We used a decoding method to detect the
new foreperiod effect. We demonstrated that this effect is not an artifact of
the decoding method, since it can be observed using simple trial averages
(Figure~\ref{fig:avgDMPsAndCoefs}). Importantly, the strength of the new
foreperiod effect (i.e.; the accuracy in decoding the \gls{SFPD} from the phase
activity elicited by a standard) is correlated with subjects' detection
performance and reaction speed (example and summary of correlations with
detection performance in Figures~\ref{fig:correlationsWithBehavior}a
and~\ref{fig:correlationsWithBehavior}b, clusters with a double dagger in
Figure~\ref{fig:clusters}, and blue entries in
Table~\ref{table:statsErrorRates}; with reaction speed in
Figures~\ref{fig:correlationsWithBehavior}c
and~\ref{fig:correlationsWithBehavior}d, clusters with a dagger in
Figure~\ref{fig:clusters}, and blue entries in
Table~\ref{table:statsMeanRTs}). The effect occurred earlier when attention
was oriented to the auditory modality (Figure~\ref{fig:timing}a), for visual
standards it happened earlier at occipital than in central brain regions
(Figure~\ref{fig:timing}b), and modulations of the timing of the effect by the
\gls{attendedModality} or by the \gls{standardModality} was stronger over a
central brain regions (Figures~\ref{fig:timing}a and~\ref{fig:timing}b).

% Some references from delormeEtAl15-erpImages.pdf
% delormeEtAl07,ontonAndMakeig09

Phase coherence has been used to describe the synchrony between different
recording electrodes~\citep[e.g.,][]{lachauxEtAl99, hanslmayrEtAl07} or, as in
this study, to characterize the alignment of phases across multiple trials at
a single site \citep[e.g.,][]{buschEtAl09,thorneEtAl11,cravoEtAl15}. Previous measures of
phase alignment at a single site were the result of averaging data across
multiple trials. However, phase coherence may vary from trial to trial, and
important information could be lost in trial averages. We have previously
proposed the ERP image plot, a useful visualization of trial-by-trial
consistencies in a set of event-related data
trials~\citep{makeigEtAl04,delormeEtAl15}. To build this plot, trials are first
sorted in order of a relevant data or external variable and then plotted as a
color-coded two-dimensional image, as in
Figure~\ref{fig:exampleSingleTrialAnalysis}c. ERP image plots have proven
useful in revealing trial-to-trial consistencies otherwise hidden in trial
averages~\citep{makeigEtAl99,jungEtAl01,delormeEtAl07,ontonAndMakeig09}.
However, methods are lacking to establish the statistical significance of
apparent features in these plots. Here we used a multivariate linear
regression model to assess the significance of the observation, from the ERP
image in Figure~\ref{fig:exampleSingleTrialAnalysis}c, that the \gls{ITPC}
evoked by standards was modulated by the \gls{SFPD}. We thus found a new
foreperiod effect on the single-trial \gls{ITPC} triggered by standards.
Importantly, this effect would have been lost if we had indiscriminately
averaged phases across all trials.

The fact that models accurately decoded the \gls{SFP} duration from the
\gls{ITPC} triggered by standard stimuli does not necessarily mean that these
two quantities are correlated with each other. In principle, a sufficiently
complex model can decode arbitrarily accurately any experimental variable from
any physiological measurements, if the complex model overfitted the
data~\citep{gemanEtAl92}. To avoid overfitting we used leave-one-out
crossvalidation to generate model decodings. To further validate that the
obtained results were not an artifact of our decoding method, we showed that
similar modulations of ITPC were obtained by a simple trial averaging procedure
(Section~\ref{sec:directEvidence}). In addition, we developed surrogate
controls to further support the hypothesis that (1) the \gls{SFP} duration is
the experimental variable causing the \gls{SFP} effect on \gls{ITPC} and (2)
the \gls{ITPC} triggered by standards is the aspect of brain activity modulated
by the \gls{SFP} effect on \gls{ITPC} (Section~\ref{sec:controls}). Further
evidence for the reliability of our results comes from the significant
correlations between the decoding power of the models and subjects behavioral
measures (Section~\ref{sec:correlationsWithBehavior}) and from the anatomical
specificity of the clusters with a large proportion of significant models
(Section~\ref{sec:measuringSFPeITPC}).

\subsection{What psychophysical processes generate the SFP effect on ITPC?}

The current investigation suggests that temporal expectancy plays an important
role in the generation of the \gls{SFP} effect on \gls{ITPC}, although further
investigations are needed to validate this suggestion, as discussed below.
First, the strongly periodic stimulation in our experiment should have induced
temporal expectations in subjects. This claim is supported by the
observation of reliable foreperiod effects on reaction times in 26\% of the
subjects and \glspl{attendedModality} (Section~\ref{sec:behavioralResults}) and
by previous arguments relating the foreperiod effect on reaction times to
expectancy~\citep[e.g.,][]{niemiAndNaatanen81}.
Second, temporal expectations contribute to faster reaction times and improved
perception~\citep[e.g.,][]{correa10}, while the strength of the \gls{SFP}
effect on \gls{ITPC} was larger in subjects achieving larger detection rates
and faster mean reaction times
(Figure~\ref{fig:correlationsWithBehavior} and Tables~\ref{table:statsMeanRTs}
and~\ref{table:statsErrorRates}).
Third, although it is possible to develop temporal expectancies for unattended
stimuli, these expectancies are stronger for attended stimuli. Hence, the
strength of the foreperiod effect on \gls{ITPC} should be larger for attended
than unattended stimuli (part II of this manuscript).
Fourth, as we discuss in the next section, most of the brain regions where we
observed significant correlations between the strength of the \gls{SFP} effect
on \gls{ITPC} and behavioral measures, have been implicated in temporal
expectation.

% references from C17 and C26, Attention and Time

% motor \citep{sanders80, sanders98, spijker90, mattesAndUlrich97, bruniaAndBoelhouwer88, requinEtAl91hasbroucqEtAl99, tandonnetEtAl03, tandonnetEtAl06, coullAndNobre98}; for review see \citet{rolkeAndUlrich10}.

% premotor \citep{bausenharEtAl06, smulders93, hackleyAndValleInclan98, hackleyAndValleInclan99, mullerGethmannEtAl03}.

% perceptual \citep{correaEtAl05, correaEtAl06a, dohertyEtAl05, correaEtAl06b, mentoEtAl13, correaEtAl07, langeEtAl06, lange09, martensAndJohson05, rolkeAndHofmann07, rolke08, bausenhartEtAl07, langeEtAl03, sandersAndAstheimer08}

% executive \citep{naccacheEtAl02, correaEtAl10}

An unresolved question in the field of temporal expectation is whether it
affects motor (i.e., rate of motor response), premotor (i.e., response
preparation), perceptual (i.e., buildup of information about the stimulus), or
executive (i.e., decision mechanism) stages. The majority of the evidence
suggests a motor influence \citep[e.g., ][]{sanders98, bruniaAndBoelhouwer88,
coullAndNobre98}, premotor effects have also been reported~\citep[e.g.,
][]{bausenhartEtAl06, hackleyAndValleInclan99, mullerGethmannEtAl03}, and more
recent studies have shown influences on perceptual \cite[e.g.,
][]{correaEtAl05, mentoEtAl13, lange09, rolke08} and executive \cite[e.g.,
][]{naccacheEtAl02, correaEtAl10} functions. 
Results from the current study suggest that the aspect of temporal expectation
measured by the \gls{SFP} effect on \gls{ITPC} mostly affects high-order
visual processing.
That the strength of the \gls{SFP} effect on \gls{ITPC} is significantly
correlated with stimuli detectability
(Figure~\ref{fig:correlationsWithBehavior}a,
and~\ref{fig:correlationsWithBehavior}b and Table~\ref{table:statsErrorRates})
and that the \gls{SFP} effect on \gls{ITPC} is calculated from EEG activity
triggered by standard stimuli to which subjects did not respond, indicates that
the temporal expectation measured by the \gls{SFP} effect on \gls{ITPC} affects
non-motor activity.
Also, the long latencies after the presentation of standards at which
modulations of the \gls{SFP} were largest (the median time of the largest peak
of the regression coefficients was 292~ms, Section~\ref{sec:timing}) argues
against a relation between the temporal expectation measured by \gls{SFP}
effect on \gls{ITPC} and early sensory stages.
Lastly, that the \gls{SFP} effect on \gls{ITPC} was substantially stronger for
the visual than for the auditory \gls{standardModality} (as indicated in
Section~\ref{sec:measuringSFPeITPC}, Figure~\ref{fig:clusters}, and
Table~\ref{table:statsMeanRTs}, the number of clusters with a large proportion
of models significantly different from the intercept-only models was larger for
the visual \gls{standardModality} and, as shown in
Section~\ref{sec:correlationsWithBehavior}, Figure~\ref{fig:clusters}, and
Tables~\ref{table:statsMeanRTs} and~\ref{table:statsErrorRates}, the number of
clusters showing significant correlations between models' decodings and
behavioral measures, as well as the strength of these correlations, was larger
for the visual than the auditory \gls{standardModality}) suggests that the
effects of the temporal expectation measured by the new foreperiod effect may
be specific to the visual modality. Note, however, that the weakness of the
\gls{SFP} effect on \gls{ITPC} on the auditory modality could be due to the
fact that our 32-channel EEG recording system may not have covered auditory
brain regions with sufficient density.

\subsection{Relations to previous investigations}

% CNV
% Consistency with C19
% Surprise with C04, but Gomez and others

% Brain regions activated during the CNV
% hultinEtAl96-magneticLocalizationOfCVN.pdf
% gomezEtAl03-visuoMotorNetworkCNVByCurrentDensity.pdf
% gomezEtAl04-cmvCNVLocalizationIsTaskSpecific.pdf
% Gomez et al. 07: fronto-parietal network
% References from gomezEtAl03.

Previous investigations have shown that brain regions associated with clusters
whose \gls{ITPC} activity is significantly correlated with detection
performance (colored cells in Table~\ref{table:statsErrorRates}) or with
reaction speed (colored cells in Table~\ref{table:statsMeanRTs}) are
related to temporal expectancy, as discussed next.

Electrophysiological correlates of temporal expectation generated by the
foreperiod effect have mainly been investigated through the CNV. Studies
demonstrating that the CNV amplitude is maximal at the vertex~\citep{walter67},
that the the supplementary motor area (SMA) and the anterior cingulate cortex
(ACC) are the sources of the CNV's O wave~\citep{cuiEtAl00, zappoliEtAl00,
gomezEtAl01}, and that the premotor cortex appears to give rise to the CNV's E
wave \citep{hultinEtAl96}, indicate the central brain region associated with
cluster 19 is related to attention orienting and response preparation. Also,
more recent studies have shown that posterior sites contribute to the
generation of the E wave of the CNV~\citep{gomezEtAl01, gomezEtAl03}. Of
particular interest are the asymmetric activations in the left middle-occipital
cortex observed by~\citet[][Figures 2 and 3]{gomezEtAl03}. These studies
suggest that activity in the left parieto-occipital region associated to
cluster 04 (Figure~\ref{fig:clusters3D}) may be related to perceptual
anticipation.

Findings from a CNV study focused on attention in time by
\citet{miniussiEtAl99} are in close agreement with results reported here. In
this study subjects were instructed to detect a target stimulus as quickly as
possible. Two central cues which predicted if the subsequent target would occur
600 or 1400 ms after cue onset. Their analysis showed that the influence of
the foreperiod duration on event-related potentials elicited by the cue changed over time. The
earliest significant influence appeared 280~ms post-cue and occurred in
left-parieto-occipital electrodes, and by 480~ms post-cue this difference has
moved to the vertex~\citep[Figure 7 in ][]{miniussiEtAl99}. Thus, the
significant correlation between the strength the \gls{SFP} effect on \gls{ITPC}
and error rates observed in the left-parieto-occipital cluster 04 and the
central-midline cluster 19 (Figure~\ref{fig:clusters3D}) agrees with
the significant influences observed in~\citet{miniussiEtAl99} in
left-parieto-occipital and vertex electrodes. In addition, the observation
that the \gls{SFP} effect on \gls{ITPC} occurred earlier in occipital than in
central clusters (Figure~\ref{fig:timing}b) agrees with the finding
by~\citet{miniussiEtAl99} that significant foreperiod influences appeared
earlier in left-parieto-occipital electrodes than in electrodes surrounding the
vertex.

% From coull09
% coullEtAl11-neuroanatomicalAndNeurochemicalSubstratesOfTiming.pdf

Fixed temporal expectations of when an event is likely to occur are marked by
activity in left premotor and parietal areas. However, if the event has still
not appeared by the expected delay, the right prefrontal cortex updates current
temporal expectations~\citep{coull09}.
The left prefrontal cluster 18 and left parietal cluster 5
(Figure~\ref{fig:clusters3D}) may be related to correspondingly
left-lateralized fMRI activations in attention on time
tasks~\citep[e.g.,][]{coullAndNobre98}.
Updates of temporal expectations occurred in the present study, since it used
variable foreperiod durations. Previous lesion~\citep{vallesiEtAl07} and
fMRI~\citep{vallesiEtAl09,coullEtAl00} studies have shown involvement of
activity in the right prefrontal cortex related to updating temporal
expectations in variable foreperiod duration tasks. Thus, we expected to
observe a significant correlation between the strength of the \gls{SFP} effect
on \gls{ITPC} and reaction times in the right prefrontal cluster 13
(Figure~\ref{fig:clusters3D}).
An early fMRI study on attention on time compared brain regions
activated by spatial and temporal attention~\citep{coullAndNobre98}. This and
other studies~\citep[e.g.,][]{coullEtAl04} found a region next the left insula
only activated during temporal attention. Interestingly, many of the \glspl{IC}
in cluster 7 are close to the insula (Figure~\ref{fig:clusters3D}).

% variable foreperiod duration (Nobre). Comment on the similarity C17 -> C16 and
% the difference with Gomez 03.

Cravo and colleagues specifically studied how temporal expectation modulates
the \gls{ITPC} of neural oscillations in the EEG of humans. In a first
study, \citet{cravoEtAl11} asked if temporal expectations influenced
low-frequency rhythmic activity in the theta range (4-8~Hz). They
found that temporal expectations modulate reaction times simultaneously with
the amplitude, \gls{ITC}, and phase-amplitude coupling in central-midline
electrodes. Interestingly, they reported increased values of \gls{ITC} at times
of higher expectancy. 
The authors interpreted these modulations as reflecting motor preparation. Our
finding of a significant correlation between the \gls{SFP} effect on \gls{ITPC}
and detectability in the central cluster 19
(Figure~\ref{fig:correlationsWithBehavior}a) is
consistent with the report in \citet{cravoEtAl11}. But, as we discussed above,
we believe the modulations of \gls{ITPC} observed in this manuscript are
related to non-motor neural processes. 
The study by Cravo and colleagues used three foreperiod durations between a
warning signal and a Go/No-go target. Differently from previous and from this
study (see below), the authors found no modulations of \gls{ITC} with the
foreperiod duration. A possible explanation for this difference is that Cravo
and colleagues averaged \glspl{ITC} across subjects before testing for
modulations by the foreperiod duration. These modulations may have been present
at the single subject level, but obscured in the average, which highlights the
relevance of avoiding averages across subjects.
In a second study \citet{cravoEtAl13} studied the effect of phase entrainment
by repetitive visual stimulation of rhythmic activity in the delta range
(1-4~Hz) over occipital cortex in an orientation discrimination task. They
showed
that 150~ms before the presentation of a target stimulus the phase of delta
oscillations clustered around the phase corresponding to the largest contrast
gain. Considering the effects of expectancy in central brain regions in the
Go/No-Go task reported by \citet{cravoEtAl11}, and those in occipital brain
regions in the orientation discrimination task reported in \citet{cravoEtAl13},
the authors argued that ``\emph{the sites where temporal expectation modulate
low-frequency oscillations \ldots are highly dependent on the task demands}.''
However, in \citet{cravoEtAl11} effects were only studied, or at least
reported, over central regions, while in \citet{cravoEtAl13} effects were only
reported over occipital regions. It is possible that extending their analysis
to multiple brain regions could show effects over occipital areas
in~\citet{cravoEtAl11} and over central areas in~\citet{cravoEtAl13} and
invalidate their argument. This highlights the relevance of reporting results
over all brain regions.

\begin{comment}
% Oddball

Our audio-visual oddball detection task (Figure~\ref{fig:experiment}) is
similar to those in previous EEG-fMRI~\citep[e.g.,][]{lindenEtAl99},
depth-electrodes ~\citep[e.g.,][]{halgrenEtAl95}, and
lesion~\citep[e.g.,][]{knightAndNakada98} studies. However, we characterized
the EEG activity evoked by standards, while these previous studies mostly
analyzed activity evoked by deviants, and our task required more frequent
shifts of attention between the visual and the auditory modality (approximately
four shifts per minute) than these previous studies. In spite of these
differences, some of the clusters where we found significant correlations
between the strength of the \gls{SFP} effect on \gls{ITPC} and behavioral
measures are located in brain regions reported in these previous studies as
reflecting oddball-related processes: the supramarginal gyrus~(clusters 6 in
our study), the posterior cingulate cortex (cluster 5), the anterior cingulate
cortex (cluster 18), and the insula (cluster 7).
\end{comment}

% Yamagishi et al. 2008

In an MEG investigation of the effect of attention on primary visual cortex
activity, \citet{yamagishiEtAl08} showed that \gls{ITC} over a left-occipital
brain region was significantly correlated with subjects' orientation
discrimination ability when attention was oriented toward, but not away from, a
grating located on the right visual field. The authors concluded that the
phase-resetting measured by \gls{ITC} reflects an attention- and vision-related
top-down modulation to primary visual cortex.
Using a different brain-imaging modality, the investigation by
\citet{yamagishiEtAl08} supports our previous finding that \gls{ITPC} is
related to behavioral performance, and the finding of part II of this
manuscript that \gls{ITPC} is modulated by attention. Also, it is an
interesting coincidence that \citet{yamagishiEtAl08} found significant
correlations between \gls{ITC} and behavior in four time frequency points
between 200 and 300~ms after an attention-orienting cue, and that we found the
median of the largest \gls{SFP} effect on \gls{ITPC} at 292~ms
(Section~\ref{sec:timing}). However, our study and that of that
\citet{yamagishiEtAl08} differ in an important regard; they used a
trial-averaged measure of \gls{ITPC} while we used a single-trial one. This
single-trial analysis allowed us to find a new foreperiod effect, which shows
that \gls{ITPC} is not constant across trials, but varies systematically as a
function of the \gls{SFPD}.
To determine if in our study \gls{ITC} was also related to behavioral
performance, we correlated the peak \gls{ITC} value following the presentation
of a standard (measured as indicated in
Section~\ref{sec:itcAndPeakITCFrequency}) with error rates and mean reaction
times for all subjects, \glspl{standardModality}, and~\glspl{attendedModality}.
Details of this investigation are reported in Section~\ref{sec:itcAndBehavior},
Figure~\ref{fig:clustersITC}, and Tables~\ref{table:peakITCStatsErrorRates}
and~\ref{table:peakITCStatsMeanRTs}. Similarly to \citet{yamagishiEtAl08}, we
found that in the left-parieto-occipital cluster 04 \gls{ITC} was
significantly correlated with detection performance for attended, but not
unattended, visual stimuli. However, in the right-parieto-occipital cluster 15
we found a significant correlation between the \gls{ITC} triggered by attended
auditory standards and detection performance, and in the central
parieto-occipital cluster 09 we found a significant correlation between the
\gls{ITC} triggered by unattended auditory standards and detection performance.
These results appear to contradict the conclusions from \citet{yamagishiEtAl08}
that phase reset over occipital regions reflects vision-related top-town
modulations and that these modulations are related to attention. However, this
contradiction may be due to differences between our experiments (our temporal-
versus their spatial-attention experiments, or our detection versus their
discrimination experiments) or differences in the selection of the
time-frequency bin used to compute \gls{ITC} (we used the subject-,
\gls{standardModality}-, and \gls{attendedModality}-specific peak \gls{ITC}
frequency while \citet{yamagishiEtAl08} used the same time-frequency point for
all subjects).
The previous significant correlations between peak \gls{ITC} values and
subjects' error rates were found after computing 56 correlations. None of these
correlations survived the multiple-comparison adjustment
(Section~\ref{sec:multipleComparisons}). Therefore, some of the these
significant correlations could have appeared by chance. However, we suspect
that not all of them are spurious since 
eight of the nine unadjusted
significant correlations were negative. Assuming that these correlations
occurred by chance, the probabilities of obtaining positive and negative
correlations should be equal. Under this assumption, the probability of
obtaining eight out of nine significant correlations is very small ($p<0.02$).
In our study, the significant correlations between \gls{ITC} and behavior and
those between the \gls{SFP} effect on \gls{ITPC} and behavior appear to reflect
different neural processes since, with the exception of the correlation with
error rates in the left parieto-occipital cluster 04, correlations with
\gls{ITC} occurred in different clusters, \glspl{standardModality}, and
\gls{attendedModality}, than those with the \gls{SFP} effect on \gls{ITPC}
(compare daggers in Figures~\ref{fig:clusters} and~\ref{fig:clustersITC}). For
example, most significant correlations between \gls{ITC} and error rates error
rates corresponded to auditory standards (Figure~\ref{fig:clustersITC} and
Table~\ref{table:peakITCStatsErrorRates}), while those between the \gls{SFP}
effect on \gls{ITPC} and error rates corresponded to visual standards
(Figure~\ref{fig:clusters} and Table~\ref{table:statsErrorRates}).

An investigation by \citet{buschEtAl09} also showed that stimulus detectability
is related to prestimulus phase. Below we highlight a few differences between
this study and ours. The purpose of this comparison is not to belittle the
research by \citet{buschEtAl09}, which we believe was excellent, but to
emphasize a few novel characteristics of our analysis.
First, the \gls{ITC} and the phase coherence index used in the main text of
\citet{buschEtAl09} are phase measures for groups of trials, while the
\gls{DMP} used in this article is a phase measure for single trials. Thus, we
were able to decode the \gls{SFP} duration from \gls{DMP} values in a
trial-by-trial basis, while these single-trial decodings are not possible with
phase measures for groups of trials.
In the supplementary Figure 2, \citet{buschEtAl09} reported the results of
using a classifier to predict in a single trial if a stimulus will be detected
based on the phase at a single time-frequency point. The fact that, after
averaging the predictions of the classifiers of different subjects, they found
a correlation between the predictions of the classifiers and the detection hit
ratio is remarkable. However the analysis method is limited by using only one
time-frequency point to predict the detectability of a stimulus. The models
used in this manuscript used a 500~ms window of phase values to decode the
\gls{SFP} duration. That is, the second difference between the study by
\citet{buschEtAl09} and ours is that they used univariate models, while we used
multivariate ones.
Note that we have optimized models to decode \gls{SFP} durations, and then
found significant correlations between the decoding power of these models and
behavior without fitting model parameters to behavioral data. That a model
optimized to predict behavior from a stimulus feature (e.g., \gls{ITPC})
achieves reliable predictions does not necessarily mean that the stimulus
feature is relevant to the behavior. A sufficiently complex model can 
predict arbitrarily closely any behavior from an irrelevant stimulus feature,
if the model overfit the data~\citep{gemanEtAl92}.
Although there exist methods to
minimize the risks of overfitting, such as using separate pieces of data to
estimate the parameters of a model and to evaluate its predictive power (the
method used in \citet{buschEtAl09}), this risk never disappears when optimizing
a model to predict behavior. Differently, the procedure used in this article to
show that \gls{ITPC} is related to behavior does not suffer from the
overfitting problem, since no model parameter was fitted to behavioral data.
The third difference is that the study by \citet{buschEtAl09} averaged
data across subject while we analyzed different subjects separately. This
single-subject analysis allowed to observe that those subjects whose brain
activity was more strongly modulated by the \gls{SFP} duration were those who
performed better
Fourth, \citet{buschEtAl09} reported the analysis of a single electrode (Fz),
while our study described the analysis of the activity of \glspl{IC} over the
whole brain. As we discussed above in relation to the research by Cravo and
collaborators, we believe it is important to report effects across the whole
brain to show how hypothesis supported by activity in a given brain region
generalize across the brain.
Fifth, \citet{buschEtAl09} analyzed activity from EEG channels, while the
present study characterized the activity of \glspl{IC}. Applied to EEG, ICA
finds maximally independent sources generating recorded potentials. Scalp
representations of these sources resemble fields generated by current dipoles
inside the brain, and biological arguments suggest that these
maximally-independent sources reflect the synchronized activity of neurons in
compact cortical regions \citep{delormeEtAl12}. Thus, our analysis most
probably characterized the activity of cortical sources while that of
\citet{buschEtAl09} described the activity of mixtures of these sources.
Sixth, the study by \citet{buschEtAl09} characterized phase correlates of
visual processing, while our study investigated phase correlates of both
visual and auditory processing. The multi-modal nature of our experiment
allowed us to see that the \gls{SFP} effect on \gls{ITPC} was substantially stronger in the
visual than in the auditory modality.
An important advantage of the study by \citet{buschEtAl09} over ours is that,
for each subject separately, they calibrated the visual stimuli so that the
detection rate of the subject was close to 50\%. No such individualized
calibration was performed in our experiment, and the averaged detection rate
across subjects was 86\%. Having a similar number of trials where subjects
perceived and failed to perceive the stimulus facilitates the estimation of
the phase correlates of stimulus perception. However, the fact that the
methods used in this manuscript could find reliable phase correlates of
stimulus detection with sub-optimal stimulation shows that these methods
could be applied to a larger number of EEG experiments where stimulation was
not optimized for each subject.

\subsection{What neural processes generate the SFP effect on ITPC?}
\label{sec:neuralProcesses}

What brain processes can make the phases of standards closer to the
\gls{warningSignal} more aligned at some times but more misaligned
at other times than the phases of of standards further away from the
\gls{warningSignal}?
Also, how can it happen that phase alignment 
fluctuates in an oscillatory manner, as indicated in
Figures~\ref{fig:avgDMPsAndCoefs}a-c and~\ref{fig:avgDMPsC04Visualswitchvision}?
Large values of \gls{ITPC} are typically associated with phase resets,
like those that occur upon the presentation of visual stimuli~\citep[e.g.,][]{makeigEtAl02}. But a single phase-reset event is
not the answer to the previous question, since after it all trials should have a similar
phase, so that there would be no difference in \gls{ITPC} between trials
closer to and further away from the \gls{warningSignal}. A single phase-reset
event can neither explain the fluctuations in \gls{ITPC}.
These fluctuations are reminiscent of known variations in perception and
reaction times that have been related to the phase of low-frequency
neural
rhythms~\citep{bartleyAndBishop32,bishop32,jarcho49,lindsley52,lansing57,callawayAndYeager60,dustmanAndBeck65,trimbleAndPotts75,valeraEtAl81,mathewsonEtAl09,buschEtAl09,mathewsonEtAl12,buschAndVanRullen10,bonnefondAndJensen15,liuEtAl14,drewesAndVanRullen11,thorneEtAl11,fiebelkornEtAl11,chakravarthiAndVanRullen12,cravoEtAl15,miltonAndPleydellPearce16},
reviewed in~\citet{vanRullenAndKoch03}.
However, if after a reset the phase evoked by standards oscillated identically
across trials there would not be a difference in \gls{ITPC} between trials
closer and further-away from the \gls{warningSignal}. 

A brain mechanism generating the above effects should be able to replicate the
fluctuations in \gls{ITPC} shown in Figures~\ref{fig:avgDMPsAndCoefs}a-c
and~\ref{fig:avgDMPsC04Visualswitchvision}. We make three observation:
First, these fluctuations appear to oscillate at a very low
frequency (i.e., less than 1~Hz), independently of the frequency at which the
phase of trials was measured (the latter frequency is that of the cosine of the
mean phase given by the dashed curve in Figures~\ref{fig:avgDMPsAndCoefs}a-c
and~\ref{fig:avgDMPsC04Visualswitchvision}). Coincidentally, fluctuations at 1~Hz in
visual detectability have been reported by \citet{fiebelkornEtAl11}.
Also,
fluctuations in somatosensory detectability (between 0.01 and 0.1~Hz) have been
described by \citet{montoEtAl08}.
Second, in some figures the phase of the oscillations in \gls{ITPC} at time
zero differs between early and late trials (e.g., the phase of the blue curve
corresponding to early trials in Figure~\ref{fig:avgDMPsAndCoefs}a is more
advanced than that of the red curve corresponding to late trials).
Third, as shown in Figure~\ref{fig:correlationsWithBehavior}, these
fluctuations are related to subjects' behaviors.

We propose a simple model for the generation of low-frequency oscillations that
can account for the previous observations. The component of the
recorded potential corresponding to trial $i$, at frequency $f$, and time
$t$ is given by Eq.~\ref{eq:trial}:

\begin{eqnarray}
trial(i,f,t)&=&cos(2\pi ft+\theta+n(SFPD[i],t)),\;\text{with}\label{eq:trial}\\
n(SFPD,t)&=&\mathcal{M}(\mu=0,\kappa=\kappa(SFPD,t))\label{eq:noise1}\\
\kappa(SFPD,t)&=&\frac{\text{max}\kappa-{min}\kappa}{2}\;cos(2\pi f_n
t+\theta(SFPD))+\frac{{min}\kappa+\text{max}\kappa}{2},\;\text{and}\label{eq:noise1Precision}\\
\theta(SFPD)&=&\frac{\pi}{(\text{maxSFPD}-\text{minSFPD})^s}(\text{maxSFPD}-SFPD)^s\label{eq:noise1PrecisionPhase}
\end{eqnarray}

To account for the first previous observation, we assume that the phase of the
cosine is contaminated by an additive noise $n$ (Eq.~\ref{eq:noise1}) following
a von Mises distribution with a precision parameter, $\kappa$, varying
sinusoidally on time (Eq.~\ref{eq:noise1Precision}). When the precision of this
noise is small and large we should observe decoherent and coherent activity,
respectively, and, because the noise precision varies sinusoidally in time, we
should observe alternations between coherence and decoherence, as shown in
Figures~\ref{fig:avgDMPsAndCoefs}a-c
and~\ref{fig:avgDMPsC04Visualswitchvision}. 
To account for the second previous observation, we make the phase at time zero
of the precision of the noise vary smoothly as a function of the \gls{SFPD}
(Eq.~\ref{eq:noise1PrecisionPhase}).
To account for the third observation, we speculate that subjects achieving
better detection performance were those showing larger modulations in the
precision of the noise. For these subjects the \gls{ITPC} should be more
different between trials closer to and further away from the
\gls{warningSignal}, and models should more reliably decode the \gls{SFPD} from
\gls{ITPC} activity, in agreement with
Figure~\ref{fig:correlationsWithBehavior}. 
The parameters of the previous model are the frequency, $f$, the noiseless
phase, $\theta$, and the vector of \glspl{SFPD}, $SFPD$, in Eq.~\ref{eq:trial};
the minimum and maximum values, ${min}\kappa$ and $\text{max}\kappa$, and the
frequency, $f_n$, of the noise precision in Eq.~\ref{eq:noise1Precision}; and
the steepness of the change in the noise precision phase as a function of the
\gls{SFPD}, $s$, in Eq.~\ref{eq:noise1PrecisionPhase}.

To validate the previous speculations we simulated data following the model in
Eq.~\ref{eq:trial} (Figure~\ref{fig:simulatedOscillations}) and fitted a
decoding model to this data, as described in Section~\ref{sec:simulations}.
With this simulated data we obtained oscillations in averaged \gls{DMP}
(Figures~\ref{fig:avgDMPsAndCoefsSimulated}a,b), as in
Figure~\ref{fig:avgDMPsAndCoefs}a-c. We also attained oscillations in the
difference between the averaged \gls{DMP} between trials closer to and further
away from the \gls{warningSignal}
(Figures~\ref{fig:avgDMPsAndCoefsSimulated}c,d), as in
Figure~\ref{fig:avgDMPsAndCoefs}d-f.
In addition, we verified that a model fitted to data with larger fluctuations
of the precision of the noise generated more accurate decodings than those of
a model fitted to data with smaller modulations of this precision, supporting
the previous argument on the relation between fluctuation of \gls{ITPC} and
subjects' behaviors.  As explained in Section~\ref{sec:simulations}, the
correlation coefficient for a model fitted to data with smaller fluctuations of
the precision of the noise was r = 0.23 (95\% CI=[0.12, 0.34]), which was
significantly smaller than that for a model fitted to data with larger
fluctuations of the precision of the noise r = 0.51 (95\% CI=[0.43, 0.59]).

These results show that a simple sinusoidal oscillation with a noise process
added to its phase captures salient features of the \gls{SFP} effect on
\gls{ITPC}, suggesting that this model is a good first description of how the
\gls{SFPD} modulates the phase of neural oscillations. Key properties of this
noise process are its precision oscillating in time and its
\gls{SFPD}-dependent phase at time zero.

\subsection{Single-trial predictive methods for EEG}

Single-trial predictive models are seldomly used with EEG recordings to
characterize the neural basis of psychophysical process~\citep[but
see][]{sajdaEtAl09,pernetEtAl11a}.
These models are mostly applied to EEG recordings in brain computer interface
applications~\citep[e.g.,][]{nicolasAlonsoAndGomezGil12,wolpawEtAl02}.
However, these applications are usually focused on the predictive power of
these models, and not so much on the brain mechanisms underlying the
psychophysical process they attempt to predict.
This omission is striking, since single-trial predictive models have been very
successful for understanding the neural basis of psychophysical processes using
fMRI
recordings~\citep[e.g.,][]{haxbyEtAl01,haynesAndRees05,kamitaniAndTong05,kayEtAl08,naselarisEtAl09,serencesAndBoynton07,reddyEtAl10,polynEtAl05,mitchellEtAl08,abramsEtAl11,hamptonAndODoherty07}.
For reviews see~\citet{naselarisEtAl11, tongAndPratte12}.
Note that applying single-trial predictive methods to EEG recordings offers the
possibility of understanding the neurophysiological basis of behaviors that
occur at fast time scales, which cannot be studied with typical fMRI
recordings. For example, it would have been impossible to find that the
\gls{SFPD} modulates the \gls{ITPC} triggered by standards since these
modulations occur in a 500~ms-long time window that cannot be resolved with the
two seconds sampling rate of typical fMRI experiments.

LIMO EEG~\citep{pernetEtAl11b} is an existing single-trial predictive method
for EEG recordings. It attempts to predict the activity of electrodes using
experimental factors as inputs. The method described in this manuscript differs
from the previous one in that the EEG recordings are used as inputs to predict
an experimental factor (i.e., the \gls{SFP} duration). That is, LIMO EEG is an
encoding method, aiming at describing how experimental factors are encoded in
EEG activity, while the method introduced here is a decoding one, seeking to
predict experimental factors from EEG activity.
An advantage of decoding over encoding methods is that they can be readily used
to determine that a certain pattern of brain activity (e.g., \gls{ITPC}) is
related to a particular behavioral characteristic (e.g., reaction times).  This
can be done by using the decoding method to decode an experimental factor that
is related to the particular behavior (e.g., the \gls{SFPD} is related to
reaction times). Then the decoding of the method should be related to the
behavior, as we showed in this article.  On the contrary, is not
straightforward to relate patterns of brain activity and behavior with encoding
methods.
Another advantage of the method used here compared to LIMO EEG is that the
former method is a multivariate test, while the latter uses univariate tests.
Both methods can be used to test if a scalar trial attribute (e.g., the
\gls{SFPD}) is related to a time-varying trial attribute (e.g, the \gls{ITPC}
at time t). However, the method introduced here uses information across all
time points to test if a time-dependent trial attribute at a given time is
related to the scalar trial attribute, while LIMO EEG only uses information at
one time point to test if the time-dependent trial attribute at that point is
related to the scalar trial attribute. If the time-varying trial attribute is
correlated across time (as happens with most EEG time-varying attributes) a
multivariate method should be more powerful than a univariate one.

The construction of predictive methods for EEG data requires to address two
signal processing problems: the curse of dimensionality and multicollinearity
of the EEG regressors. The curse of dimensionality refers to the exponential
increase of the volume of a space as its dimensionality
grows~\citep{bellman61}. In predictive models the dimensionality of the input
space is given by the number of predictors. To estimate accurate models one
should use training data densely covering the volume of the input space.
However, if the input space is very high dimensional, collecting enough
training data to densely cover the ``exponentially large'' volume of the
input space becomes unfeasible. For example, in EEG predictive methods one
would like to use all EEG samples in an interval preceding an event of
interest (i.e., presentation of a standard) as regressors to predict
an stimulus property (i.e., the \gls{SFP} duration). However, if this interval is
long and/or the EEG sampling rate is large, the number of regressors will be
very large and the predictive method will become impractical, due to the curse
of dimensionality. Multicollinearity in the inputs is another important
problem in predictive models because, among other things, it yields
parameters estimates with large variance~\citep{belsleyEtAl04}, which in turn
limits the inferences that can be made from these estimates.

The predictive method used in this article was a simple one. To address the
previous signal processing problems we used a regularized linear regression
model (Section~\ref{sec:linearRegressionModel}) and estimated its parameters
using a variational method (Section~\ref{sec:lmEstimationMethod}).
We have developed more elaborated predictive methods for the characterization
of responses of visual cells from their responses to (correlated) natural
stimuli~\citep{rapelaEtAl06, rapelaEtAl10}. These are nonlinear methods
specifically designed to address the curse of dimensionality and
multicollinearity problems. Using nonlinear predictive models, carefully
validated with behavioral data, to characterize EEG
recordings could generate superior predictions and be able to model a
broader range of psychophysical processes than those achievable with linear
models. We will apply these more advanced predictive
model to EEG recordings in future investigations.

In summary, the decoding method used in this article is just a first step in
using EEG activity to decode behavior and there exist ample possibilities for
further improvements.

\subsection{Future work}

In future investigations we will perform new experiments to test how
psychophysical factors modulate the \gls{SFP} effect on \gls{ITPC}. For this we
will build on thoughtful investigations on the psychophysical factors
associated with the CNV. Is the \gls{SFP} effect on \gls{ITPC} related to
expectancy~\citep{walterEtAl64}? Would it disappear if responses to deviants
are not required? Would it be attenuated if a proportion of deviants do not
require response? Is the \gls{SFP} effect on \gls{ITPC} linked to motor
preparation~\citep{gaillard77, gaillard78}? Does the emergence of this effect
requires motor responses? What is the relation between the \gls{SFP} effect on
\gls{ITPC} and intention to respond~\citep{lowEtAl66}? Would the strength of
this effect be directly proportional to the anticipated force needed in the
motor response? Is the \gls{SFP} effect on \gls{ITPC} related to
motivation~\citep{mcCallumAndWalter68}? Would the strength of this effect be
augmented when subjects are instructed to concentrate hard and respond quickly
to deviants? How does this effect relate to arousal~\citep{tecce72}? Would the
relation between arousal and this effect be U-shaped, as is the relation
between arousal and CNV amplitude~\citep[][Figure 6b]{tecce72}? Is the
\gls{SFP} effect on \gls{ITPC} linked to to explicit~\citep{macarAndVidal03,
pfeutyEtAl03, pfeutyEtAl05} or implicit~\citep{praamstraEtAl06} time
perception? Would this effect reflect the perceived duration of a target
interval when subjects compare it to the duration of a memorized interval? We
will address these questions in future investigations.

\section{Summary}
\label{sec:summary}

We claimed that phase coherence is a key aspect of neural activity currently
under used for the characterization of EEG recordings.
We proposed a single-trial measure of phase coherence (Figure~\ref{fig:dmp}) an
a method to relate it to experimental variables
(Figure~\ref{fig:exampleSingleTrialAnalysis}).
Using this measure and method we demonstrated a novel foreperiod effect on
single-trial phase coherence. 
We showed that the new foreperiod effect is not an artifact of the proposed
method since it can be observed in simple trial averages
(Figure~\ref{fig:avgDMPsAndCoefs}).
We demonstrated the relevance of the new foreperiod effect by reporting strong
correlations between the strength of the effect and subjects' error rates and
mean reaction times (Figure~\ref{fig:correlationsWithBehavior}).
We argued that temporal expectancy plays an important role in the generation of
the new foreperiod effect.
We hope the present manuscript has provided convincing evidence for the
relevance of phase coherence for the characterization of EEG recordings and
that it has shown that single-trial decoding methods can provide useful
information about brain processes.
We anticipate that the new foreperiod effect on phase coherence will lead to as
many important insights about anticipatory behavior as the foreperiod effect on
reaction times has done until now.

\section{Methods}
\label{sec:methods}

% EEG acquisition and preprocessing
% ICA decomposition and IC clustering
% Characterized epochs
% Peak ITC frequency
% DMP
% Variational Bayes clustering
% Linear regression analysis
% Statistical information:
%   . 95% confidence intervals for regression coefficients
%   . bootstrapped F test for goodness of fit
%   . cross-validated prediction, 
%   . statistics of correlation coefficient between sFP predictions and 
%     SFP durations (Fig. 3d), prediction errors and reaction times (Figs.
%     A.2a and A.2b, DFP durations and reaction times (Figs. A.2c and A.2d)
%   . dual bootrap for correlation coefficient between subjects' averaged sFP
%     durations and subjects' mean reaction time Fig. A1)
%   . bootstrap for difference of means in errors (Fig. 5)
%   . bootstrap for difference of medians in proportion of significant models
%     (Figs. 6, and A.3)
%   . 
% ANCOVA for behavioral data.
% ANCOVA analysis of mean reaction times (Fig. 7b)
% Equivalent filters (Fig. 7c)
% ANCOVA analysis of mean filter values (Fig. 7d)
%
% Provide code and data
%
\subsection{Experiment information}
\label{sec:experimentalInformation}

We analyzed the experimental data first characterized in
~\citet{ceponieneEtAl08}. Below we summarize features of this data relevant
to the present study;  further details are given in the previous article.

\subsubsection{Subjects}

We only characterized the younger-adult subpopulation
in~\citet{ceponieneEtAl08}, comprising 19 subjects (11 females) with a mean age
of $25.67\pm5.94$ years.

\subsubsection{Stimuli}
\label{sec:stimuli}

Stimuli were sequentially presented in a visual and an auditory stream
(Figure~\ref{fig:experiment}a).  Auditory stimuli were 100~ms duration, 550~Hz
(deviants) and 500~Hz (standards) sine-wave tones. These tones were played via
two loudspeakers located at the sides of a 21-inch computer monitor. Visual
stimuli were light-blue (deviants) and dark-blue (standards) filled squares
subtending 3.3 degrees of visual angle, presented on a computer monitor for
100~ms on a light-gray background. Interspersed among the deviant and standard
stimuli were attention-shifting cues. These cues were presented bimodally for
200~ms, by displaying the words HEAR (LOOK) on orange letters on the computer
monitor and simultaneously playing the words HEAR (LOOK) on the loudspeakers.
The inter-stimulus interval (ISI) between any two
consecutive stimuli varied pseudo-randomly between 100 and 700~ms as
random samples from short, medium, and
large uniform distributions, with supports [100-300~ms], [301-500~ms], and
[501-700~ms], respectively. 
In each block, the ISI of 132, 60, and 72
pseudo-randomly chosen stimuli were drawn from the short, medium, and large
ISI distributions, respectively.  
Stimuli were presented in
blocks of 264 for a duration of 158 seconds. Each block consisted of 24 cue
stimuli (12 visual) and 240 non-cue stimuli (120 visual). The 120 non-cue
stimuli of each modality comprised 24 deviants and 96 standards.  
A video of the experimental
stimuli appears at \url{http://sccn.ucsd.edu/~rapela/avshift/experiment.MP4}

\subsubsection{Experimental design}
\label{sec:experimentalDesign}

The experiment comprised FOCUS-VISION, FOCUS-AUDITION, and SWITCH blocks. In
FOCUS-VISION (AUDITION) blocks subjects had to detect visual (auditory)
deviants and ignore attention-shifting cues. In SWITCH blocks these cues became
relevant, and after a LOOK (HEAR) cue subjects had to detect visual (auditory)
deviants. The type of block was told to the subject at the beginning of each
block. Subjects pressed a button upon detection of deviants.  Each
subject completed four FOCUS-VISION blocks, four FOCUS-AUDITION blocks, and 12
SWITCH blocks.  Here we only characterize SWITCH blocks.  Correct responses in
SWITCH blocks are shown in Figure~\ref{fig:experiment}b.  After a LOOK (HEAR)
cue subjects oriented their attention to the visual (auditory) modality, as
indicated by the magenta (cyan) segments in Figure~\ref{fig:experiment}b.  A
\gls{warningSignal} is a stimulus initiating a period of expectancy for a
forthcoming deviant.  LOOK (HEAR) cues were the \glspl{warningSignal}
in the SWITCH blocks characterized in this study.

\subsection{EEG acquisition}

Continuous EEG was recorded from 33 scalp sites of the International 10-20
system, with xxx electrodes and a yyy amplifier, and digitized at 250 Hz. The
right mastroid served as reference.

\subsection{Characterized epochs}
\label{sec:characterizedEpochs}

For each \gls{IC} of each subject we built four sets epochs, aligned at time
zero to the presentation of attended visual standards, attended auditory
standards, unattended visual standards, and unattended auditory standards, with
a duration of 500~ms (Figure~\ref{fig:exampleSingleTrialAnalysis}c)
We call \gls{standardModality} to the modality of the stimuli used to align a
set of epochs (visual or auditory), and we call \gls{attendedModality} to the
attended modality corresponding to a set of epochs (e.g., the visual (auditory)
modality for epochs aligned to the presentation of attended (unattended) visual
standards).
The mean number of epochs aligned to the presentation of attended visual
standards, attended auditory standards, unattended visual standards, and
unattended auditory standards was 150, 136, 99, and 94, respectively. Note that
the mean number of epochs aligned to the presentation of attended stimuli was
considerably larger than that aligned to the presentation of unattended
stimuli, for both the visual and auditory standards. Thus, for testing the
influence of attention on the \gls{SFPD} effect on \gls{ITPC}
(Section~\ref{sec:correlationsWithBehavior}) we equalized the number of epochs
used to fit attended and unattended models. For each attended model, we
selected a random subset of the attended epochs of size equal the number of
epochs used to fit the corresponding unattended model.
Figure~\ref{fig:histDLatenciesAndSFPDs}a displays histograms of deviant
latencies, with respect to the presentation time of the previous standard, for
all combinations of \glspl{standardModality} and \glspl{attendedModality}.  To
avoid possible movement artifacts due to the response to deviants, we excluded
from the analysis epochs where deviants of the attended modality were presented
in the 500~ms after the presentation of the standard at time zero.  The number
of deviants excluded from the analysis is given by the sum of counts to the
left of the red vertical lines in Figure~\ref{fig:histDLatenciesAndSFPDs}a.
Figure~\ref{fig:histDLatenciesAndSFPDs}b shows histograms of \glspl{SFPD} for
all combinations of \glspl{standardModality} and \glspl{attendedModality}. The
vertical green line correspond to the mean \gls{SFPD}, which was shorter for
epochs where attention was directed to the visual (1.9~sec, left panels) than
to the auditory modality (2.4~sec, right panels). A repeated-measures ANOVA
with \gls{SFPD} as dependent variable showed a significant main effect of
\gls{attendedModality} (F(1, 16264)=234.31, p\textless 0.0001) and a posthoc
analysis indicated that the mean \gls{SFPD} was shorter in epochs corresponding
to visual than to auditory attention (z=15.31, p\textless2e-16).  Further
information on this ANOVA appears in Section~\ref{sec:anovaFurtherInfo}.

\subsection{Deviation from the mean phase}
\label{sec:dmp}

Let $\theta_0$ be the phase of a given trial, and $\theta_1, \ldots, \theta_N$
be the phases of the trials in a reference set, all phases measured at the same
time-frequency point. Then, the deviation from the mean phase is a measure of
the difference between the phase of the given trial (i.e., $\theta_0$) and the
mean phase of all trials in the reference set, (i.e., mean direction,
$\bar{\theta}(\theta_1, \ldots, \theta_N)$, Eq.~\ref{eq:meanDirection}), as
illustrated in Figure~\ref{fig:dmp}. The circular variance ($CV$,
Eq.~\ref{eq:cv}) is used to measure this difference:

\begin{eqnarray}
DMP(\theta_0|\{\theta_1, \ldots, \theta_N\})=CV(\theta_0, \bar{\theta}(\theta_1, \ldots, \theta_N))
\label{eq:dmp}
\end{eqnarray}

\begin{figure}
\begin{center}
\includegraphics[width=4.5in]{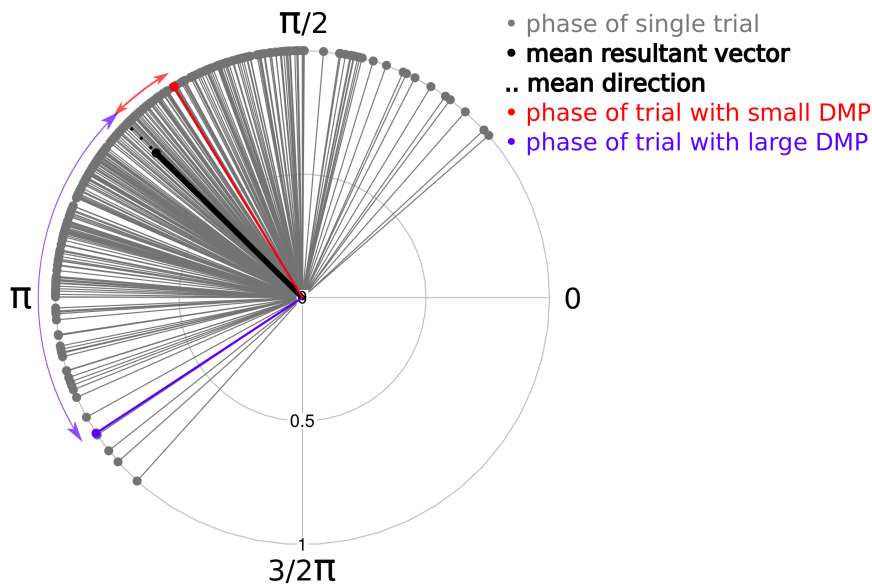}

\caption{Computation of the \gls{DMP}.  Given a set of phases, represented by
gray unit vectors, one first calculates the mean of these unit vectors (the
mean resultant vector, black solid vector) and extracts the phase of this mean
vector (mean direction, Eq.~\ref{eq:meanDirection}, black dotted
unit vector). Then the deviation from the mean phase for a given phase (e.g.,
the phase corresponding to the red unit vector) is a measure of distance (i.e.,
circular variance, Eq.~\ref{eq:cv}) between this phase and the mean direction
(e.g., red arc). The red and blue unit vectors correspond to phases with small
and large \glspl{DMP}, respectively.}

\label{fig:dmp}
\end{center}
\end{figure}

The \gls{DMP} is zero (one) if the phase of the given trial is equal (opposite)
to the mean phase of the trials in the reference set. The \gls{DMP} has
previously been used to investigate interelectrode phase coherence with EEG
recordings~\citep[][Figure 2c]{hanslmayrEtAl07}.  It is an appealing measure of
single-trial phase coherence since, as demonstrated in
Section~\ref{sec:proofAveragedDMPApproachesITC}, when there is large phase
concentration (i.e., \gls{ITC}$\simeq$ 1) the average value of the \gls{DMP}
approaches the ITC, a widely used measure of averaged \gls{ITPC}
(Section~\ref{sec:itcAndPeakITCFrequency}).  When there is not large phase
concentration, the mean phase cannot be estimated reliably, and therefore the
\gls{DMP} becomes unsound.  Here we only computed \gls{DMP} values at times
where the distribution of phases across trials was significantly different from
the uniform distribution (p\textless 0.01, Raleigh test).

\section{Acknowledgments}

We thank Dr.\ Steve Hillyard for suggesting Figure~\ref{fig:avgDMPsAndCoefs}
and to investigate the influence of the maximum \gls{SFPD} in the decodings of
the models (Figure~\ref{fig:selectionMaxSFPD}).
We thank Dr.\ Cyril Pernet for suggesting the use of robust correlation
coefficients.
We also thank Mr.\ Robert Buffington for computer assistance.
Most of the research in this paper was performed with free software.
Computations were done with R~\citep{r12}, the article was written in
\LaTeX~\citep{lamport94}, figures were prepared with
Inkscape~\citep{harringtonEtAl04}, all in a personal computer and in clusters
of personal computers running Linux~\citep{torvalds08}. The only paid software
used in this investigation was MATLAB~\citep{matlab13}, employed solely for
initial preprocessing with ICA~\citep{makeigEtAl96} and for plotting
\glspl{IC}' equivalent dipoles (e.g., inset in
Figure~\ref{fig:exampleSingleTrialAnalysis}b) and clusters of these diples
(e.g., Figure~\ref{fig:clusters}).

% \pagebreak
% \setcounter{page}{1}
% \glsaddall
\printglossaries

\bibliographystyle{plainnatNoNote}
% \bibliographystyle{plain}
% \bibliography{stats,machineLearning,attention,temporalAttention,rhythms,eeg,cnv,timePerception,bci,fmri,receptiveFields,math,brainImaging,signalProcessing,machineLearning,others}
\bibliography{stats,attention,temporalAttention,rhythms,eeg,cnv,timePerception,bci,fmri,receptiveFields,math,brainImaging,signalProcessing,machineLearning,pac,speech,others}

\begin{thebibliography}{187}
\providecommand{\natexlab}[1]{#1}
\providecommand{\url}[1]{\texttt{#1}}
\expandafter\ifx\csname urlstyle\endcsname\relax
  \providecommand{\doi}[1]{doi: #1}\else
  \providecommand{\doi}{doi: \begingroup \urlstyle{rm}\Url}\fi

\bibitem[Abrams et~al.(2011)Abrams, Bhatara, Ryali, Balaban, Levitin, and
  Menon]{abramsEtAl11}
D.A. Abrams, A.~Bhatara, S.~Ryali, E.~Balaban, D.J. Levitin, and V.~Menon.
\newblock Decoding temporal structure in music and speech relies on shared
  brain resources but elicits different fine-scale spatial patterns.
\newblock \emph{Cerebral Cortex}, 21:\penalty0 1507--1518, 2011.

\bibitem[Ahissar et~al.(2001)Ahissar, Nagarajan, Ahissar, Protopapas, Mahncke,
  and Merzenich]{ahissarEtAl01}
Ehud Ahissar, Srikantan Nagarajan, Merav Ahissar, Athanassios Protopapas, Henry
  Mahncke, and Michael~M Merzenich.
\newblock Speech comprehension is correlated with temporal response patterns
  recorded from auditory cortex.
\newblock \emph{Proceedings of the National Academy of Sciences}, 98\penalty0
  (23):\penalty0 13367--13372, 2001.

\bibitem[Bartley and Bishop(1932)]{bartleyAndBishop32}
S.H. Bartley and G.H. Bishop.
\newblock The cortical response to stimulation of the optic nerve in the
  rabbit.
\newblock \emph{American Journal of Physiology--Legacy Content}, 103\penalty0
  (1):\penalty0 159--172, 1932.

\bibitem[Bates(1951)]{bates51}
JAV Bates.
\newblock Electrical activity of the cortex accompanying movement.
\newblock \emph{The Journal of physiology}, 113\penalty0 (2-3):\penalty0
  240--257, 1951.

\bibitem[Bausenhart et~al.(2006)Bausenhart, Rolke, and
  Ulrich]{bausenhartEtAl06}
K.M. Bausenhart, B.~Rolke, and R.~Ulrich.
\newblock The locus of temporal preparation effects: evidence from the
  psychological refractory period paradigm.
\newblock \emph{Psychonomic Bulletin and Review}, 13:\penalty0 536--542, 2006.

\bibitem[Bellman(1961)]{bellman61}
R.E. Bellman.
\newblock \emph{Adaptive control processes}.
\newblock Princeton University Press, Princeton, NJ, 1961.

\bibitem[Belsley et~al.(2004)Belsley, Kuh, and Welsch]{belsleyEtAl04}
D.A. Belsley, E~Kuh, and R.~E. Welsch.
\newblock \emph{Regression Diagnostics: Identifying Influential Data and
  Sources of Colinearity}.
\newblock Wiley-Interscience, Hoboken, NJ, 2004.

\bibitem[Berger et~al.(2014)Berger, Minarik, Liuzzi, Hummel, and
  Sauseng]{bergerEtAl14}
Barbara Berger, Tamas Minarik, Gianpiero Liuzzi, Friedhelm~C Hummel, and Paul
  Sauseng.
\newblock {EEG} oscillatory phase-dependent markers of corticospinal
  excitability in the resting brain.
\newblock \emph{BioMed research international}, 2014, 2014.

\bibitem[Besle et~al.(2011)Besle, Schevon, Mehta, Lakatos, Goodman, McKhann,
  Emerson, and Schroeder]{besleEtAl11}
J.~Besle, C.A. Schevon, A.D. Mehta, P.~Lakatos, R.R. Goodman, G.M. McKhann,
  R.G. Emerson, and C.E. Schroeder.
\newblock Tuning of the human neocortex to the temporal dynamics of attended
  events.
\newblock \emph{The Journal of Neuroscience}, 31\penalty0 (9):\penalty0
  3176--3185, 2011.

\bibitem[Bishop(2006)]{bishop06}
C.M. Bishop.
\newblock \emph{Pattern recognition and machine learning}.
\newblock Springer, New York, NY, 2006.

\bibitem[Bishop(1932)]{bishop32}
G.H. Bishop.
\newblock Cyclic changes in excitability of the optic pathway of the rabbit.
\newblock \emph{American Journal of Physiology--Legacy Content}, 103\penalty0
  (1):\penalty0 213--224, 1932.

\bibitem[Bonnefond and Jensen(2012)]{bonnefondAndJensen12}
Mathilde Bonnefond and Ole Jensen.
\newblock Alpha oscillations serve to protect working memory maintenance
  against anticipated distracters.
\newblock \emph{Current biology}, 22\penalty0 (20):\penalty0 1969--1974, 2012.

\bibitem[Bonnefond and Jensen(2015)]{bonnefondAndJensen15}
Mathilde Bonnefond and Ole Jensen.
\newblock Gamma activity coupled to alpha phase as a mechanism for top-down
  controlled gating.
\newblock \emph{PloS one}, 10\penalty0 (6):\penalty0 e0128667, 2015.

\bibitem[Bosman et~al.(2009)Bosman, Womelsdorf, Desimone, and
  Fries]{bosmanEtAl09}
Conrado~A Bosman, Thilo Womelsdorf, Robert Desimone, and Pascal Fries.
\newblock A microsaccadic rhythm modulates gamma-band synchronization and
  behavior.
\newblock \emph{The Journal of Neuroscience}, 29\penalty0 (30):\penalty0
  9471--9480, 2009.

\bibitem[Botwinick and Brinley(1962)]{botwinickAndBrinley62}
J.~Botwinick and J.F. Brinley.
\newblock An analysis of set in relation to reaction time.
\newblock \emph{Journal of Experimental Psychology}, 63\penalty0 (6):\penalty0
  568--574, 1962.

\bibitem[Bretz et~al.(2010)Bretz, Hothorn, and Westfall]{bretzEtAl10}
F.~Bretz, T.~Hothorn, and P.~Westfall.
\newblock \emph{Multiple comparisons using \uppercase{R}}.
\newblock Chapman \& Hall/CRC Press, Boca Raton, FL, 2010.

\bibitem[Brunia and van Boxtel(2001)]{bruniaAndVanBoxtel01}
C.H. Brunia and G.J. van Boxtel.
\newblock Wait and see.
\newblock \emph{International Journal of Psychophysiology}, 43\penalty0
  (1):\penalty0 59--75, 2001.

\bibitem[Brunia and Boelhouwer(1988)]{bruniaAndBoelhouwer88}
C.H.M. Brunia and A.H.W. Boelhouwer.
\newblock Relfexes as a tool: a window in the central nervous system.
\newblock In P.K. Ackles, J.R. Jennigs, and M.G.H. Coles, editors,
  \emph{Advances in Psychophysiology}, volume~3, pages 1--67. JAI Press,
  Greenwich, CT, 1988.

\bibitem[Busch and VanRullen(2010)]{buschAndVanRullen10}
N.A. Busch and R.~VanRullen.
\newblock Spontaneous {EEG} oscillations reveal periodic sampling of visual
  attention.
\newblock \emph{Proceedings of the National Academy of Sciences}, 107\penalty0
  (37):\penalty0 16048--16053, September 2010.

\bibitem[Busch et~al.(2009)Busch, Dubois, and VanRullen]{buschEtAl09}
N.A. Busch, J.~Dubois, and R.~VanRullen.
\newblock The phase of ongoing {EEG} oscillations predicts visual preception.
\newblock \emph{The Journal of Neuroscience}, 29\penalty0 (24):\penalty0
  7869--7876, June 2009.

\bibitem[Callaway and Yeager(1960)]{callawayAndYeager60}
E.~Callaway and C.L. Yeager.
\newblock Relationship between reaction time and electroencephalographic alpha
  phase.
\newblock \emph{Science}, 132\penalty0 (3441):\penalty0 1765--1766, 1960.

\bibitem[Capotosto et~al.(2009)Capotosto, Babiloni, Romani, and
  Corbetta]{capotostoEtAl09}
Paolo Capotosto, Claudio Babiloni, Gian~Luca Romani, and Maurizio Corbetta.
\newblock Frontoparietal cortex controls spatial attention through modulation
  of anticipatory alpha rhythms.
\newblock \emph{The Journal of Neuroscience}, 29\penalty0 (18):\penalty0
  5863--5872, 2009.

\bibitem[Ceponiene et~al.(2008)Ceponiene, Westerfield, Torki, and
  Townsend]{ceponieneEtAl08}
R.~Ceponiene, M.~Westerfield, M.~Torki, and J.~Townsend.
\newblock Modality-specificity of sensory aging in vision and audition:
  evidence from event-related potentials.
\newblock \emph{Brain Research}, 1215:\penalty0 53--68, 2008.

\bibitem[Chakravarthi and VanRullen(2012)]{chakravarthiAndVanRullen12}
R.~Chakravarthi and R.~VanRullen.
\newblock Conscious updating is a rhythmic process.
\newblock \emph{Proceedings of the National Academy of Sciences}, 109\penalty0
  (26):\penalty0 10599--604, 2012.

\bibitem[Cogan and Poeppel(2011)]{coganAndPoeppel11}
Gregory~B Cogan and David Poeppel.
\newblock A mutual information analysis of neural coding of speech by
  low-frequency {MEG} phase information.
\newblock \emph{Journal of neurophysiology}, 106\penalty0 (2):\penalty0
  554--563, 2011.

\bibitem[Correa(2010)]{correa10}
A.~Correa.
\newblock Enhancing behavioural performance by visual temporal orienting.
\newblock In A.C. Nobre and J.T. Coull, editors, \emph{Attention and Time},
  chapter~3, pages 359--370. Oxford University Press, Oxford, UK, 2010.

\bibitem[Correa et~al.(2005)Correa, {n}ez, and Tudela]{correaEtAl05}
\'{A}. Correa, J.~Lup\'{a}\ {n}ez, and P.~Tudela.
\newblock Attentional preparation based on temporal expectancy modulates
  processing at the perceptual-level.
\newblock \emph{Psychonomic Bulletin and Review}, 12\penalty0 (2):\penalty0
  328--334, 2005.

\bibitem[Correa et~al.(2010)Correa, Cappucci, Nobre, and {n}ez]{correaEtAl10}
\'{A}. Correa, P.~Cappucci, A.C. Nobre, and J.~Lupia\ {n}ez.
\newblock The two sides of temporal orienting: facilitating perceptual
  selection, disrupting response selection.
\newblock \emph{Experimental Psychology}, 57\penalty0 (2):\penalty0 142--148,
  2010.

\bibitem[Coull(2009)]{coull09}
J.T. Coull.
\newblock Neural substrates of mounting temporal expectation.
\newblock \emph{PLoS Biol}, 7\penalty0 (8):\penalty0 e1000166, 2009.

\bibitem[Coull and Nobre(1998)]{coullAndNobre98}
J.T. Coull and A.C. Nobre.
\newblock Where and when to pay attention: the neural systems for directing
  attention to spatial locations and to time intervals revealed by both
  \uppercase{PET} and f\uppercase{MRI}.
\newblock \emph{Journal of Neuroscience}, 18\penalty0 (18):\penalty0
  7426--7435, 1998.

\bibitem[Coull et~al.(2000)Coull, Frith, B{\"u}chel, and Nobre]{coullEtAl00}
J.T. Coull, C.D. Frith, C~B{\"u}chel, and A.C. Nobre.
\newblock Orienting attention in time: behavioural and neuroanatomical
  distinction between exogenous and endogenous shifts.
\newblock \emph{Neuropsychologia}, 38\penalty0 (6):\penalty0 808--819, 2000.

\bibitem[Coull et~al.(2004)Coull, Vidal, Nazarian, and Macar]{coullEtAl04}
J.T. Coull, F.~Vidal, B.~Nazarian, and F.~Macar.
\newblock Functional anatomy of the attentional modulation of time estimation.
\newblock \emph{Science}, 303\penalty0 (5663):\penalty0 1506--1508, 2004.

\bibitem[Cravo et~al.(2011)Cravo, Rohenkohl, Wyart, and Nobre]{cravoEtAl11}
A.M. Cravo, G.~Rohenkohl, V.~Wyart, and A.C. Nobre.
\newblock Endogenous modulation of low frequency oscillations by temporal
  expectations.
\newblock \emph{Journal of Neurophysiology}, 106\penalty0 (6):\penalty0
  2964--2972, 2011.

\bibitem[Cravo et~al.(2013)Cravo, Rohenkohl, Wyart, and Nobre]{cravoEtAl13}
A.M. Cravo, G.~Rohenkohl, V.~Wyart, and A.C. Nobre.
\newblock Temporal expectation enhances contrast sensitivity by phase
  entrainment of low-frequency oscillations in visual cortex.
\newblock \emph{Journal of Neuroscience}, 33\penalty0 (9):\penalty0 4002--4010,
  2013.

\bibitem[Cravo et~al.(2015)Cravo, Santos, Reyes, Caetano, and
  Claessens]{cravoEtAl15}
Andre~Mascioli Cravo, Karin~Moreira Santos, Marcelo~Bussotti Reyes,
  Marcelo~Salvador Caetano, and Peter~ME Claessens.
\newblock Visual causality judgments correlate with the phase of alpha
  oscillations.
\newblock \emph{Journal of cognitive neuroscience}, 2015.

\bibitem[Cui et~al.(2000)Cui, Egkher, Huter, Lang, Lindinger, and
  Deecke]{cuiEtAl00}
R.Q. Cui, A.~Egkher, D.~Huter, W.~Lang, G.~Lindinger, and L.~Deecke.
\newblock High resolution spatio-temporal analysis of the contingent neg- ative
  variation in simple or motor complex motor tasks and a non-motor task.
\newblock \emph{Clinical Neurophysiology}, 111:\penalty0 1847--1859, 2000.

\bibitem[Delorme and Makeig(2004)]{delormeAndMakeig04}
A.~Delorme and S.~Makeig.
\newblock {EEGLAB}: an open source toolbox for analysis of single-trial {EEG}
  dynamics including independent component analsys.
\newblock \emph{Journal of Neuroscience Methods}, 134\penalty0 (1):\penalty0
  9--21, 2004.

\bibitem[Delorme et~al.(2007)Delorme, Sejnowski, and Makeig]{delormeEtAl07}
A.~Delorme, T.~Sejnowski, and S.~Makeig.
\newblock Enhanced detection of artifacts in {EEG} data using higher-order
  statistics and independent component analysis.
\newblock \emph{Neuroimage}, 34\penalty0 (4):\penalty0 1443--1449, 2007.

\bibitem[Delorme et~al.(2012)Delorme, Palmer, Onton, Oostenveld, and
  Makeig]{delormeEtAl12}
A.~Delorme, J.~Palmer, J.~Onton, R.~Oostenveld, and S.~Makeig.
\newblock Independent {EEG} sources are dipolar.
\newblock \emph{PLoS One}, 7\penalty0 (2), 2012.

\bibitem[Delorme et~al.(2015)Delorme, Miyakoshi, Jung, and
  Makeig]{delormeEtAl15}
A.~Delorme, M.~Miyakoshi, T.P. Jung, and S.~Makeig.
\newblock Grand average erp-image plotting and statistics: A method for
  comparing variability in event-related single-trial {EEG} activities across
  subjects and conditions.
\newblock \emph{Journal of Neuroscience Methods}, 250:\penalty0 3--6, 2015.

\bibitem[Doelling et~al.(2014)Doelling, Arnal, Ghitza, and
  Poeppel]{doellingEtAl14}
Keith~B Doelling, Luc~H Arnal, Oded Ghitza, and David Poeppel.
\newblock Acoustic landmarks drive delta--theta oscillations to enable speech
  comprehension by facilitating perceptual parsing.
\newblock \emph{Neuroimage}, 85:\penalty0 761--768, 2014.

\bibitem[Drewes and VanRullen(2011)]{drewesAndVanRullen11}
J.~Drewes and R.~VanRullen.
\newblock This is the rhythm of your eyes: the phase of ongoing
  electroencephalogram oscillations modulates saccadic reaction time.
\newblock \emph{Journal of Neurosciences}, 31\penalty0 (12):\penalty0
  4698--4708, 2011.

\bibitem[Dustman and Beck(1965)]{dustmanAndBeck65}
R.E. Dustman and E.C. Beck.
\newblock Phase of alpha brain waves, reaction time and visually evoked
  potentials.
\newblock \emph{Electroencephalography and Clinical Neurophysiology},
  18:\penalty0 433--440, 1965.

\bibitem[Efron and Tibshirani(1993)]{efronAndTibshirani93}
B.~Efron and R.J. Tibshirani.
\newblock \emph{An introduction to the bootstrap}.
\newblock Chapman \& Hall, New York, NY, 1993.

\bibitem[Fiebelkorn et~al.(2011)Fiebelkorn, Foxe, Butler, Mercier, Snyder, and
  Molholm]{fiebelkornEtAl11}
I.C. Fiebelkorn, J.J. Foxe, J.S. Butler, M.R. Mercier, A.C. Snyder, and
  S.~Molholm.
\newblock Ready, set, reset: stimulus-locked periodicity in behavioral
  performance demonstrates the consequences of cross-sensory phase reset.
\newblock \emph{The Journal of Neuroscience}, 31\penalty0 (27):\penalty0
  9971--9981, July 2011.

\bibitem[Fiebelkorn et~al.(2013)Fiebelkorn, Snyder, Mercier, Butler, Molholm,
  and Foxe]{fiebelkornEtAl13}
I.C. Fiebelkorn, A.C. Snyder, M.R. Mercier, J.S. Butler, S.~Molholm, and J.J.
  Foxe.
\newblock Cortical cross-frequency coupling predicts perceptual outcomes.
\newblock \emph{Neuroimage}, 69:\penalty0 126--137, 2013.

\bibitem[Frigge et~al.(1989)Frigge, Hoaglin, and Iglewicz]{friggeEtAl89}
M.~Frigge, D.C. Hoaglin, and B.~Iglewicz.
\newblock Some implementations of the boxplot.
\newblock \emph{American statistical association}, 43\penalty0 (1):\penalty0
  50--54, 1989.

\bibitem[Fujioka et~al.(2012)Fujioka, Trainor, Large, and Ross]{fujiokaEtAl12}
Takako Fujioka, Laurel~J Trainor, Edward~W Large, and Bernhard Ross.
\newblock Internalized timing of isochronous sounds is represented in
  neuromagnetic beta oscillations.
\newblock \emph{The Journal of Neuroscience}, 32\penalty0 (5):\penalty0
  1791--1802, 2012.

\bibitem[Gaillard(1977)]{gaillard77}
A.W. Gaillard.
\newblock The late \uppercase{CNV} wave: preparation versus expectancy.
\newblock \emph{Psychophysiology}, 14:\penalty0 563--568, 1977.

\bibitem[Gaillard(1978)]{gaillard78}
A.W. Gaillard.
\newblock \emph{Slow brain potential preceding task performance}.
\newblock Academische Press, Amsterdam, 1978.

\bibitem[Galambos et~al.(1981)Galambos, Makeig, and Talmachoff]{galambosEtAl81}
Robert Galambos, Scott Makeig, and Peter~J Talmachoff.
\newblock A 40-hz auditory potential recorded from the human scalp.
\newblock \emph{Proceedings of the National Academy of Sciences}, 78\penalty0
  (4):\penalty0 2643--2647, 1981.

\bibitem[Geman et~al.(1992)Geman, Bienenstock, and Doursat]{gemanEtAl92}
S.~Geman, E.~Bienenstock, and R.~Doursat.
\newblock Neural networks and the bias/variance dilemma.
\newblock \emph{Neural Computation}, 4:\penalty0 1--58, 1992.

\bibitem[Giraud and Poeppel(2012)]{giraudAndPoeppel12}
Anne-Lise Giraud and David Poeppel.
\newblock Cortical oscillations and speech processing: emerging computational
  principles and operations.
\newblock \emph{Nature neuroscience}, 15\penalty0 (4):\penalty0 511--517, 2012.

\bibitem[Gomez et~al.(2001)Gomez, Delinte, Vaquero, Cardoso, Vazquez,
  Crommelynck, and Roucoux]{gomezEtAl01}
C.M. Gomez, A.~Delinte, E.~Vaquero, M.J. Cardoso, M.~Vazquez, M.~Crommelynck,
  and A.~Roucoux.
\newblock Current source density analysis of cnv during temporal gap paradigm.
\newblock \emph{Brain Topography}, 13:\penalty0 149--159, 2001.

\bibitem[Gomez et~al.(2003)Gomez, Marco, and Grau]{gomezEtAl03}
C.M. Gomez, J.~Marco, and C.~Grau.
\newblock Preparatory visuo-motor cortical network of the contingent negative
  variation estimated by current density.
\newblock \emph{Neuroimage}, 20:\penalty0 216--224, 2003.

\bibitem[Gomez-Ramirez et~al.(2011)Gomez-Ramirez, Kelly, Molholm, Sehatpour,
  Schwartz, and Foxe]{gomezRamirezEtAl11}
M.~Gomez-Ramirez, S.P. Kelly, S.~Molholm, P~Sehatpour, T.H. Schwartz, and J.J.
  Foxe.
\newblock Oscillatory sensory selection mechanism during intersensory attention
  to rhythmic auditory and visual inputs: a human electrocortigraphic
  investigation.
\newblock \emph{The Journal of Neuroscience}, 31\penalty0 (50):\penalty0
  18556--18567, 2011.

\bibitem[Gray et~al.(2015)Gray, Frey, Wilson, and Foxe]{grayEtAl15}
Michael~J Gray, Hans-Peter Frey, Tommy~J Wilson, and John~J Foxe.
\newblock Oscillatory recruitment of bilateral visual cortex during spatial
  attention to competing rhythmic inputs.
\newblock \emph{The Journal of Neuroscience}, 35\penalty0 (14):\penalty0
  5489--5503, 2015.

\bibitem[Grondin(2001)]{grondin01}
S.~Grondin.
\newblock From physical time to the first and second moments of psychological
  time.
\newblock \emph{Psychological Bulletin}, 127:\penalty0 22--44, 2001.

\bibitem[Gross et~al.(2013)Gross, Hoogenboom, Thut, Schyns, Panzeri, Belin, and
  Garrod]{grossEtAl13}
J.~Gross, N.~Hoogenboom, G.~Thut, P.~Schyns, S.~Panzeri, P.~Belin, and
  S.~Garrod.
\newblock Speech rhythms and multiplexed oscillatory sensory coding in the
  human brain.
\newblock \emph{PLoS Biology}, 11\penalty0 (12), 2013.
\newblock \doi{10.1371/journal.pbio.1001752}.

\bibitem[Hackley and Valle-Inclan(1999)]{hackleyAndValleInclan99}
S.A. Hackley and F.~Valle-Inclan.
\newblock Accesory stimulus effects on response selection: Does arousal speed
  decision making?
\newblock \emph{Journal of Cognitive Neuroscience}, 11:\penalty0 321--329,
  1999.

\bibitem[Hamm et~al.(2010)Hamm, Dyckman, Ethridge, McDowell, and
  Clementz]{hammEtAl10}
J.P. Hamm, K.A. Dyckman, L.E. Ethridge, J.E. McDowell, and B.A. Clementz.
\newblock Preparatory activations across a distributed cortical network
  determine production of express saccades in humans.
\newblock \emph{Journal of Neurosciences}, 30\penalty0 (21):\penalty0
  7350--7357, 2010.

\bibitem[Hampton and O'Doherty(2007)]{hamptonAndODoherty07}
A.N. Hampton and J.P. O'Doherty.
\newblock Decoding the neural substrates of reward-related decision making with
  functional mri.
\newblock \emph{Proceedings of the National Academy of Science}, 104:\penalty0
  1377--1382, 2007.

\bibitem[H\"{a}ndel and Haarmeier(2009)]{handelAndHaarmeier09}
B.~H\"{a}ndel and T.~Haarmeier.
\newblock Cross-frequency coupling of brain oscillations indicates the success
  in visual motion discrimination.
\newblock \emph{Neuroimage}, 45\penalty0 (3):\penalty0 1040--1046, 2009.

\bibitem[Hanslmayr et~al.(2007)Hanslmayr, Aslan, Staudigl, Klimesch, Herrmann,
  and Bauml]{hanslmayrEtAl07}
S.~Hanslmayr, A.~Aslan, T.~Staudigl, W.~Klimesch, C.S. Herrmann, and K.~Bauml.
\newblock Prestimulis oscillations predict visual perception performance
  between and within subjects.
\newblock \emph{Neuroimage}, 37:\penalty0 1465--1473, 2007.

\bibitem[Hanslmayr et~al.(2013)Hanslmayr, Volberg, Wimber, Dalal, and
  Greenlee]{hanslmayrEtAl13}
Simon Hanslmayr, Gregor Volberg, Maria Wimber, Sarang~S Dalal, and Mark~W
  Greenlee.
\newblock Prestimulus oscillatory phase at 7 hz gates cortical information flow
  and visual perception.
\newblock \emph{Current Biology}, 23\penalty0 (22):\penalty0 2273--2278, 2013.

\bibitem[Hastie et~al.(2009)Hastie, Tibshirani, and Friedman]{hastieEtAl09}
T.~Hastie, R.~Tibshirani, and J.~Friedman.
\newblock \emph{The Elements of Statistical Learning}.
\newblock Springer, New York, NY, 2009.

\bibitem[Haxby et~al.(2001)Haxby, Gobbini, Furey, Ishai, Schouten, and
  Pietrini]{haxbyEtAl01}
J.V. Haxby, M.I. Gobbini, M.L. Furey, A.~Ishai, J.L. Schouten, and P.~Pietrini.
\newblock Distributed and overlapping representations of faces and objects in
  ventral temporal cortex.
\newblock \emph{Science}, 293:\penalty0 2425--2430, 2001.

\bibitem[Haynes and Rees(2005)]{haynesAndRees05}
J.D. Haynes and G.~Rees.
\newblock Predicting the orientation of invisible stimuli from activity in
  human primary visual cortex.
\newblock \emph{Nature Neuroscience}, 8:\penalty0 686--691, 2005.

\bibitem[Henry and Obleser(2012)]{henryAndObleser12}
Molly~J Henry and Jonas Obleser.
\newblock Frequency modulation entrains slow neural oscillations and optimizes
  human listening behavior.
\newblock \emph{Proceedings of the National Academy of Sciences}, 109\penalty0
  (49):\penalty0 20095--20100, 2012.

\bibitem[Hickok et~al.(2015)Hickok, Farahbod, and Saberi]{hickokEtAl15}
Gregory Hickok, Haleh Farahbod, and Kourosh Saberi.
\newblock The rhythm of perception entrainment to acoustic rhythms induces
  subsequent perceptual oscillation.
\newblock \emph{Psychological science}, page 0956797615576533, 2015.

\bibitem[Howard and Poeppel(2010)]{howardAndPoeppel10}
Mary~F Howard and David Poeppel.
\newblock Discrimination of speech stimuli based on neuronal response phase
  patterns depends on acoustics but not comprehension.
\newblock \emph{Journal of neurophysiology}, 104\penalty0 (5):\penalty0
  2500--2511, 2010.

\bibitem[Howard and Poeppel(2012)]{howardAndPoeppel12}
Mary~F Howard and David Poeppel.
\newblock The neuromagnetic response to spoken sentences: co-modulation of
  theta band amplitude and phase.
\newblock \emph{Neuroimage}, 60\penalty0 (4):\penalty0 2118--2127, 2012.

\bibitem[Hultin et~al.(1996)Hultin, Rossini, Romani, Hogstedt, Tecchio, and
  Pizzella]{hultinEtAl96}
L.~Hultin, P.~Rossini, G.L. Romani, P.~Hogstedt, F.~Tecchio, and V.~Pizzella.
\newblock Neuromagnetic localization of the late component of the cognitive
  negative variation.
\newblock \emph{Electroencephalography and Clinical Neurophysiology},
  98:\penalty0 435--448, 1996.

\bibitem[Jarcho(1949)]{jarcho49}
L.W. Jarcho.
\newblock Excitability of cortical afferent systems during barbiturate
  anesthesia.
\newblock \emph{Journal of neurophysiology}, 12\penalty0 (6):\penalty0
  447--457, 1949.

\bibitem[Jung et~al.(2001)Jung, Makeig, Westerfield, Townsend, Courchesne, and
  Sejnowski]{jungEtAl01}
T.P. Jung, S.~Makeig, M.~Westerfield, J.~Townsend, E.~Courchesne, and T.J.
  Sejnowski.
\newblock Analysis and visualization of single-trial event-related potentials.
\newblock \emph{Human brain mapping}, 14\penalty0 (3):\penalty0 166--185, 2001.

\bibitem[Kamitani and Tong(2005)]{kamitaniAndTong05}
Y.~Kamitani and F.~Tong.
\newblock Decoding the visual subjective contents of human brain.
\newblock \emph{Nature Neuroscience}, 8\penalty0 (5):\penalty0 679--685, 2005.

\bibitem[Kay et~al.(2008)Kay, Naseralis, Prenger, and Gallant]{kayEtAl08}
K.N. Kay, T.~Naseralis, R.J. Prenger, and J.L. Gallant.
\newblock Identifying natural images from human brain activity.
\newblock \emph{Nature}, 452\penalty0 (7185):\penalty0 352--355, 2008.

\bibitem[Kayser et~al.(2008)Kayser, Petkov, and Logothetis]{kayserEtAl08}
C.~Kayser, C.L. Petkov, and N.K. Logothetis.
\newblock Visual modulation in auditory cortex.
\newblock \emph{Cerebral Cortex}, 18\penalty0 (7):\penalty0 1560--1574, 2008.

\bibitem[Kutner et~al.(2005)Kutner, Nachtsheim, Neter, and Li]{kutnerEtAl05}
M.H. Kutner, C.J. Nachtsheim, J.~Neter, and W.~Li.
\newblock \emph{Applied linear statistical models}.
\newblock McGraw-Hill/Irwin, New York, NY, 2005.

\bibitem[Lachaux et~al.(1999)Lachaux, Rodriguez, Martinerie, and
  Varela]{lachauxEtAl99}
J.P. Lachaux, E.~Rodriguez, J.~Martinerie, and F.J. Varela.
\newblock Measuring phase synchrony in brain signals.
\newblock \emph{Human brain mapping}, 8\penalty0 (4):\penalty0 194--208, 1999.

\bibitem[Lakatos et~al.(2005)Lakatos, Shah, Knuth, Ulbert, Kamos, and
  Schroeder]{lakatosEtAl05}
P.~Lakatos, A.S. Shah, K.H. Knuth, I.~Ulbert, G.~Kamos, and C.E. Schroeder.
\newblock An oscillatory hierarchy controlling neural excitability and stimulus
  processing in the auditory cortex.
\newblock \emph{J. Neurophysiology}, 94\penalty0 (3):\penalty0 1904--1911,
  2005.

\bibitem[Lakatos et~al.(2008)Lakatos, Karmos, Mehta, Ulbert, and
  Schroeder]{lakatosEtAl08}
P.~Lakatos, G.~Karmos, A.D. Mehta, I.~Ulbert, and C.E. Schroeder.
\newblock Entrainment of neuronal oscillations as a mechanism of attentional
  selection.
\newblock \emph{Science}, 320:\penalty0 110--113, 2008.

\bibitem[Lakatos et~al.(2009)Lakatos, O'Connell, Barczak, Mills, Javitt, and
  Schroeder]{lakatosEtAl09}
P.~Lakatos, M.~N. O'Connell, A.~Barczak, A.~Mills, D.C. Javitt, and C.E.
  Schroeder.
\newblock The leading sense: supramodal control of neurophysiological contex by
  attention.
\newblock \emph{Neuron}, 64:\penalty0 419--430, 2009.

\bibitem[Lakatos et~al.(2013)Lakatos, Musacchia, O'connel, Falchier, Javitt,
  and Schroeder]{lakatosEtAl13}
P.~Lakatos, G.~Musacchia, M.N. O'connel, A.Y. Falchier, D.C. Javitt, and C.E.
  Schroeder.
\newblock The spectrotemporal filter mechanism of auditory selective attention.
\newblock \emph{Neuron}, 77\penalty0 (4):\penalty0 750--761, 2013.

\bibitem[Lamport(1994)]{lamport94}
L.~Lamport.
\newblock \emph{Latex}.
\newblock Addison-Wesley, 1994.

\bibitem[Lancaster et~al.(1997)Lancaster, Rainey, Summerlin, Freitas, Fox,
  Evans, Toga, and Mazziotta]{lancasterEtAl97}
J.L. Lancaster, L.~Rainey, J.L. Summerlin, C.S. Freitas, P.T. Fox, A.C. Evans,
  A.W. Toga, and J.C. Mazziotta.
\newblock Automated labeling of the human brain: A preliminary report on the
  development and evaluation of a forward-transform method.
\newblock \emph{Human Brain Mapping}, 5:\penalty0 238--242, 1997.

\bibitem[Lancaster et~al.(2000)Lancaster, Woldorff, Parsons, Liotti, Freitas,
  Rainey, Kochunov, Nickerson, Mikiten, and Fox]{lancasterEtAl00}
J.L. Lancaster, M.G. Woldorff, L.M. Parsons, M.~Liotti, C.S. Freitas,
  L.~Rainey, P.V. Kochunov, D.~Nickerson, S.A. Mikiten, and P.T. Fox.
\newblock Automated talairach atlas labels for functional brain mapping.
\newblock \emph{Human Brain Mapping}, 10:\penalty0 12--131, 2000.

\bibitem[Lange(2009)]{lange09}
K.~Lange.
\newblock Brain correlates of early auditory processing are attenuated by
  expectations for time and pitch.
\newblock \emph{Brain and Cognition}, 69\penalty0 (1):\penalty0 127--137, 2009.

\bibitem[Lansing(1957)]{lansing57}
R.W. Lansing.
\newblock Relation of brain and tremor rhythms to visual reaction time.
\newblock \emph{Electroencephalography and clinical neurophysiology},
  9\penalty0 (3):\penalty0 497--504, 1957.

\bibitem[Lindsley(1952)]{lindsley52}
D.B. Lindsley.
\newblock Psychological phenomena and the electroencephalogram.
\newblock \emph{Electroencephalography and clinical neurophysiology},
  4\penalty0 (4):\penalty0 443--456, 1952.

\bibitem[Liu et~al.(2014)Liu, Bengson, Huang, Mangun, and Ding]{liuEtAl14}
Yuelu Liu, Jesse Bengson, Haiqing Huang, George~R Mangun, and Mingzhou Ding.
\newblock Top-down modulation of neural activity in anticipatory visual
  attention: control mechanisms revealed by simultaneous {EEG-fMRI}.
\newblock \emph{Cerebral Cortex}, page bhu204, 2014.

\bibitem[Low et~al.(1966)Low, Borda, Frost, and Kellaway]{lowEtAl66}
M.D. Low, R.~Borda, J.~Frost, and P.~Kellaway.
\newblock Surface negative slow potential shift associated with conditioning in
  man.
\newblock \emph{Neurology}, 16:\penalty0 771--782, 1966.

\bibitem[Luo et~al.(2013)Luo, Tian, Song, Shou, and Poeppel]{luoEtAl13}
H.~Luo, X.~Tian, K.~Song, K.~Shou, and D.~Poeppel.
\newblock Neural response phase tracks how listeners learn new acoustic
  representations.
\newblock \emph{Curent Biology}, 23:\penalty0 968--974, 2013.

\bibitem[Luo and Poeppel(2007)]{luoAndPoeppel07}
Huan Luo and David Poeppel.
\newblock Phase patterns of neuronal responses reliably discriminate speech in
  human auditory cortex.
\newblock \emph{Neuron}, 54\penalty0 (6):\penalty0 1001--1010, 2007.

\bibitem[Macar and Vidal(2003)]{macarAndVidal03}
F.~Macar and F.~Vidal.
\newblock The \uppercase{CNV} peak: an index of decision making and temporal
  memory.
\newblock \emph{Psychophysiology}, 40:\penalty0 950--954, 2003.

\bibitem[Makeig et~al.(1996)Makeig, Bell, Jung, and Sejnowski]{makeigEtAl96}
S.~Makeig, A.J. Bell, T.P. Jung, and T.J. Sejnowski.
\newblock Independent component analysis of electroencephalographic data.
\newblock In D.S. Touretzky, M.C. Mozer, and M.E. Hasselmo, editors,
  \emph{Advances in neural information processing sytems 8}, Cambridge MA,
  1996. MIT Press.

\bibitem[Makeig et~al.(1999)Makeig, Westerfield, Jung, Covington, Townsend,
  Sejnowski, and Courchesne]{makeigEtAl99}
S.~Makeig, M.~Westerfield, T.P. Jung, J.~Covington, J.~Townsend, T.~J
  Sejnowski, and E.~Courchesne.
\newblock Functionally independent components of the late positive
  event-related potential during visual spatial attention.
\newblock \emph{The journal of neuroscience}, 19\penalty0 (7):\penalty0
  2665--2680, 1999.

\bibitem[Makeig et~al.(2002)Makeig, Westerfield, Jung, Enghoff, Townsend,
  Couchesne, and Sejnowski]{makeigEtAl02}
S.~Makeig, M.~Westerfield, T.P. Jung, S.~Enghoff, J.~Townsend, E.~Couchesne,
  and T.J. Sejnowski.
\newblock Dynamic brain sources of visual evoked responses.
\newblock \emph{Science}, 295\penalty0 (5555):\penalty0 690--694, 2002.

\bibitem[Makeig et~al.(2004)Makeig, Debener, Onton, and Delorme]{makeigEtAl04}
S.~Makeig, .~Debener, J.~Onton, and A.~Delorme.
\newblock Mining event-related brain dynamics.
\newblock \emph{Trends in cognitive sciences}, 8\penalty0 (5):\penalty0
  204--210, 2004.

\bibitem[Mardia(1972)]{mardia72}
K.V. Mardia.
\newblock \emph{Statistics of directional data}.
\newblock Acadmic Press, New York, NY, 1972.

\bibitem[Mathewson et~al.(2009)Mathewson, Gratton, Fabiani, Beck, and
  Ro]{mathewsonEtAl09}
K.E. Mathewson, G.~Gratton, M.~Fabiani, D.M. Beck, and T.~Ro.
\newblock To see or not to see: prestimulus $alpha$ phase predicts visual
  awareness.
\newblock \emph{The Journal of Neuroscience}, 29\penalty0 (9):\penalty0
  2725--2732, 2009.

\bibitem[Mathewson et~al.(2011)Mathewson, Lleras, Beck, Fabiani, Ro, and
  Gratton]{mathewsonEtAl11}
K.E. Mathewson, A.~Lleras, D.M. Beck, M.~Fabiani, T.~Ro, and G.~Gratton.
\newblock Pulsed out awareness: {EEG} alpha oscillations represent a
  pulsed-inhibition of ongoing cortical processing.
\newblock \emph{Frontiers in Psychology}, 2:\penalty0 1--15, 2011.

\bibitem[Mathewson et~al.(2012)Mathewson, Prudhomme, Fabiani, Beck, Lleras, and
  Gratton]{mathewsonEtAl12}
K.E. Mathewson, C.~Prudhomme, M.~Fabiani, D.M. Beck, A.~Lleras, and G.~Gratton.
\newblock Making waves in the stream of consciousness: entraining oscillations
  in {EEG} alpha and fluctuations in visual awareness with rhythmic visual
  stimulation.
\newblock \emph{Journal of cognitive neuroscience}, 24\penalty0 (12):\penalty0
  2321--2333, 2012.

\bibitem[Mathewson et~al.(2014)Mathewson, Beck, Ro, Maclin, Low, Fabiani, and
  Gratton]{mathewsonEtAl14}
Kyle~E Mathewson, Diane~M Beck, Tony Ro, Edward~L Maclin, Kathy~A Low, Monica
  Fabiani, and Gabriele Gratton.
\newblock Dynamics of alpha control: preparatory suppression of posterior alpha
  oscillations by frontal modulators revealed with combined {EEG} and
  event-related optical signal.
\newblock \emph{Journal of cognitive neuroscience}, 26\penalty0 (10):\penalty0
  2400--2415, 2014.

\bibitem[MATLAB(2013)]{matlab13}
MATLAB.
\newblock \emph{version 8.1.0.604 (2013a)}.
\newblock The MathWorks Inc., Natick, Massachusetts, 2013.

\bibitem[McCallum and Walter(1968)]{mcCallumAndWalter68}
W.C. McCallum and W.G. Walter.
\newblock The effects of attention and distraction on the contingent negative
  variation in nomral and neurotic subjects.
\newblock \emph{Electroencephalography and Clinical Neurophysiology},
  25:\penalty0 319--329, 1968.

\bibitem[McLelland et~al.(2016)McLelland, Lavergne, and
  VanRullen]{mclellandEtAl16}
Douglas McLelland, Louisa Lavergne, and Rufin VanRullen.
\newblock The phase of ongoing eeg oscillations predicts the amplitude of
  peri-saccadic mislocalization.
\newblock \emph{Scientific Reports}, 6, 2016.

\bibitem[Mento et~al.(2013)Mento, Tarantino, Sarlo, and Bisiacchi]{mentoEtAl13}
G.~Mento, V.~Tarantino, M.~Sarlo, and P.S. Bisiacchi.
\newblock Automatic temporal expectation: a high-density event-related
  potential study.
\newblock \emph{PLOS ONE}, 8\penalty0 (5):\penalty0 1--11, 2013.

\bibitem[Metropolis and Ulam(1949)]{metropolisAndUlam49}
N.~Metropolis and S.~Ulam.
\newblock The monte carlo method.
\newblock \emph{Journal of the American Statistical Asociation}, 44\penalty0
  (247):\penalty0 335--41, 1949.

\bibitem[Milton and Pleydell-Pearce(2016)]{miltonAndPleydellPearce16}
A.~Milton and C.W. Pleydell-Pearce.
\newblock The phase of pre-stimulus alpha oscillations influences the visual
  perception of stimulus timing.
\newblock \emph{NeuroImage}, 133:\penalty0 53--61, 2016.

\bibitem[Miniussi et~al.(1999)Miniussi, Wilding, Coull, and
  Nobre]{miniussiEtAl99}
C.~Miniussi, E.L. Wilding, J.T. Coull, and A.C. Nobre.
\newblock Orienting attetion in time: Modulation of brain potentials.
\newblock \emph{Brain}, 122:\penalty0 1507--1518, 1999.

\bibitem[Mitchell et~al.(2008)Mitchell, Shinkareva, Carlson, Chang, Malave,
  Marson, and Just]{mitchellEtAl08}
T.M. Mitchell, S.V. Shinkareva, A.~Carlson, K.~Chang, V.L. Malave, R.A. Marson,
  and M.A. Just.
\newblock Predicting human brain activity associated with the meaning of nouns.
\newblock \emph{Science}, 320\penalty0 (5880):\penalty0 1191--1195, 2008.

\bibitem[Monto et~al.(2008)Monto, Palva, Voipio, and Palva]{montoEtAl08}
S.~Monto, S.~Palva, J.~Voipio, and J.M. Palva.
\newblock Very slow {EEG} fluctuations predict the dynamics of stimulus
  detection and oscillation amplitudes in humans.
\newblock \emph{Journal of Neuroscience}, 28\penalty0 (33), 2008.
\newblock \doi{10.1523/JNEUROSCI.1910-08.2008}.

\bibitem[Morillon et~al.(2012)Morillon, Li{\'e}geois-Chauvel, Arnal, B{\'e}nar,
  and Giraud]{morillonEtAl12}
Benjamin Morillon, Catherine Li{\'e}geois-Chauvel, Luc~H Arnal, Christian-G
  B{\'e}nar, and Anne-Lise Giraud.
\newblock Asymmetric function of theta and gamma activity in syllable
  processing: an intra-cortical study.
\newblock \emph{Frontiers in psychology}, 3, 2012.

\bibitem[M\"{u}ller-Gethmann et~al.(2003)M\"{u}ller-Gethmann, Ulrich, and
  Rinkenauer]{mullerGethmannEtAl03}
H.~M\"{u}ller-Gethmann, R.~Ulrich, and G.~Rinkenauer.
\newblock Locus of the effect of temporal preparation: Evidence from the
  lateralized readiness potential.
\newblock \emph{Psychophysiology}, 40:\penalty0 597--611, 2003.

\bibitem[Myers et~al.(2014)Myers, Stokes, Walther, and Nobre]{myersEtAl14}
Nicholas~E Myers, Mark~G Stokes, Lena Walther, and Anna~C Nobre.
\newblock Oscillatory brain state predicts variability in working memory.
\newblock \emph{The Journal of Neuroscience}, 34\penalty0 (23):\penalty0
  7735--7743, 2014.

\bibitem[N\"{a}\"{a}t\"{a}nen(1970)]{naatanen70}
R.~N\"{a}\"{a}t\"{a}nen.
\newblock The diminishing time-uncertainty with the lapse of time after the
  warning signal in reaction-time experiments with varying fore-periods.
\newblock \emph{Acta Psychologica}, 34:\penalty0 399--418, 1970.

\bibitem[Naccache et~al.(2002)Naccache, Blandin, and Dehaene]{naccacheEtAl02}
L.~Naccache, E.~Blandin, and S.~Dehaene.
\newblock Unconscious masked priming depends on temporal attention.
\newblock \emph{Psychological Science}, 13\penalty0 (5):\penalty0 416--424,
  2002.

\bibitem[Naselaris et~al.(2009)Naselaris, Prenger, Kay, Oliver, and
  Gallant]{naselarisEtAl09}
T.~Naselaris, R.J. Prenger, K.N. Kay, M.~Oliver, and J.L. Gallant.
\newblock Bayesian reconstruction of natural images from human brain activity.
\newblock \emph{Neuron}, 63:\penalty0 902--915, 2009.

\bibitem[Naselaris et~al.(2011)Naselaris, Kay, Nishimoto, and
  Gallant]{naselarisEtAl11}
T.~Naselaris, K.N. Kay, S.~Nishimoto, and J.L. Gallant.
\newblock Encoding and decoding in fmri.
\newblock \emph{Neuroimage}, 56:\penalty0 400--410, 2011.

\bibitem[Neuling et~al.(2012)Neuling, Rach, Wagner, Wolters, and
  Herrmann]{neulingEtAl12}
T.~Neuling, S.~Rach, S.~Wagner, C.H. Wolters, and C.S. Herrmann.
\newblock Good vibrations: oscillatory phase shapes perception.
\newblock \emph{Neuroimage}, 63\penalty0 (2):\penalty0 771--778, 2012.

\bibitem[Ng et~al.(2012)Ng, Schroeder, and Kayser]{ngEtAl12}
B.S.W. Ng, T.~Schroeder, and C.~Kayser.
\newblock A precluding but not ensuring role of entrained low-frequency
  oscillations for auditory perception.
\newblock \emph{The Journal of Neuroscience}, 32\penalty0 (35):\penalty0
  12268--12276, 2012.

\bibitem[Nicolas-Alonso and Gomez-Gil(2012)]{nicolasAlonsoAndGomezGil12}
L.~F. Nicolas-Alonso and J.~Gomez-Gil.
\newblock Brain computer interfaces, a review.
\newblock \emph{Sensors}, 12:\penalty0 1211--1279, 2012.

\bibitem[Niemi and Naatanen(1981)]{niemiAndNaatanen81}
P.~Niemi and R.~Naatanen.
\newblock Foreperiod and simple reaction time.
\newblock \emph{Psychololgical Bulletin}, 89\penalty0 (1):\penalty0 133--162,
  1981.

\bibitem[O'Connell et~al.(2011)O'Connell, Falchier, McGinnis, Schroeder, and
  Lakatos]{oconnellEtAl11}
M.N. O'Connell, A.~Falchier, T.~McGinnis, C.E. Schroeder, and P.~Lakatos.
\newblock Dual mechanism of neuronal ensemble inhibition in primary auditory
  cortex.
\newblock \emph{Neuron}, 69\penalty0 (4):\penalty0 805--817, 2011.

\bibitem[Onton and Makeig(2009)]{ontonAndMakeig09}
J.A. Onton and S.~Makeig.
\newblock High-frequency broadband modulation of electroencephalographic
  spectra.
\newblock \emph{Frontiers in human neuroscience}, 3:\penalty0 61, 2009.

\bibitem[Palmer et~al.(2007)Palmer, Kreutz-Delgado, Rao, and
  Makeig]{palmerEtAl07}
J.A. Palmer, K.~Kreutz-Delgado, B.D. Rao, and S.~Makeig.
\newblock Modeling and estimation of dependent subspaces with non-radially
  symmetric and skewed densities.
\newblock In M.E. Davies, C.J. James, A.A. Abdallah, and M.D. Plumbey, editors,
  \emph{Proceedings of the 7th International Symposium on Independent Component
  Analysis}, Lecture Notes in Computer Science. Springer, 2007.

\bibitem[Park et~al.(2015)Park, Ince, Schyns, Thut, and Gross]{parkEtAl15}
Hyojin Park, Robin~AA Ince, Philippe~G Schyns, Gregor Thut, and Joachim Gross.
\newblock Frontal top-down signals increase coupling of auditory low-frequency
  oscillations to continuous speech in human listeners.
\newblock \emph{Current Biology}, 2015.

\bibitem[Peelle et~al.(2013)Peelle, Gross, and Davis]{peelleEtAl13}
Jonathan~E Peelle, Joachim Gross, and Matthew~H Davis.
\newblock Phase-locked responses to speech in human auditory cortex are
  enhanced during comprehension.
\newblock \emph{Cerebral cortex}, 23\penalty0 (6):\penalty0 1378--1387, 2013.

\bibitem[Pernet et~al.(2011{\natexlab{a}})Pernet, Chauveau, Gaspar, and
  Rousselet]{pernetEtAl11b}
C.R. Pernet, N.~Chauveau, C.~Gaspar, and G.A. Rousselet.
\newblock {LIMO EEG}: a toolbox for hierarchical \uppercase{LI}near {MO}deling
  of \uppercase{E}lectro\uppercase{E}ncephalo\uppercase{G}raphic data.
\newblock \emph{Computational Intelligence and Neurosciences}, 2011,
  2011{\natexlab{a}}.

\bibitem[Pernet et~al.(2011{\natexlab{b}})Pernet, Sajda, and
  Rousselet]{pernetEtAl11a}
C.R. Pernet, P.~Sajda, and G.A. Rousselet.
\newblock Single-trial analyses: why bother?
\newblock \emph{Frontiers in Psychology}, 2, 2011{\natexlab{b}}.

\bibitem[Pernet et~al.(2013)Pernet, Wilcox, and Rousselet]{pernetEtAl13}
C.R. Pernet, R.~Wilcox, and G.A. Rousselet.
\newblock Robust correlation analyses: false positive power validation using a
  new open source matlab toolbox.
\newblock \emph{Frontiers in psychology}, 3, 2013.

\bibitem[Pfeuty et~al.(2003)Pfeuty, Ragot, and Pouthas]{pfeutyEtAl03}
M.~Pfeuty, R.~Ragot, and V.~Pouthas.
\newblock When time is up: \uppercase{CNV} time course differentiaties the
  roles of the hemispheres in the discrimination of short tone dudrations.
\newblock \emph{Experimental Brain Research}, 151:\penalty0 372--379, 2003.

\bibitem[Pfeuty et~al.(2005)Pfeuty, Ragot, and Pouthas]{pfeutyEtAl05}
M.~Pfeuty, R.~Ragot, and V.~Pouthas.
\newblock Relation between the \uppercase{CNV} and timing of an upcoming event.
\newblock \emph{Neuroscience Letters}, 382:\penalty0 106--111, 2005.

\bibitem[Polyn et~al.(2005)Polyn, Natu, Cohen, and Norman]{polynEtAl05}
S.M. Polyn, V.S. Natu, J.D. Cohen, and K.A. Norman.
\newblock Category-specific cortical activity precedes retrieval during memory
  search.
\newblock \emph{Science}, 310:\penalty0 1963--1966, 2005.

\bibitem[Praamstra et~al.(2006)Praamstra, Kourtis, Kwok, and
  Oostenveld]{praamstraEtAl06}
P.~Praamstra, D.~Kourtis, H.F. Kwok, and R.~Oostenveld.
\newblock Neuropphysiology of implicit timing in serial choice reaction-time
  performance.
\newblock \emph{Journal of Neurosciences}, 26:\penalty0 5448--5455, 2006.

\bibitem[{R Core Team}(2012)]{r12}
{R Core Team}.
\newblock \emph{R: A Language and Environment for Statistical Computing}.
\newblock R Foundation for Statistical Computing, Vienna, Austria, 2012.
\newblock URL \url{http://www.R-project.org/}.
\newblock {ISBN} 3-900051-07-0.

\bibitem[Rapela(2016)]{rapela16-techReportVBLR}
J.~Rapela.
\newblock Derivation of variational bayes linear regression.
\newblock Technical report, University of California San Diego, 2016.

\bibitem[Rapela et~al.(2006)Rapela, Mendel, and Grzywacz]{rapelaEtAl06}
J.~Rapela, J.M. Mendel, and N.M. Grzywacz.
\newblock Estimating nonlinear receptive fields from natural images.
\newblock \emph{Journal of Vision}, 6\penalty0 (4):\penalty0 441--474, 2006.

\bibitem[Rapela et~al.(2010)Rapela, Felsen, Touryan, Mendel, and
  Grzywacz]{rapelaEtAl10}
J.~Rapela, G.~Felsen, J.~Touryan, J.M. Mendel, and N.M. Grzywacz.
\newblock e\uppercase{PPR}: a new startegy for the characterization of sensory
  cells from input/output data.
\newblock \emph{Network: Computation in Neural Systems}, 21\penalty0
  (1-2):\penalty0 35--90, 2010.

\bibitem[Reey et~al.(2010)Reey, Tsuchiya, and Serre]{reddyEtAl10}
L.~Reey, N.~Tsuchiya, and T.~Serre.
\newblock Reading the mind’s eye: decoding category information during mental
  imagery.
\newblock \emph{Neuroimage}, 50:\penalty0 818--825, 2010.

\bibitem[Regan(1966)]{regan66}
D~Regan.
\newblock Some characteristics of average steady-state and transient responses
  evoked by modulated light.
\newblock \emph{Electroencephalography and clinical neurophysiology},
  20\penalty0 (3):\penalty0 238--248, 1966.

\bibitem[Rice and Hagstrom(1989)]{riceAndHagstrom89}
D.M. Rice and E.C. Hagstrom.
\newblock Some evidence in support of a relationship between human auditory
  signal-detection performance and the phase of the alpha cycle.
\newblock \emph{Perceptual and motor skills}, 1989.

\bibitem[Rolke(2008)]{rolke08}
B.~Rolke.
\newblock Temporal perpation facilitates perceptual identification of letters.
\newblock \emph{Perception and Psychophysics}, 70:\penalty0 1305--1313, 2008.

\bibitem[Romei et~al.(2012)Romei, Gross, and Thut]{romeiEtAl12}
V.~Romei, J.~Gross, and G.~Thut.
\newblock Sounds reset rhythms of visual cortex and corresponding human visual
  perception.
\newblock \emph{Current Biology}, 22:\penalty0 807--813, May 2012.

\bibitem[Sajda et~al.(2009)Sajda, Philiastides, and Parra]{sajdaEtAl09}
P.~Sajda, M.G. Philiastides, and L.C. Parra.
\newblock Single-trial analysis of neuroimaging data: inferring neural networks
  underlying perceptual decision-making in the human brain.
\newblock \emph{Biomedical Engineering, IEEE Reviews in}, 2:\penalty0 97--109,
  2009.

\bibitem[Saleh et~al.(2010)Saleh, Reimer, Penn, Ojakangas, and
  Hatsopoulos]{salehEtAl10}
Maryam Saleh, Jacob Reimer, Richard Penn, Catherine~L Ojakangas, and Nicholas~G
  Hatsopoulos.
\newblock Fast and slow oscillations in human primary motor cortex predict
  oncoming behaviorally relevant cues.
\newblock \emph{Neuron}, 65\penalty0 (4):\penalty0 461--471, 2010.

\bibitem[Samaha et~al.(2015)Samaha, Bauer, Cimaroli, and Postle]{samahaEtAl15}
Jason Samaha, Phoebe Bauer, Sawyer Cimaroli, and Bradley~R Postle.
\newblock Top-down control of the phase of alpha-band oscillations as a
  mechanism for temporal prediction.
\newblock \emph{Proceedings of the National Academy of Sciences}, 112\penalty0
  (27):\penalty0 8439--8444, 2015.

\bibitem[Sanders(1998)]{sanders98}
A.F. Sanders.
\newblock \emph{Elements of human performance: Reaction processes and attention
  in human skill}.
\newblock Lawrence Erbaum, Mahwah, NJ, 1998.

\bibitem[Sauseng et~al.(2008)Sauseng, Klimesch, Gruber, and
  Birbaumer]{sausengEtAl08}
P.~Sauseng, W.~Klimesch, W.R. Gruber, and N.~Birbaumer.
\newblock Cross-frequency phase synchornization: a brain mechanism of memory
  matching and attention.
\newblock \emph{Neuroimage}, 40:\penalty0 308--317, 2008.

\bibitem[Schroeder and Lakatos(2009)]{schroederAndLakatos09}
C.E. Schroeder and P.~Lakatos.
\newblock Low-frequency neuronal oscillations as instruments of sensory
  selection.
\newblock \emph{Trends in Neurosciences}, 32\penalty0 (1):\penalty0 9--18,
  2009.

\bibitem[Serences and Boynton(2007)]{serencesAndBoynton07}
J.T. Serences and G.M. Boynton.
\newblock The representation of behavioral choice for motion in human visual
  cortex.
\newblock \emph{Journal of Neurosciences}, 27:\penalty0 12893--12899, 2007.

\bibitem[Sherman et~al.(2016)Sherman, Kanai, Seth, and
  VanRullen]{shermanEtAl16}
Maxine~T Sherman, Ryota Kanai, Anil~K Seth, and Rufin VanRullen.
\newblock Rhythmic influence of top--down perceptual priors in the phase of
  prestimulus occipital alpha oscillations.
\newblock \emph{Journal of cognitive neuroscience}, 2016.

\bibitem[Spaak et~al.(2014)Spaak, {de Lange}, and Jensen]{spaakEtAl14}
E.~Spaak, F.P. {de Lange}, and O~Jensen.
\newblock Local entrainment of alpha oscillations by visual stimuli causes
  cyclic modulation of perception.
\newblock \emph{J Neurosci}, 34\penalty0 (10):\penalty0 3536--44, 2014.

\bibitem[Stefanics et~al.(2010)Stefanics, Hangya, Hern\'{a}di, Winkler,
  Lakatos, and Ulbert]{stefanicsEtAl10}
G.~Stefanics, B.~Hangya, I.~Hern\'{a}di, I.~Winkler, P.~Lakatos, and I.~Ulbert.
\newblock Phase entrainment of human delta oscillations can mediate the effects
  of expectation on reaction speed.
\newblock \emph{The Journal of Neurosciences}, 30\penalty0 (41):\penalty0
  13578--13585, 2010.

\bibitem[{Tallon Baudry} et~al.(1996){Tallon Baudry}, Bertrand, Delpuech, and
  Pernier]{tallonBaudryEtAl96}
C.~{Tallon Baudry}, O.~Bertrand, C.~Delpuech, and J.~Pernier.
\newblock Stimulus specificity of phase-locked and non-phase-locked 40 hz
  visual responses in human.
\newblock \emph{Journal of Neuroscience}, 16:\penalty0 4240--4349, 1996.

\bibitem[Tecce(1972)]{tecce72}
J.J. Tecce.
\newblock Contingent negative variation (\uppercase{CNV}) and psychological
  processes in man.
\newblock \emph{Psychological bulletin}, 77\penalty0 (2):\penalty0 73--108,
  1972.

\bibitem[{The Inkscape Team}(2004)]{harringtonEtAl04}
{The Inkscape Team}.
\newblock Inkscape, 2004.
\newblock URL \url{http://www.inkscape.org}.

\bibitem[Thorne et~al.(2011)Thorne, Vos, Viola, and Debener]{thorneEtAl11}
J.D. Thorne, M.~De Vos, F.~Campos Viola, and S.~Debener.
\newblock Cross-modal phase reset predicts auditory task performance in humans.
\newblock \emph{The Journal of Neuroscience}, 31\penalty0 (10):\penalty0
  3853--3861, 2011.

\bibitem[Thut et~al.(2012)Thut, Miniussi, and Gross]{thutEtAl12}
G.~Thut, C.~Miniussi, and J.~Gross.
\newblock The functional importance of rhythmic activity in the brain.
\newblock \emph{Current Biology}, 22\penalty0 (16), 2012.
\newblock \doi{10.1016/j.cub.2012.06.061}.

\bibitem[Tong and Pratte(2012)]{tongAndPratte12}
F.~Tong and M.S. Pratte.
\newblock Decoding patterns of human brain activity.
\newblock \emph{Annual Review of Psychology}, 63:\penalty0 483--509, 2012.

\bibitem[Torvalds(2008)]{torvalds08}
L.~Torvalds.
\newblock The linux kernel, 2008.
\newblock URL \url{http://www.kernel.org}.

\bibitem[Trimble and Potts(1975)]{trimbleAndPotts75}
J.L. Trimble and A.M. Potts.
\newblock Ongoing occipital rhythms and ver. i. stimulation at peaks of the
  alpha-rhythms.
\newblock \emph{Investigative Opthalmology}, 18:\penalty0 537--546, 1975.

\bibitem[Valera et~al.(1981)Valera, Toro, John, and Schwartz]{valeraEtAl81}
Francisco~J Valera, Alfredo Toro, E~Roy John, and Eric~L Schwartz.
\newblock Perceptual framing and cortical alpha rhythm.
\newblock \emph{Neuropsychologia}, 19\penalty0 (5):\penalty0 675--686, 1981.

\bibitem[Vallesi et~al.(2007)Vallesi, Mussoni, Mondani, Budai, Skrap, and
  Shallice]{vallesiEtAl07}
A.~Vallesi, A.~Mussoni, M.~Mondani, R.~Budai, M.~Skrap, and T.~Shallice.
\newblock The neural basis of temporal preparation: insights from brain tumor
  patients.
\newblock \emph{Neuropsychologia}, 45\penalty0 (12):\penalty0 2755--2763, 2007.

\bibitem[Vallesi et~al.(2009)Vallesi, McIntosh, Shallice, and
  Stuss]{vallesiEtAl09}
A.~Vallesi, A.R. McIntosh, T.~Shallice, and D.T. Stuss.
\newblock When time shapes behavior: fmri evidence of brain correlates of
  temporal monitoring.
\newblock \emph{Journal of Cognitive Neuroscience}, 21\penalty0 (6):\penalty0
  1116--1126, 2009.

\bibitem[van Diepen et~al.(2015)van Diepen, Cohen, Denys, and
  Mazaheri]{vanDiepenEtAl15}
Rosanne~M van Diepen, Michael~X Cohen, Damiaan Denys, and Ali Mazaheri.
\newblock Attention and temporal expectations modulate power, not phase, of
  ongoing alpha oscillations.
\newblock \emph{Journal of cognitive neuroscience}, 2015.

\bibitem[{van Elswijk} et~al.(2010){van Elswijk}, Maij, Schoffelen, Overeem,
  Stegeman, and Fries]{vanElswijkEtAl10}
G.~{van Elswijk}, F.~Maij, J.M. Schoffelen, S.~Overeem, D.F. Stegeman, and
  P.~Fries.
\newblock Corticospinal beta-band synchronization entails rhythmic gain
  modulation.
\newblock \emph{Journal of Neurosciences}, 30\penalty0 (12):\penalty0
  4481--4488, 2010.

\bibitem[VanRullen and Koch(2003)]{vanRullenAndKoch03}
R.~VanRullen and C.~Koch.
\newblock Is perception discrete or continous?
\newblock \emph{Trends in cognitive science}, 7\penalty0 (5):\penalty0
  207--213, 2003.

\bibitem[VanRullen et~al.(2011)VanRullen, Busch, Drewes, and
  Dubois]{vanRullenEtAl11}
R.~VanRullen, N.A. Busch, J.~Drewes, and J.~Dubois.
\newblock Ongoing {EEG} phase as a trial-by-trial predictor of perceptual and
  attentional variability.
\newblock \emph{Frontiers in Psychology}, 2, 2011.

\bibitem[VanRullen et~al.(2014)VanRullen, Zoefel, and Ilhan]{vanRullenEtAl14}
Rufin VanRullen, Benedikt Zoefel, and Barkin Ilhan.
\newblock On the cyclic nature of perception in vision versus audition.
\newblock \emph{Philosophical Transactions of the Royal Society B: Biological
  Sciences}, 369\penalty0 (1641):\penalty0 20130214, 2014.

\bibitem[Voytek et~al.(2010)Voytek, Canolty, Shestyuk, Crone, Parvizi, and
  Knight]{voytekEtAl10}
B.~Voytek, R.T. Canolty, A.~Shestyuk, N.E. Crone, J.~Parvizi, and R.T. Knight.
\newblock Shifts in gamma phase--amplitude coupling frequency from theta to
  alpha over posterior cortex during visual tasks.
\newblock \emph{Frontiers in human neuroscience}, 4, 2010.

\bibitem[Walter(1967)]{walter67}
W.G. Walter.
\newblock Slow potential changes in the human brain associated with expectancy,
  decision, and intention.
\newblock In W.~Cobb and C.~Morocutti, editors, \emph{The evoked potentials}.
  Elsevier, Amsterdam, 1967.

\bibitem[Walter et~al.(1964)Walter, Cooper, Aldridge, McCallum, and
  Winter]{walterEtAl64}
W.G. Walter, R.~Cooper, V.J. Aldridge, W.C. McCallum, and A.L. Winter.
\newblock Contingent negative variation: an electrical sign of sensorimotor
  association and expectancy in the human brain.
\newblock \emph{Nature}, 203:\penalty0 380--384, 1964.

\bibitem[Westfall and Young(1993)]{westfallAndYoung93}
Peter~H Westfall and S~Stanley Young.
\newblock \emph{Resampling-based multiple testing: Examples and methods for
  p-value adjustment}, volume 279.
\newblock John Wiley \& Sons, 1993.

\bibitem[Wilcox(2012)]{wilcox12}
R.~Wilcox.
\newblock \emph{Introduction to robust estimation and hypothesis testing}.
\newblock Academic Press, New York, NY, 2012.

\bibitem[Wolpaw et~al.(2002)Wolpaw, Birbaumer, McFarland, Pfurtscheller, and
  Vaughan]{wolpawEtAl02}
R.~Wolpaw, N.~Birbaumer, D.~J. McFarland, G.~Pfurtscheller, and T.~M. Vaughan.
\newblock Brain-computer interfaces for communication and control.
\newblock \emph{Clinical Neurophysiology}, 113:\penalty0 767--791, 2002.

\bibitem[Woodrow(1914)]{woodrow14}
H.~Woodrow.
\newblock The measurement of attention.
\newblock \emph{The Psychological Monographs}, 17\penalty0 (5):\penalty0
  i--158, 1914.

\bibitem[Wundt(1874)]{wundt1874}
W.~Wundt.
\newblock \emph{Grunz\"{u}ge der physiologischen psychologie}.
\newblock Leipzig: W. Engelman, 1874.

\bibitem[Xiang et~al.(2013)Xiang, Poeppel, and Simon]{xiangEtAl13}
Juanjuan Xiang, David Poeppel, and Jonathan~Z Simon.
\newblock Physiological evidence for auditory modulation filterbanks: Cortical
  responses to concurrent modulations.
\newblock \emph{The Journal of the Acoustical Society of America}, 133\penalty0
  (1):\penalty0 EL7--EL12, 2013.

\bibitem[Yamagishi et~al.(2008)Yamagishi, Callan, Anderson, and
  Kawato]{yamagishiEtAl08}
N.~Yamagishi, D.E. Callan, S.J. Anderson, and M.~Kawato.
\newblock Attentional changes in pre-stimulus oscillatory activity within early
  visual cortex are predictive of human visual performance.
\newblock \emph{Brain Research}, 1197:\penalty0 115--122, 2008.

\bibitem[Zappoli et~al.(2000)Zappoli, Versari, Zappoli, Chiaramonti,
  Zappoli-Thyrion, Grazia-Arneodo, and Zerauschek]{zappoliEtAl00}
R.~Zappoli, A.~Versari, F.~Zappoli, R.~Chiaramonti, G.D. Zappoli-Thyrion,
  M.~Grazia-Arneodo, and V.~Zerauschek.
\newblock The effects on auditory neurocognitive evoked responses and cognitive
  negative variation activity of frontal cortex lesions or ablations in man:
  three new case studies.
\newblock \emph{International Journal of Psychophysiology}, 38:\penalty0
  109--144, 2000.

\bibitem[{Zion Golumbic} et~al.(2013){Zion Golumbic}, Ding, Bickel, Lakatos,
  Schevon, McKhann, Goodman, Emerson, Mehta, Simon, Poeppel, and
  Schroeder]{zionGolumbicEtAl13}
E.M. {Zion Golumbic}, N.~Ding, S.~Bickel, P.~Lakatos, C.A. Schevon, G.M.
  McKhann, R.R. Goodman, R.~Emerson, A.D. Mehta, J.Z. Simon, D.~Poeppel, and
  C.E. Schroeder.
\newblock Mechanisms underlying selective neuronal tracking of attended speech
  at a ``cocktail party''.
\newblock \emph{Neuron}, 77\penalty0 (5):\penalty0 980--991, 2013.

\bibitem[Zoefel and VanRullen(2015{\natexlab{a}})]{zoefelAndVanRullen15a}
Benedikt Zoefel and Rufin VanRullen.
\newblock Selective perceptual phase entrainment to speech rhythm in the
  absence of spectral energy fluctuations.
\newblock \emph{The Journal of Neuroscience}, 35\penalty0 (5):\penalty0
  1954--1964, 2015{\natexlab{a}}.

\bibitem[Zoefel and VanRullen(2015{\natexlab{b}})]{zoefelAndVanRullen15b}
Benedikt Zoefel and Rufin VanRullen.
\newblock The role of high-level processes for oscillatory phase entrainment to
  speech sound.
\newblock \emph{Frontiers in human neuroscience}, 9, 2015{\natexlab{b}}.

\bibitem[Zoefel and VanRullen(2016)]{zoefelAndVanRullen16}
Benedikt Zoefel and Rufin VanRullen.
\newblock {EEG} oscillations entrain their phase to high-level features of
  speech sound.
\newblock \emph{NeuroImage}, 124:\penalty0 16--23, 2016.

\bibitem[Zumer et~al.(2014)Zumer, Scheeringa, Schoffelen, Norris, and
  Jensen]{zumerEtAl14}
Johanna~M Zumer, Ren{\'e} Scheeringa, Jan-Mathijs Schoffelen, David~G Norris,
  and Ole Jensen.
\newblock Occipital alpha activity during stimulus processing gates the
  information flow to object-selective cortex.
\newblock \emph{PLoS Biol.}, 12\penalty0 (10), 2014.
\newblock \doi{10.1371/journal.pbio.1001965}.

\end{thebibliography}

\pagebreak

\setcounter{page}{1}

\renewcommand\thefigure{\thesection.\arabic{figure}}    
\setcounter{figure}{0}

\renewcommand\thetable{\thesection.\arabic{table}}    
\setcounter{table}{0}

\clearpage
\appendix
\numberwithin{equation}{section}

\section{Appendix}
\label{sec:appendixA}

\subsection{Controls on the SFP effect on ITPC}
\label{sec:controls}

The previous sections demonstrated the existence of the \gls{SFP} effect on
\gls{ITPC}, and argued for its relevance by showing that its strength
correlates with subjects' perceptual abilities and reaction speeds.  From the
successful decoding of \glspl{SFPD} from the \gls{ITPC} following the
presentation of standards we inferred that the \gls{SFPD} (i.e.,
the modulation variable) modulates the \gls{ITPC} evoked by standards
(i.e., the modulated variable). Below we present evidence supporting the
inference that the modulated variable is the \gls{ITPC} evoked by standard
stimuli (Section~\ref{sec:randomizedSTDsOnsetsControl}), and that the
modulating variable is the \gls{SFPD}
(Section~\ref{sec:shuffledSFPDsControl}).

\subsubsection{The ITPC triggered by standards is the modulated variable}
\label{sec:randomizedSTDsOnsetsControl}

Besides the presentation of standards, any other experimental event
(e.g., the presentation of a \gls{warningSignal}) could have generated
\gls{ITPC} modulations that allowed models to decode, from a 500~ms-long
segment of phase coherence starting at time $t_0$ after the presentation of a
\gls{warningSignal}, that the \gls{warningSignal} was presented $t_0$~ms before
the start of the segment.  We tested this possibility by building surrogate
datasets, identical to the original ones, with the exception that epochs were
aligned to random times, instead of being aligned to the presentation time of
standards, and then fitting decoding models to these surrogate datasets. 

For each epoch in an original dataset we built a corresponding epoch in the
corresponding surrogate dataset. To construct the surrogate epoch, the onset time of the
standard on the original dataset was shifted by a random number between a
minimum and a maximum value.  The minimum value was the negative of the
\gls{SFPD}, in order to guarantee that the surrogate standard onset
occurred after the previous \gls{warningSignal}.  The maximum value was the
minimum among the latency of the next deviant and the latency of the next
\gls{warningSignal}, minus 500~ms. In this way a surrogate standard occurred
before the next deviant and the next \gls{warningSignal}, and neither deviants
or \glspl{warningSignal} appeared in the 500~ms-long window used to decode the
\gls{SFPD}.  Since the onset time of standards determines the \gls{DMP}
values used as independent variables in the linear regression model
(Section~\ref{sec:linearRegressionModel}), the surrogate datasets only
randomized the independent variables of this model.
In the surrogate datasets, the distribution of phases at every moment in time
was not significantly different from the uniform distribution (p\textgreater
0.01, Raleigh test). In the main analysis, to decode \glspl{SFPD}, we
only used \gls{DMP} values at time points at which the distribution of phases
was significantly different from the uniform distribution. If we had applied
this same criterion with the surrogate epochs we would have no time points to
decode \glspl{SFPD}. Thus, for epochs in surrogate datasets we
used the same time points as in the original epochs to decode \gls{SFP}
durations.

For each original dataset, we repeated the above randomization 30 times. For
each surrogate dataset, we fitted a decoding model (in exactly the same way as
with the original dataset), and computed the correlation coefficient between
the model decodings and the \glspl{SFPD}. We used the median of these 30
correlation coefficients to measure the decoding power of models fitted to
surrogate datasets.
For models derived from \glspl{IC} in the left parieto-occipital cluster~04 and
the attended visual standards, Figure~\ref{fig:controls}a plots the
decoding power of models fitted to the original dataset versus that of models
fitted to surrogate datasets. The decoding power of models fitted to the
original dataset was significantly larger than that of models fitted to
surrogate datasets. The median pairwise difference between the decoding power
of models fitted original datasets minus that of models fitted to surrogate
datasets was 0.21, which was significantly larger than zero (95\%
confidence interval [0.17, 0.25]).

\begin{figure}
\begin{center}
\includegraphics{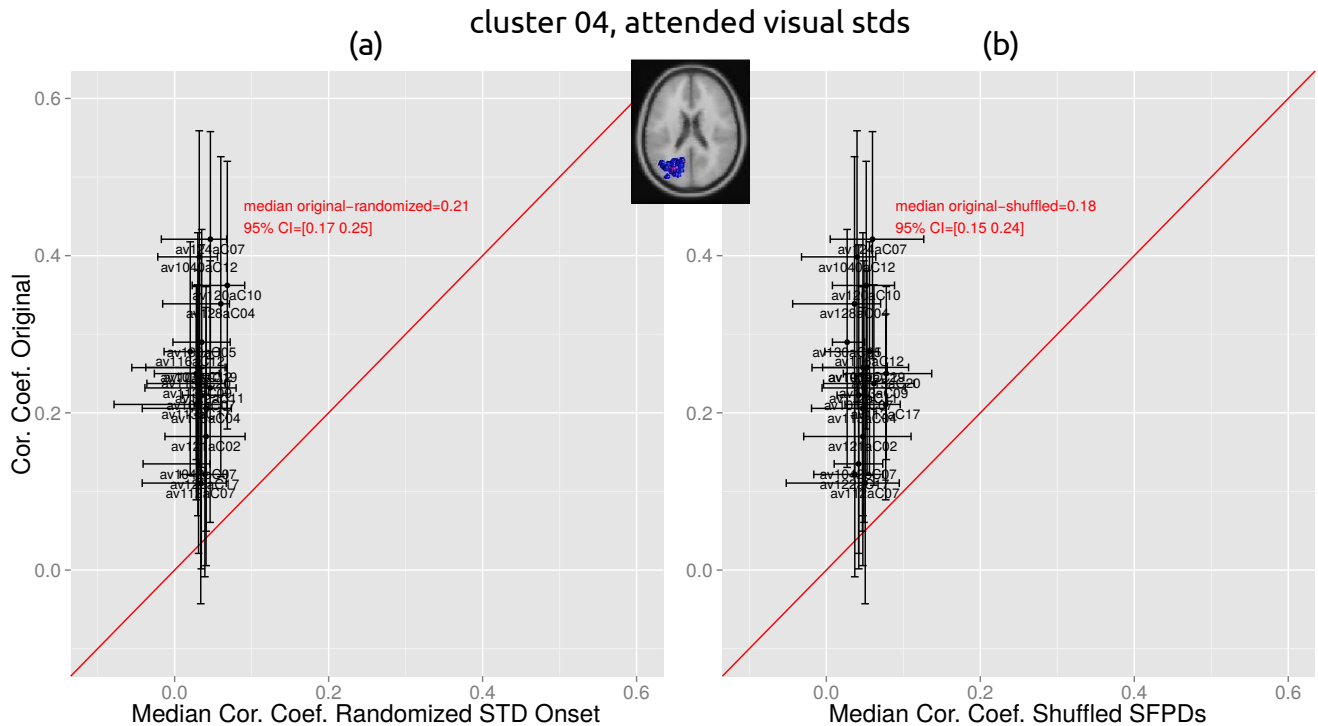}

\caption{Controls on the \gls{SFP} effect on \gls{ITPC}.  Panel (a) plots
correlation coefficients for models fitted to original datasets, with epochs
aligned to the presentation time of standards, versus correlation coefficients
for models fitted to surrogate datasets, with epochs aligned to random times.
Panel (b) is as panel (a) but for surrogate datasets with \glspl{SFPD} shuffled
among trials. Both panels correspond to the left parieto-occipital cluster~04
and attended visual standards. Correlations for models fitted to original
datasets are significantly larger than those from models fitted to surrogate
datasets, supporting our assumption that in the \gls{SFP} effect on \gls{ITPC}
the \gls{ITPC} triggered by standards is the modulated variable (panel (a)) and
that the \gls{SFPD} is the modulating variable (panel (b)).}

\label{fig:controls}
\end{center}
\end{figure}

Note that the estimation of the decoding power for models fitted to
surrogate datasets was computationally very expensive. For each cluster,
\gls{standardModality}, and \gls{attendedModality},
we fitted approximately
28,500 models, corresponding to 19 subjects $\times$ around 50 maximal
\glspl{SFPD} to select the optimal one
(Section~\ref{sec:selectionMaxSFPD}) $\times$ 30 repetitions of the
randomization procedure. For this reason, we limited this control
to one representative cluster, \gls{standardModality}, and
\gls{attendedModality}.

This control supports the claim that the \gls{ITPC} evoked by standards
is the experimental variable modulated by the \gls{SFPD}. In
particular, it constitutes evidence against the possibility that the observed
modulations on \gls{ITPC} were evoked by the \gls{warningSignal}, since in this
case modulations of \gls{ITPC} by the \gls{warningSignal}
should not be affected by the randomization of the onsets of standards.

\subsubsection{The SFP duration is the modulating variable}
\label{sec:shuffledSFPDsControl}

It is possible that from the \gls{DMP} values triggered by standards
one could reliably decode any variable, and not only the \gls{SFPD}.
This would be the case, for example, if the decoding model were overfitted to
the data.  Note, however, that the estimation and the model evaluation methods
used here were designed to avoid overfitting (i.e., we used a regularized error
function for parameter estimation, Section~\ref{sec:lmEstimationMethod}, and
cross-validation to evaluate the goodness of fit of models,
Section~\ref{sec:additionalStatInfo}).  To control for this possibility, we
proceed as in Section~\ref{sec:randomizedSTDsOnsetsControl}, but with
surrogate datasets with shuffled \glspl{SFPD}. That is, instead of
assigning to each trial its corresponding \gls{SFPD}, we assigned the
\gls{SFPD} of a randomly chosen (without replacement) trial. Since the
\gls{SFP} is the dependent variable in the linear regression model
(Section\ref{sec:linearRegressionModel}), the surrogate dataset only shuffled
the dependent values of this model.

Figure~\ref{fig:controls}b is as Figure~\ref{fig:controls}a, but for control
models estimate from surrogate datasets with shuffled \glspl{SFPD}. The
decoding power of models estimated from original datasets was significantly
larger than that of models estimated from the shuffled surrogate datasets. The
median pairwise difference between the correlation coefficient for a model
estimated from the original dataset minus that for a model estimated from the
corresponding surrogate dataset was 0.18, which was significantly different from
zero (95\% confidence interval: [0.15, 0.24]). Thus, Figure~\ref{fig:controls}b
supports the inference that the \gls{SFPD} is the experimental variable
modulating the \gls{ITPC} elicited by standards.

\subsection{Correlations between ITC and behavior}
\label{sec:itcAndBehavior}

For each \gls{IC} in a cluster and for a given \gls{standardModality} and
\gls{attendedModality} we extracted the \gls{ITC} value at the peak time and
frequency (Section~\ref{sec:itcAndPeakITCFrequency}). Then we computed the
correlation coefficient between peak \gls{ITC} values of \glspl{IC} and
error rates or mean reaction times of the subjects corresponding to the
\glspl{IC}.  Tables~\ref{table:peakITCStatsErrorRates}
and~\ref{table:peakITCStatsMeanRTs} show the obtained correlation coefficients
and corresponding p-values for error rates and mean reaction times,
respectively, across all clusters, \glspl{standardModality}, and
\glspl{attendedModality}. Highlighted in blue are entries with uncorrected
p-values smaller than 0.05. A double dagger (dagger) below the image of a
cluster in Figure~\ref{fig:clustersITC} indicates a significant correlation
between peak \gls{ITC} values and error rates (mean reaction times) in the
corresponding cluster and in the \gls{standardModality} and
\gls{attendedModality} given by the color of the double dagger (dagger).

\begin{table}[ht]

\caption{Correlations between peak ITC values and subjects' error rates. Each
cell shows the correlation coefficient, r,  and p-values unadjusted, p, and
adjusted, adj\_p, for multiple comparisons. Blue cells highlight correlations
with p\textless0.05.}

\begin{center}
\begin{tabular}{|c|c|c c c|c c c|}\hline
\multicolumn{1}{|c}{Cluster} & \multicolumn{1}{|c|}{Standard}  & \multicolumn{3}{c|}{Visual Attention}  & \multicolumn{3}{c|}{Auditory Attention}\\
\multicolumn{1}{|c}{} & \multicolumn{1}{|c|}{Modality}  & \multicolumn{1}{c}{r} & \multicolumn{1}{c}{p} & \multicolumn{1}{c|}{adj\_p} & \multicolumn{1}{c}{r} & \multicolumn{1}{c}{p} & \multicolumn{1}{c|}{adj\_p}\\ \hline \hline
03 & Visual & -0.36 & 0.1624 & 1.00 & 0.35 & 0.2162 & 1.00 \\
 & Auditory & -0.51 & 0.0656 & 0.99 & 0.15 & 0.6132 & 1.00 \\
\hline
04 & Visual & \cellcolor{NavyBlue}{-0.53} & \cellcolor{NavyBlue}{0.0382} & \cellcolor{NavyBlue}{0.96} & -0.14 & 0.5882 & 1.00 \\
 & Auditory & 0.52 & 0.0678 & 0.99 & -0.32 & 0.1938 & 1.00 \\
\hline
05 & Visual & -0.36 & 0.1944 & 1.00 & 0.13 & 0.6906 & 1.00 \\
 & Auditory & -0.24 & 0.4678 & 1.00 & -0.38 & 0.2162 & 1.00 \\
\hline
06 & Visual & 0.03 & 0.9206 & 1.00 & -0.01 & 0.9718 & 1.00 \\
 & Auditory & -0.23 & 0.4138 & 1.00 & -0.35 & 0.1852 & 1.00 \\
\hline
07 & Visual & -0.49 & 0.1682 & 1.00 & 0.20 & 0.6042 & 1.00 \\
 & Auditory & -0.39 & 0.2816 & 1.00 & -0.33 & 0.3748 & 1.00 \\
\hline
09 & Visual & 0.03 & 0.9092 & 1.00 & -0.34 & 0.2234 & 1.00 \\
 & Auditory & \cellcolor{NavyBlue}{-0.62} & \cellcolor{NavyBlue}{0.0190} & \cellcolor{NavyBlue}{0.80} & \cellcolor{NavyBlue}{-0.58} & \cellcolor{NavyBlue}{0.0278} & \cellcolor{NavyBlue}{0.90} \\
\hline
10 & Visual & -0.38 & 0.2706 & 1.00 & 0.31 & 0.4092 & 1.00 \\
 & Auditory & -0.13 & 0.7218 & 1.00 & -0.10 & 0.7738 & 1.00 \\
\hline
11 & Visual & 0.12 & 0.7058 & 1.00 & 0.08 & 0.7788 & 1.00 \\
 & Auditory & \cellcolor{NavyBlue}{-0.69} & \cellcolor{NavyBlue}{0.0130} & \cellcolor{NavyBlue}{0.67} & -0.29 & 0.3292 & 1.00 \\
\hline
13 & Visual & -0.24 & 0.4756 & 1.00 & \cellcolor{NavyBlue}{0.77} & \cellcolor{NavyBlue}{0.0056} & \cellcolor{NavyBlue}{0.40} \\
 & Auditory & -0.34 & 0.3394 & 1.00 & \cellcolor{NavyBlue}{-0.72} & \cellcolor{NavyBlue}{0.0162} & \cellcolor{NavyBlue}{0.75} \\
\hline
14 & Visual & -0.50 & 0.0854 & 1.00 & -0.39 & 0.2208 & 1.00 \\
 & Auditory & -0.27 & 0.3572 & 1.00 & -0.31 & 0.3626 & 1.00 \\
\hline
15 & Visual & 0.25 & 0.3156 & 1.00 & -0.02 & 0.9458 & 1.00 \\
 & Auditory & -0.35 & 0.2036 & 1.00 & \cellcolor{NavyBlue}{-0.54} & \cellcolor{NavyBlue}{0.0196} & \cellcolor{NavyBlue}{0.81} \\
\hline
17 & Visual & -0.25 & 0.5424 & 1.00 & -0.06 & 0.8630 & 1.00 \\
 & Auditory & 0.48 & 0.2202 & 1.00 & -0.02 & 0.9368 & 1.00 \\
\hline
18 & Visual & -0.34 & 0.2390 & 1.00 & -0.03 & 0.9120 & 1.00 \\
 & Auditory & -0.21 & 0.4862 & 1.00 & \cellcolor{NavyBlue}{-0.67} & \cellcolor{NavyBlue}{0.0242} & \cellcolor{NavyBlue}{0.87} \\
\hline
19 & Visual & 0.15 & 0.5896 & 1.00 & 0.18 & 0.5638 & 1.00 \\
 & Auditory & -0.04 & 0.9056 & 1.00 & \cellcolor{NavyBlue}{-0.85} & \cellcolor{NavyBlue}{0.0010} & \cellcolor{NavyBlue}{0.11} \\
\hline
\end{tabular}

\end{center}
\label{table:peakITCStatsErrorRates}
\end{table}

\begin{table}[ht]

\caption{Correlations between peak ITC values and subjects' mean
reaction times. Same format as Table~\ref{table:peakITCStatsErrorRates}.}

\begin{center}
\begin{tabular}{|c|c|c c c|c c c|}\hline
\multicolumn{1}{|c}{Cluster} & \multicolumn{1}{|c|}{Standard}  & \multicolumn{3}{c|}{Visual Attention}  & \multicolumn{3}{c|}{Auditory Attention}\\
\multicolumn{1}{|c}{} & \multicolumn{1}{|c|}{Modality}  & \multicolumn{1}{c}{r} & \multicolumn{1}{c}{p} & \multicolumn{1}{c|}{adj\_p} & \multicolumn{1}{c}{r} & \multicolumn{1}{c}{p} & \multicolumn{1}{c|}{adj\_p}\\ \hline \hline
03 & Visual & -0.22 & 0.4036 & 1.00 & 0.25 & 0.3388 & 1.00 \\
 & Auditory & -0.46 & 0.0602 & 0.97 & -0.09 & 0.7218 & 1.00 \\
\hline
04 & Visual & -0.08 & 0.7416 & 1.00 & 0.00 & 0.9882 & 1.00 \\
 & Auditory & 0.02 & 0.9322 & 1.00 & -0.28 & 0.2512 & 1.00 \\
\hline
05 & Visual & 0.15 & 0.5894 & 1.00 & 0.41 & 0.1536 & 1.00 \\
 & Auditory & -0.06 & 0.8322 & 1.00 & -0.19 & 0.5250 & 1.00 \\
\hline
06 & Visual & -0.28 & 0.2808 & 1.00 & 0.00 & 0.9882 & 1.00 \\
 & Auditory & -0.23 & 0.4088 & 1.00 & -0.33 & 0.2104 & 1.00 \\
\hline
07 & Visual & -0.04 & 0.9086 & 1.00 & 0.08 & 0.8406 & 1.00 \\
 & Auditory & 0.05 & 0.8878 & 1.00 & 0.03 & 0.9334 & 1.00 \\
\hline
09 & Visual & -0.16 & 0.5782 & 1.00 & -0.05 & 0.8532 & 1.00 \\
 & Auditory & 0.12 & 0.6622 & 1.00 & -0.31 & 0.2766 & 1.00 \\
\hline
10 & Visual & 0.19 & 0.5968 & 1.00 & 0.40 & 0.2500 & 1.00 \\
 & Auditory & 0.39 & 0.2844 & 1.00 & -0.04 & 0.9142 & 1.00 \\
\hline
11 & Visual & -0.21 & 0.4930 & 1.00 & -0.17 & 0.5838 & 1.00 \\
 & Auditory & -0.43 & 0.1400 & 1.00 & 0.26 & 0.3822 & 1.00 \\
\hline
13 & Visual & 0.25 & 0.4262 & 1.00 & 0.54 & 0.0746 & 0.99 \\
 & Auditory & 0.09 & 0.7674 & 1.00 & -0.26 & 0.4130 & 1.00 \\
\hline
14 & Visual & -0.24 & 0.4220 & 1.00 & -0.17 & 0.5694 & 1.00 \\
 & Auditory & -0.27 & 0.3618 & 1.00 & -0.22 & 0.4852 & 1.00 \\
\hline
15 & Visual & -0.23 & 0.3250 & 1.00 & -0.03 & 0.9086 & 1.00 \\
 & Auditory & 0.03 & 0.8986 & 1.00 & -0.23 & 0.3204 & 1.00 \\
\hline
17 & Visual & -0.18 & 0.6160 & 1.00 & 0.06 & 0.8692 & 1.00 \\
 & Auditory & -0.15 & 0.6830 & 1.00 & \cellcolor{NavyBlue}{-0.74} & \cellcolor{NavyBlue}{0.0250} & \cellcolor{NavyBlue}{0.78} \\
\hline
18 & Visual & -0.31 & 0.2542 & 1.00 & \cellcolor{NavyBlue}{-0.58} & \cellcolor{NavyBlue}{0.0240} & \cellcolor{NavyBlue}{0.76} \\
 & Auditory & 0.15 & 0.5918 & 1.00 & -0.21 & 0.4384 & 1.00 \\
\hline
19 & Visual & -0.39 & 0.1592 & 1.00 & -0.02 & 0.9502 & 1.00 \\
 & Auditory & -0.13 & 0.6522 & 1.00 & -0.17 & 0.5450 & 1.00 \\
\hline
\end{tabular}

\end{center}
\label{table:peakITCStatsMeanRTs}
\end{table}

\begin{figure}
\begin{center}
\includegraphics{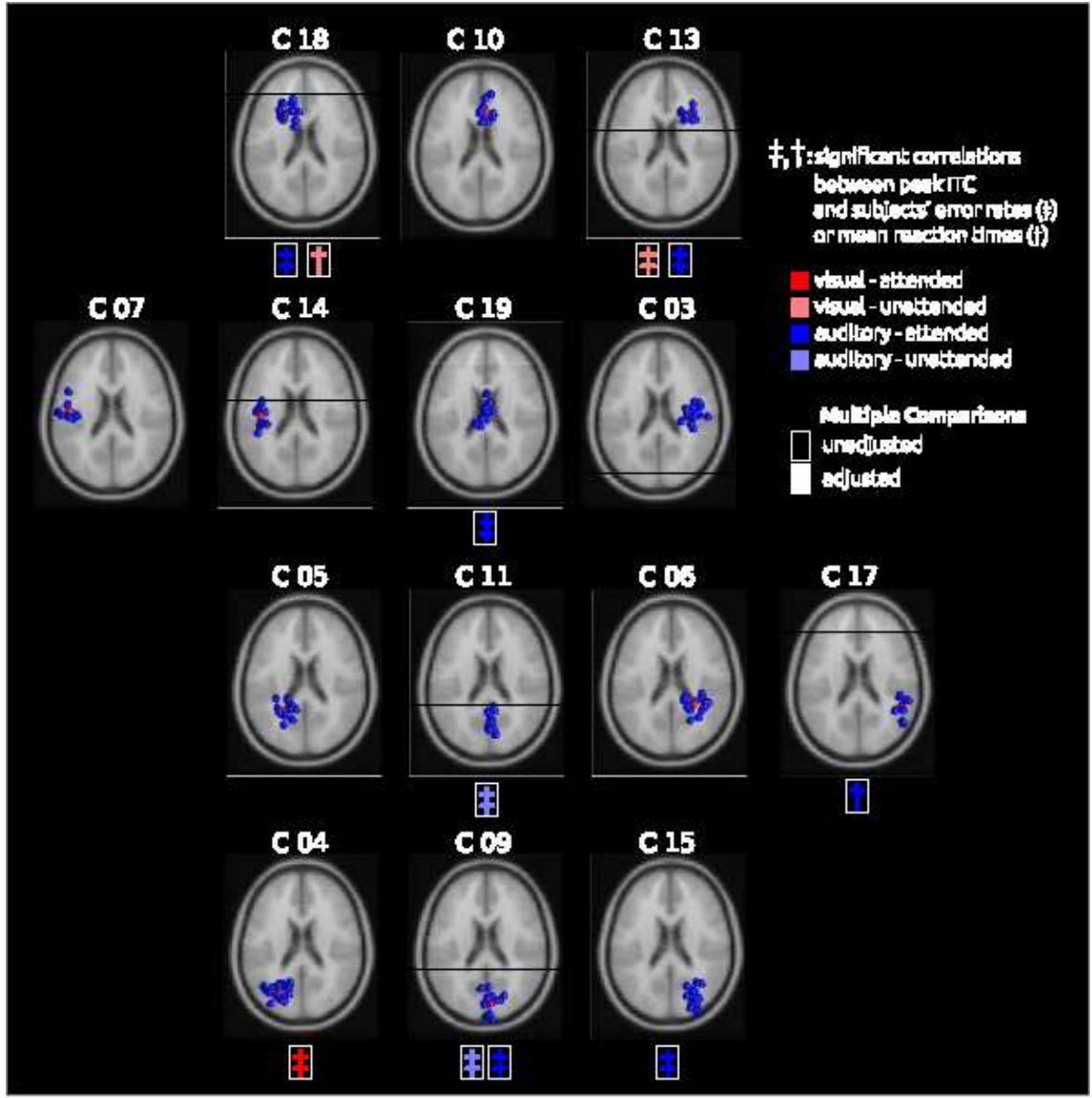}

\caption{Clusters of \glspl{IC}, and correlations between peak \gls{ITC} values
and subjects' behaviors.  
A dagger (double dagger) shows a significant correlation, unadjusted for
multiple comparisons, between peak \gls{ITC} values and subjects' mean reaction
times (error rates).  
With the exception of the left parieto-occipital cluster~04 and attended visual
standards, significant correlations between peak \gls{ITC} values and subjects'
behaviors
occur in different clusters, \glspl{standardModality} and
\glspl{attendedModality}, than correlations between models' decodings and
subjects' behaviors (Figure~\ref{fig:clusters}).
}

\label{fig:clustersITC}
\end{center}
\end{figure}

\subsection{Simulations of the SFP effect on ITPC}
\label{sec:simulations}

We simulated the model in Eq.~\ref{eq:trial} twice, with the precision of the
phase noise $n$ displaying larger ($\text{max}\kappa$=5.00,
$\text{min}\kappa$=0.01) and smaller ($\text{max}\kappa$=3.50,
$\text{min}\kappa$=1.51) oscillations and with parameters $f=7\,Hz$,
$\theta=\pi$, $f_n=3\,Hz$, and $s=2$.  For these simulations we used the
\glspl{SFPD}, $SFPD$ in Eq.~\ref{eq:trial}, from one experimental session
(subject av124a and attended visual standards). 
The simulated oscillations are shown in Figure~\ref{fig:simulatedOscillations}
as an erpimage~\citep{makeigEtAl04} sorted by \glspl{SFPD} (black curve to the
left of time zero).  Figure~\ref{fig:simulatedOscillations}a
and~\ref{fig:simulatedOscillations}b correspond to larger and smaller
fluctuations in the precision of the phase noise, respectively.
Figure~\ref{fig:avgDMPsAndCoefsSimulated} is as
Figure~\ref{fig:avgDMPsAndCoefs} but for the simulated data. The fluctuations
in averaged \gls{DMP} in trials closest to and furthest from the
\gls{warningSignal}, those in their difference, and fluctuations in the
estimated model coefficients are similar for the simulated and experimental
data. The correlation coefficient for the model fitted to data with smaller
fluctuations of the precision of the noise was $r=0.23$ (95\% CI=[0.12, 0.34]),
which was significantly smaller than that for the model fitted to data with
larger fluctuations of the precision of the noise $r=0.51$ (95\% CI=[0.43,
0.59]).

\begin{figure}
\begin{center}
\includegraphics{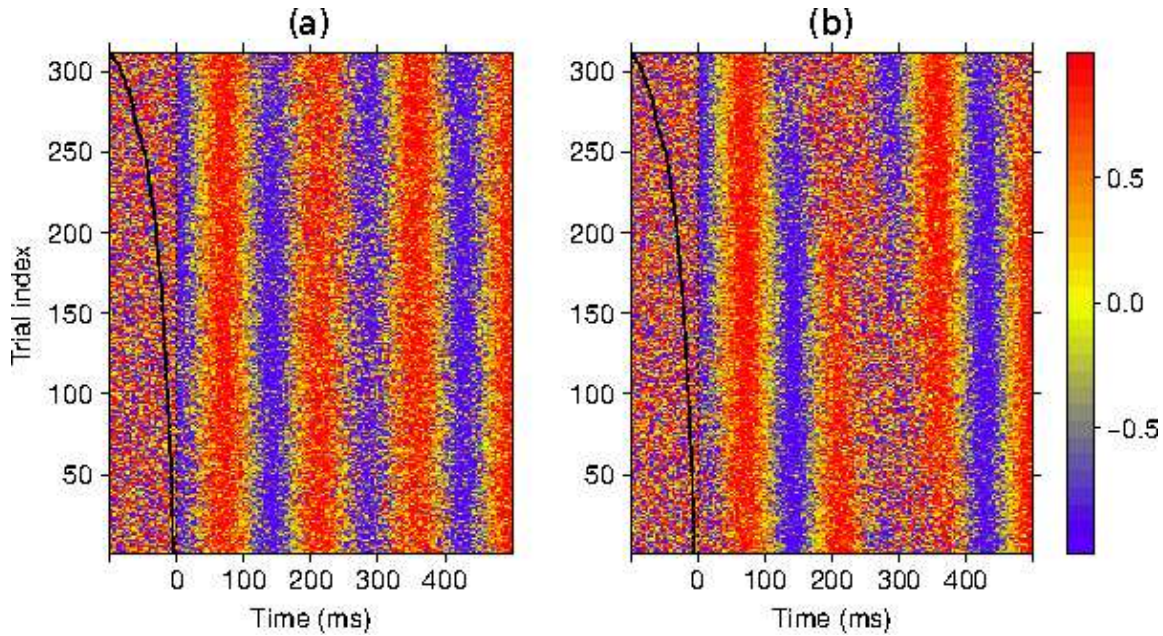}

\caption{Oscillations simulated according to Eq.~\ref{eq:trial} with phase
noise $n$ having small ($min\kappa=1.51$, $max\kappa=3.5$, panel (a))
and large ($min\kappa=0.01$, $max\kappa=5.0$, panel (b)) fluctuations in
precision (Eq.~\ref{eq:noise1Precision}).}

\label{fig:simulatedOscillations}
\end{center}
\end{figure}

\begin{figure}
\begin{center}
\includegraphics[height=8.0in]{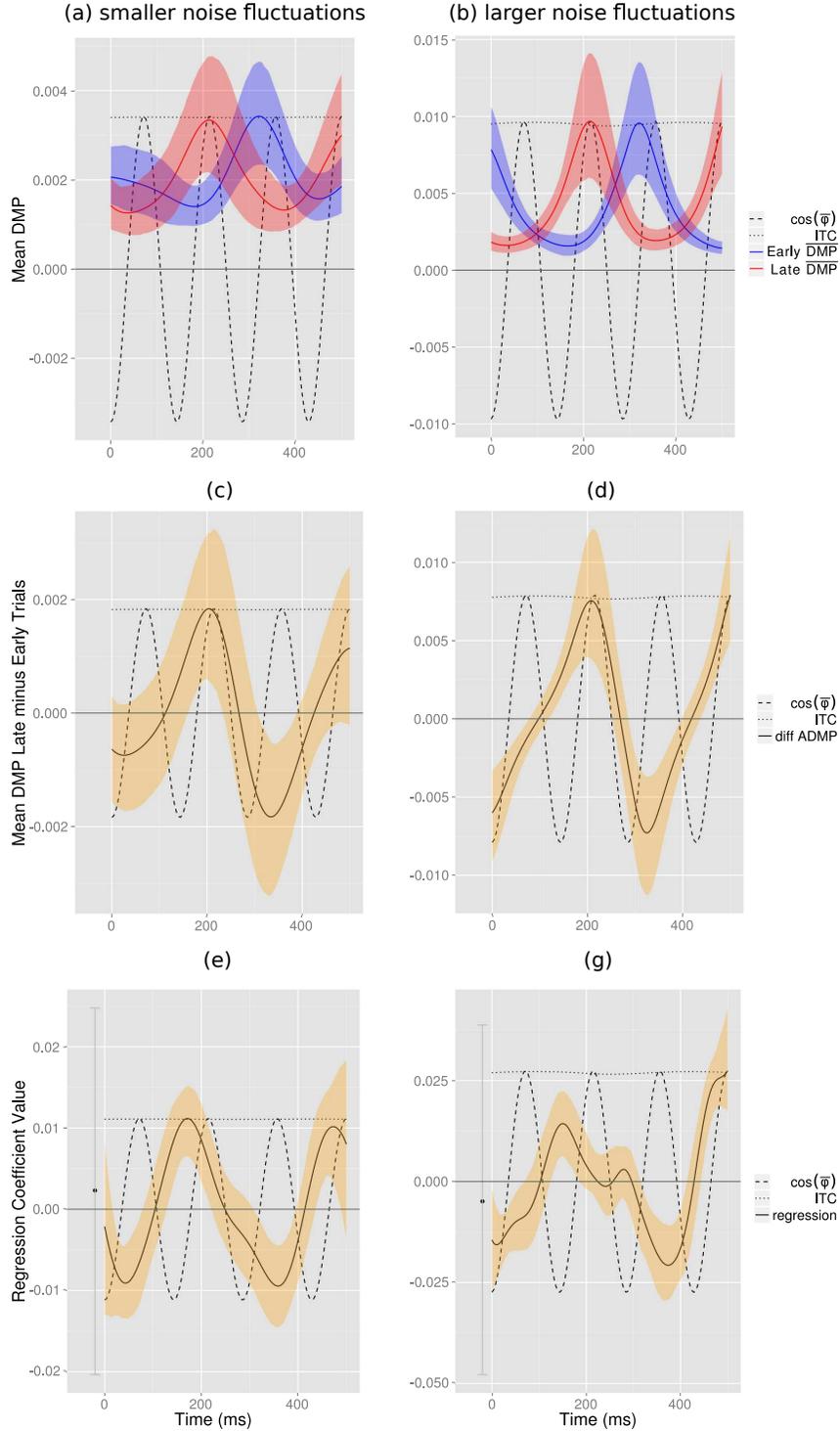}

\caption{Average \gls{DMP} in trials furthest from and closest to the
\gls{warningSignal} and regression coefficients for data simulated from
Eq.~\ref{eq:trial}. Same format as Figure~\ref{fig:avgDMPsAndCoefs}.
That the panels in this figure capture salient features of those in
Figure~\ref{fig:avgDMPsAndCoefs}, generated from neural recordings, suggests
that the simple oscillatory model in Eq.~\ref{eq:trial} is a good first
description of how the \gls{SFPD} modulates the phase of neural oscillations.
}

\label{fig:avgDMPsAndCoefsSimulated}
\end{center}
\end{figure}

\subsection{Selection of the optimal maximum SFPD}
\label{sec:selectionMaxSFPD}

To fit decoding models we used data from standards that were presented before
an optimal maximum \gls{SFPD} after the \gls{warningSignal}.  For models fitted
to data from \gls{IC} 05 of subject av130a, and unattended visual standards,
Figure~\ref{fig:selectionMaxSFPD} plots the decoding power of models as a
function on their maximum \gls{SFPD}.  We see that the decoding power of models
varied smoothly as a function of this maximum \gls{SFPD}. For each estimated
model we selected the maximum \gls{SFPD} that optimized the decoding power of
the model as the optimal maximum \gls{SFPD}.
We chose maximum \glspl{SFPD} between one second and the largest \gls{SFPD} of
any trial in the data, since maximum \glspl{SFPD} shorter than one second may
be contaminated by modulations from the nearby \gls{warningSignal}. For
\gls{IC} 05 of subject av130a and unattended visual standards, we selected an
optimal maximum \gls{SFP} duration of 2.8~seconds (time highlighted in red in
Figure~\ref{fig:selectionMaxSFPD}).

The optimal maximum \gls{SFPD} gives the range of \glspl{SFPD} at which
modulations of \gls{ITPC} by the \gls{SFPD} are stronger.
The mean of the optimal maximum \gls{SFPD} for models corresponding to the
visual and auditory \gls{attendedModality} were 2.4 and 2.7~seconds,
respectively. A repeated-measures ANOVA with optimal maximum \gls{SFPD} as
dependent variable found a significant main effect of \gls{attendedModality}
(F(1, 579)=8.45, p=0.0038). A posthoc analysis indicated that the mean of the
optimal maximum \gls{SFPD} was significantly shorter for models estimated from
epochs where attention was directed to the visual than to the auditory modality
(z=2.91, p=0.0018). Further information on this ANOVA appear in
Section~\ref{sec:anova}.
Therefore, modulations of \gls{ITPC} by the \gls{SFPD} were stronger earlier
when attention is directed to the visual than to the auditory modality.

\begin{figure}
\begin{center}
\includegraphics{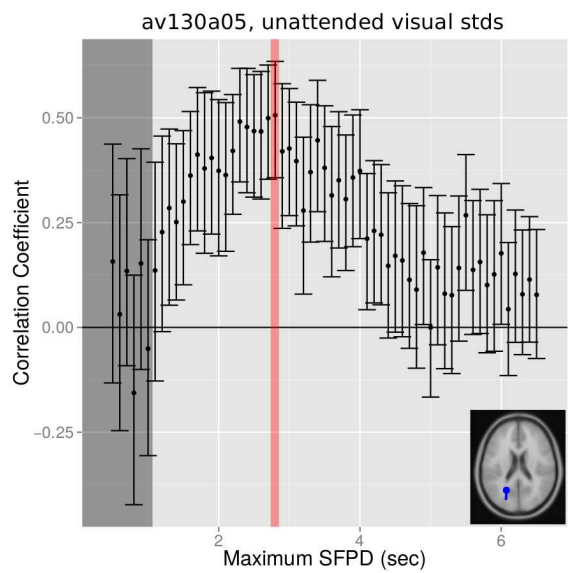}

\caption{Selection of the optimal maximum \gls{SFPD} for the decoding model for
\gls{IC} 05 of subject av130a and unattended visual standards. Decoding models
were estimated with maximum \glspl{SFPD} between 500~ms and the largest
\gls{SFPD} in a set of trials, in steps of 100~ms (abscissa), and the
correlation coefficient were computed between models' decodings and
experimental \glspl{SFPD} (ordinate). The \gls{SFPD} at which this correlation
coefficient was largest was selected as the optimal maximum \gls{SFPD} (red
bar).  To avoid possible modulations of the \gls{ITPC} by the
\gls{warningSignal}, we excluded from this selection \glspl{SFPD} shorter than
one second (gray bar).}

\label{fig:selectionMaxSFPD}
\end{center}
\end{figure}

In Section~\ref{sec:characterizedEpochs} we showed that the mean \gls{SFPD} was
significantly shorter in epochs where attention was directed to the visual than
to the auditory modality. Coincidentally, in this section we demonstrated that
the time window with the strongest modulations of \gls{ITPC} by the \gls{SFPD}
occurred earlier when attention was directed to the visual than to the auditory
modality. This coincidence suggests that subjects had learned the distribution
of \glspl{SFPD} and that their modulations of \gls{ITPC} reflected this
learning.

\subsection{Supplementary methods}
\label{sec:supplementaryMethods}

\subsubsection{ICA data preprocessing}
\label{sec:icaPreprocessing}

For each subject, we performed an ICA decomposition on his/her EEG recorded
brain potentials \citep{makeigEtAl96}. Briefly, if
$\mathbf{x}\in\mathbb{R}^N$ is an $N$-dimensional random vector representing
the EEG activity in $N$ channels, we estimated a mixing matrix
$A\in\mathbb{R}^{N\times N}$ so that

\begin{eqnarray}
\mathbf{x}=A\mathbf{s}
\label{eq:icaDecomposition}
\end{eqnarray}

\noindent and the components of the random vector
$\mathbf{s}\in\mathbb{R}^N$ were maximally independent. These components are called
the independent components (\glspl{IC}) through the manuscript.
We used the 
AMICA algorithm \citep{palmerEtAl07} with one model to compute this
decomposition.
Prior to the ICA decomposition, the EEG data was high- and low-pass
filtered with cutoff frequencies of of 1 and 50~Hz respectively.
After this decomposition, non-brain \glspl{IC}
(e.g., components corresponding to eye blinks, lateral eye movements, muscle
activity, bad channels, and muscle activity) were removed from
the ICA decomposition.
We kept, an average, of 26 components per subject, out of the $N=32$ components
obtained in the ICA decomposition.

\subsubsection{Clustering of ICs}
\label{sec:clusteringOfICs}

All \glspl{IC} of all subjects were grouped into clusters according to the
proximity of their equivalent dipole locations (see below) using the k-means
algorithm. This algorithm groups points of a D-dimensional metric space into
clusters, in such a way that the sum of the distances between points and their
corresponding cluster centroid is minimized. We represented each \gls{IC} by
the three-dimensional point given by the Talairach coordinates of its
equivalent dipole location.  A free parameter in k-means is the number of
clusters. We set this parameter to 17 to obtain clusters of reasonable
coarseness.  Clusters 8, 12, and 16 were not analyzed because they contained
too few \glspl{IC}, 4, 4 and 5, respectively (Table~\ref{table:clustersInfo}).

\vspace{0.3\baselineskip}\noindent\textbf{Equivalent Dipole location} 

\noindent The scalp map of the ith \gls{IC} is the ith column of the mixing
matrix $A$ in (\ref{eq:icaDecomposition}).  For each \gls{IC}, the location of
the electric current dipole whose projection in the scalp best matched the
\gls{IC} scalp map was estimated using the DIPFIT2 plugin of the EEGLAB
software. The location of this dipole is the equivalent dipole location of the
\gls{IC}.

\subsubsection{Circular statistics concepts}
\label{sec:circularStats}

This section introduces concepts from circular
statistics \citep{mardia72} used below
to define \gls{ITC} and \gls{DMP}. 
Given a set of circular variables (e.g., phases), $\theta_1, \ldots,
\theta_N$, we associate to each circular variable a two-dimensional unit
vector. Using notation from complex numbers, the unit vector associated with
variable $\theta_i$ is:

\begin{eqnarray}
vec(\theta_i)=e^{j\theta_i}
\label{eq:unitVector}
\end{eqnarray}

\noindent The \emph{resultant vector}, $\mathbf{R}$, is the sum of the associated unit
vectors: 

\begin{eqnarray}
\mathbf{R}(\theta_1, \ldots, \theta_N)=\sum_{i=1}^{N}vec(\theta_i)
\label{eq:resultantVector}
\end{eqnarray}

\noindent The \emph{mean resultant length}, $\bar{R}$, is the length of the resultant
vector divided by the number of variables:

\begin{eqnarray}
\bar{R}(\theta_1, \ldots, \theta_N)=\frac{1}{N}|\mathbf{R}(\theta_1, \ldots,
\theta_N)|
\label{eq:meanResultantLength}
\end{eqnarray}

\noindent The \emph{circular variance}, $CV$, is one minus the mean resultant length:

\begin{eqnarray}
CV(\theta_1, \ldots, \theta_N)=1-\bar{R}(\theta_1, \ldots, \theta_N)
\label{eq:cv}
\end{eqnarray}

\noindent The \emph{mean direction}, $\bar{\theta}$, is the angle of the resultant
vector:

\begin{eqnarray}
\bar{\theta}(\theta_1, \ldots, \theta_N)=\arg(\mathbf{R}(\theta_1, \ldots,
\theta_N))
\label{eq:meanDirection}
\end{eqnarray}

\noindent Note that the mean direction (and therefore the \gls{DMP}) is not
defined when the resultant vector is zero, since the angle of the zero vector
is undefined.

\subsubsection{ITC and Peak ITC frequency}
\label{sec:itcAndPeakITCFrequency}

The Inter-Trial Coherence (ITC) is a measure of \gls{ITPC} resulting from
averaging phase information among multiple epochs~\citep{tallonBaudryEtAl96,
delormeAndMakeig04}.
To compute ITC we extracted epochs from one second before to three seconds after
the presentation of standards. Then, a continuous wavelet transform was
performed on these epochs, using the Morlet wavelet with three significant
cycles, eight octaves, and 12 frequencies per octave. With a 250~Hz EEG
sampling rate, the previous parameters furnished a Morlet transform with a time
resolution of 70.71~ms, and a frequency resolution of 4.5 Hz, both at 10 Hz.
The function \texttt{cwt} in the \texttt{Rwave} package of the language
\texttt{R}~\citep{r12} was used to compute the continuous wavelet transform.
This transform provided the phases of every trial, for frequencies between 0.5
and 128 Hz, and for times between the start and end of the epoch. The \gls{ITC}
of a set of phases $\theta_1,\ldots,\theta_n$ (gray vectors in
Figure~\ref{fig:dmp}) is the mean resultant length ($\bar{R}$,
Eq.~\ref{eq:meanResultantLength}) of these phases (length of black vector in
Figure~\ref{fig:dmp}):

\begin{center}
\begin{eqnarray}
ITC(\theta_1,\ldots,\theta_N)=\bar{R}(\theta_1,\ldots,\theta_N)
\label{eq:itc}
\end{eqnarray}
\end{center}

For each \gls{IC} of every subject, \gls{standardModality}, and \gls{attendedModality}, we used the corresponding set of epochs to
calculate \gls{ITC} values between the start and end times of the epochs, and
between 0.5 Hz to 125 Hz. We selected the peak of these values between 100 and
500~ms after the presentation of the standard at time zero, and
between 1 and 14 Hz. The frequency corresponding to this peak (i.e., the peak
\gls{ITC} frequency) was then used to measure the single-trial phase values, as
described in Section~\ref{sec:measuringSingleTrialPhases}. The median peak
\gls{ITC} frequency and its
95\% confidence interval were 4.93~Hz and [4.59, 5.21]~Hz, respectively, and the
median time of the \gls{ITC} peak and its 95\% confidence interval were
215~ms and [212, 226]~ms, respectively.
\gls{ITC} values
between -200 and 500~ms around the presentation of attended visual standards, from \gls{IC} 5 of subject av130a are shown
in Figure~\ref{fig:itcPeakSelection}. The black cross in this
figure marks the peak \gls{ITC} value, occurring at peak frequency 8.15 Hz.

\begin{figure}
\begin{center}
\includegraphics[width=4.5in]{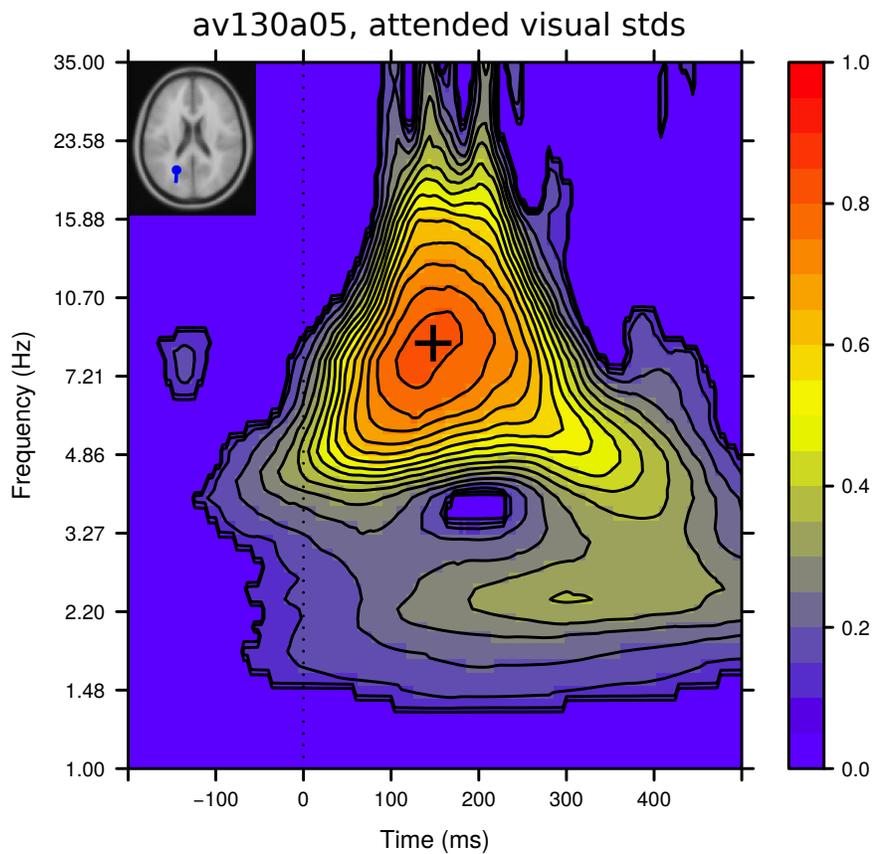}

\caption{Example set of \gls{ITC} values and selected peak from \gls{IC} 05 of
subject av130a and attended visual standards. The peak \gls{ITC} value (black
cross) was selected between 100 and 500~ms and between 1 and 14~Hz. The
selected peak ITC frequency, time, and amplitude were  8.52~Hz, 148~ms, and
0.83, respectively. Non-significant \gls{ITC} values (p\textgreater 0.01,
Rayleigh uniformity test) are masked in light blue.}

\label{fig:itcPeakSelection}
\end{center}
\end{figure}

\subsubsection{Measuring phases in single trials}
\label{sec:measuringSingleTrialPhases}

After selecting the peak \gls{ITC} frequency
(Section~\ref{sec:itcAndPeakITCFrequency}) for an \gls{IC} of a subject, a
\gls{standardModality}, and an \gls{attendedModality}, we computed a Gabor
transform of all epochs.  These epochs started one second before and ended
three
seconds after the presentation of standards. For this calculation, we used the
function \texttt{cgt} in the \texttt{Rwave} package of the language
\texttt{R}~\citep{r12}. We adjusted the scale of the Gabor's Gaussian window to
obtain three significant sinusoids at the frequency of the Gabor transform.
The phases of single trials were extracted from the complex coefficients of
this Gabor transform.  Figure~\ref{fig:exampleSingleTrialAnalysis}c illustrates
calculated phases for \gls{IC} 13, of subject av124a, and unattended visual
standards. Although we computed phases for the whole epoch duration, between -1
and 3~seconds, this figure only shows phases between -100 and 500~ms. To avoid
border effects, epochs were clipped to their final size, between 0 and
500~ms, after the computation of this Gabor transform.

\subsubsection{Relation between the averaged DMP and the ITC}
\label{sec:proofAveragedDMPApproachesITC}

The \gls{DMP} is a measure of phase decoherence, while the \gls{ITC} is a
measure of phase coherence. Here we derive an upper bound for the averaged
\gls{DMP}, $\overline{DMP}$, by the \gls{ITC}
(Proposition~\ref{prop:upperBoundAvgDMP}). From this bound we infer that
to minimal values of the $\overline{DMP}$ (i.e., $\overline{DMP}\simeq0$)
correspond large values of \gls{ITC} (i.e., $\gls{ITC}\simeq1$), and that the
maximal value of the $\overline{DMP}$ is 1/2.

\begin{upperBoundAvgDMP}
\label{prop:upperBoundAvgDMP}
$\overline{DMP}\le\frac{1}{2}-\frac{ITC}{2}$
\end{upperBoundAvgDMP}

\begin{proof}

We first rewrite DMP 

\begin{center}
\begin{eqnarray}
DMP(\theta_i|\{\theta_1,\ldots,\theta_N\})&=&CV(\theta_i,\bar{\theta}(\theta_1,\ldots,\theta_N))\nonumber\\
&=&1-\bar{R}(\theta_i,\bar{\theta}(\theta_1,\ldots,\theta_N))\nonumber\\
&=&1-\frac{1}{2}|\mathbf{R}(\theta_i,\bar{\theta}(\theta_1,\ldots,\theta_N))|\nonumber\\
&=&1-\frac{1}{2}|vec(\theta_i)+vec(\bar{\theta}(\theta_1,\ldots,\theta_N))|
\label{eq:dmpRewritten}
\end{eqnarray}
\end{center}

\noindent The first equality comes from Eq.~\ref{eq:dmp}, the second one from
Eq.~\ref{eq:cv}, the third one from Eq.~\ref{eq:meanResultantLength}, and the
fourth one from Eq.~\ref{eq:resultantVector}. We also rewrite
$vec(\bar{\theta}(\theta_1,\ldots,\theta_N))$:

\begin{center}
\begin{eqnarray}
vec(\bar{\theta}(\theta_1,\ldots,\theta_N))&=&vec(\arg(\mathbf{R}(\theta_1,\ldots,\theta_N)))\nonumber\\
&=&\frac{\mathbf{R}(\theta_1,\ldots,\theta_N)}{|\mathbf{R}(\theta_1,\ldots,\theta_N)|}\nonumber\\
&=&\frac{\mathbf{R}(\theta_1,\ldots,\theta_N)}{N\bar{R}(\theta_1,\ldots,\theta_N)|}
\label{eq:vecMeanPhaseRewritten}
\end{eqnarray}
\end{center}

\noindent The first equality comes from Eq.~\ref{eq:meanDirection}, the second
one from the fact that $vec$ applied to $\arg$ of any vector gives the
vector normalized to unit length, the third equality comes from
Eq.~\ref{eq:meanResultantLength}. Then,

\begin{center}
\begin{eqnarray}
1-\overline{DMP}&=&1-\frac{1}{N}\Sigma_{i=1}^NDMP(\theta_i|\{\theta_1,\ldots,\theta_N\})\nonumber\\
&=&\frac{1}{2N}\Sigma_{i=1}^N|vec(\theta_i)+vec(\bar{\theta}(\theta_1,\ldots,\theta_N))|\nonumber\\
&\ge&\frac{1}{2N}\left|\Sigma_{i=1}^Nvec(\theta_i)+N\,vec(\bar{\theta}(\theta_1,\ldots,\theta_N))\right|\nonumber\\
&=&\frac{1}{2N}\left|\mathbf{R}(\theta_1,\ldots,\theta_N)+\frac{\mathbf{R}(\theta_1,\ldots,\theta_N))}{\bar{R}(\theta_1,\ldots,\theta_N)}\right|\nonumber\\
&=&\frac{1}{2N}\left(1+\frac{1}{\bar{R}(\theta_1,\ldots,\theta_N)}\right)\left|\mathbf{R}(\theta_1,\ldots,\theta_N)\right|\nonumber\\
&=&\frac{1}{2}\left(1+\frac{1}{\bar{R}(\theta_1,\ldots,\theta_N)}\right)\bar{R}(\theta_1,\ldots,\theta_N)\nonumber\\
&=&\frac{1}{2}\left(\bar{R}(\theta_1,\ldots,\theta_N)+1\right)\nonumber\\
&=&\frac{ITC}{2}+\frac{1}{2}
\label{eq:boundAvgDMPByITC}
\end{eqnarray}
\end{center}

\noindent The first equality comes from the definition of $\overline{DMP}$, the
second one from Eq.~\ref{eq:dmp}, the third one from the triangle inequality,
the fourth one from Eq.~\ref{eq:resultantVector} and
Eq.~\ref{eq:vecMeanPhaseRewritten}, the fifth one from taking
$\mathbf{R(\theta_1,\ldots,\theta_N)}$ as common factor, the sixth one from
Eq.~\ref{eq:meanResultantLength}, the seventh one from distributing
$\bar{R}(\theta_1,\ldots,\theta_N)$, and the eight one from Eq.~\ref{eq:itc}.
The proposition follows by re-arranging Eq.~\ref{eq:boundAvgDMPByITC}.
\end{proof}

Because $0\le DMP\le 1$ it follows that $0\le \overline{DMP}\le 1$. When \gls{ITC} is
maximal (i.e., \gls{ITC}=1) Proposition~\ref{prop:upperBoundAvgDMP} shows that
$\overline{DMP}\le0$, thus $\overline{DMP}$ must achieve it minimal value
(i.e., $\overline{DMP}=0$). Also, this proposition shows that 1/2 is an upper
bound for $\overline{DMP}$, corresponding to zero \gls{ITC}.

\subsubsection{Linear regression model}
\label{sec:linearRegressionModel}

We used a linear regression model to decode the \gls{SFPD} of
trial $n$, $y[n]$, from the samples of \gls{DMP} values from this trial,
$x[n,\cdot]$.

\begin{eqnarray}
\hat{y}[n,\mathbf{w}]=\sum\limits_{k=1}^Kw[k]x[n, k]
\label{eq:linearRegressionModel}
\end{eqnarray}

\noindent where $\hat{y}[n,\mathbf{w}]$ is the decoded \gls{SFPD},
$w[k]$ is a regression coefficient, $x[n,k]$ is the \gls{DMP} value for trial
$n$ at sample time $k$, and $K$ is the number of sample points in the
500~ms-long time window following the presentation of standards for which the
distribution of phases across trials was significantly different from the
uniform distribution (p\textless 0.01; Rayleigh test).

\subsubsection{Method to estimate the coefficients in the linear regression model}
\label{sec:lmEstimationMethod}

We seek coefficients  $\mathbf{w}$ in the linear regression model
(Eq.~\ref{eq:linearRegressionModel}) such that the decodings of the
model, $\hat{y}[\cdot,\mathbf{w}]$, are as close as possible to the experimental
\glspl{SFPD}, $y[\cdot]$. Mathematically, we estimate the regression
coefficients $\mathbf{w}$ that minimize the least-squares error function

\begin{eqnarray}
\text{MSE}(\mathbf{w})=\sum\limits_{n=1}^N(y[n]-\hat{y}[n,\mathbf{w}])^2
\label{eq:mseFunction}
\end{eqnarray}

\noindent where $N$ is the number of epochs. A difficulty in this estimation is
that \gls{DMP} values at neighboring sample points are highly correlated, and
correlations increase the variance of ordinary least-squares estimates. In
addition, the number of coefficients, $K$ in (\ref{eq:linearRegressionModel}),
is comparable to, and in some cases even larger than, the number of epochs, $N$
in~(\ref{eq:mseFunction}). This further increases the variance of the ordinary
least-squares estimates, and when $K>N$ these estimates are not even defined.
To address these problems here we use ridge regression~\citep[][Section
3.4.1]{hastieEtAl09}, which adds a penalty constraint to the least-squares
error function in Eq.~\ref{eq:mseFunction}, shrinking coefficients estimates
toward zero and therefore reducing their variability:

\begin{eqnarray}
\text{RMSE}(\mathbf{w})=\text{MSE}(\mathbf{w})+\alpha||\mathbf{w}||^2
\label{eq:ridgeRegressionErrorFunction}
\end{eqnarray}

\noindent where $\alpha$ determines the strength of the constraint (i.e., for
larger values of $\alpha$ the penalty constraint more strongly bias the
coefficient estimates away from minimizing the least-squares error function in
Eq.~\ref{eq:mseFunction} and toward zero).  We call
Eq.~\ref{eq:ridgeRegressionErrorFunction} the ridge-regression error function.

To find the optimal value $\mathbf{w}$ in
Eq.~\ref{eq:ridgeRegressionErrorFunction} we took a Bayesian
approach~\citep{bishop06}. We used a Gaussian likelihood function:

\begin{eqnarray}
P(\mathbf{y}|\mathbf{w}, \tau, \Phi)=N(\mathbf{y}|\Phi\mathbf{w},\tau^{-1}I)
\label{eq:likelihoodFunction}
\end{eqnarray}

\noindent where $\Phi$ is the matrix of \gls{DMP} values (i.e.,
$\Phi[n,k]=x[n,k]$) and $\tau$ is the constant precision of $\mathbf{y}$. We
chose a normal-gamma prior for $\mathbf{w}$ and $\tau$:

\begin{eqnarray}
P(\mathbf{w},\tau|\alpha)=N(\mathbf{w}|\mathbf{0},(\tau\alpha)^{-1}I)\;\text{Gam}(\tau|a_0,b_0)
\label{eq:priorWTau}
\end{eqnarray}

\noindent a Gamma hyper-prior for the hyper-parameter $\alpha$:

\begin{eqnarray}
P(\alpha)=\text{Gam}(\alpha|c_0,d_0)
\label{eq:priorAlpha}
\end{eqnarray}

\noindent and searched for the values of $\mathbf{w}$ that
maximized the the log of the posterior distribution:

\begin{eqnarray}
\text{J}(\mathbf{w})=\log P(\mathbf{w}|\mathbf{y},\tau,\alpha,\Phi)
\label{eq:logPosterior}
\end{eqnarray}

As we prove in Proposition~\ref{prop:minRidgeEquivalentMAP}, with
this choice of likelikhood function and priors, finding the coefficients that
maximize the log of the posterior distribution in Eq.~\ref{eq:logPosterior} is
equivalent to finding the coefficients that minimize the ridge-regression
error function in Eq.~\ref{eq:ridgeRegressionErrorFunction}, as we set
to do at the beginning of this section.

Due to the large dimensionality of $\mathbf{w}$, evaluating the posterior
distribution in Eq.~\ref{eq:logPosterior} is not feasible and one needs to
resort to approximation schemes~\citep[][Chapter 10]{bishop06}. These
approximations can be stochastic or deterministic. In principle, stochastic
approximations, such as Markov Chain Monte
Carlo~\citep{metropolisAndUlam49}, can give exact evaluations of the
posterior distribution, given infinite computational resources, but are often
restricted to small-scale problems.  Deterministic approximations, are based on
approximations of the posterior distribution, and although can never generate
exact results, scale well for large problems. Here we use the Variational Bayes
deterministic approximation~\citep{bishop06}, and implement it as described
in~\citet{rapela16-techReportVBLR}.

We applied a logarithmic transformation to the dependent variable, $y[n]$, in
order to equalize the variance of the residuals~\citep[][Section
3.9]{kutnerEtAl05}, and standardized the dependent variable and regressors to
have zero mean and unit variance~\citep[][Section 7.5]{kutnerEtAl05}.  Highly
influential trials (i.e., trials with a Cook's distance larger than 4/(N-K-1),
where N and K are the number of trials and regressors) were deleted before
estimating the model parameters \citep[][Section 10.4]{kutnerEtAl05}.

\begin{minRidgeEquivalentMAP}
\label{prop:minRidgeEquivalentMAP}

Given the likelihood function in Eq.~\ref{eq:likelihoodFunction}, and the
priors in Eq.~\ref{eq:priorWTau} and Eq.~\ref{eq:priorAlpha}, the coefficients
$\mathbf{w}$ maximizing the logarithm of the posterior distribution
$J(\mathbf{w})$ in
Eq.~\ref{eq:logPosterior} minimize the ridge-regression error function
$\text{RMSE}(\mathbf{w})$ in
Eq.~\ref{eq:ridgeRegressionErrorFunction}.

\end{minRidgeEquivalentMAP}

\begin{proof}

We first rewrite the joint pdf $P(\mathbf{y},\mathbf{w},\tau,\alpha,\Phi)$ as:

\begin{center}
\begin{eqnarray}
P(\mathbf{y},\mathbf{w},\tau,\alpha,\Phi)&=&P(\mathbf{y}|\mathbf{w},\tau,\Phi)P(\mathbf{w},\tau|\alpha)P(\alpha)\nonumber\\
&=&N(\mathbf{y}|\Phi\mathbf{w},\tau^{-1}I)N(\mathbf{w}|\mathbf{0},(\tau\alpha)^{-1}I)\text{Gam}(\tau|a_0,b_0)\text{Gam}(\alpha|c_0,d_0)\nonumber\\
&=&\prod\limits_{n=1}^NN(y[n]|\langle\mathbf{x}[n,\cdot],\mathbf{w}\rangle,\tau^{-1})\nonumber\\
& &\prod\limits_{k=1}^KN(w[k]|0,(\tau\alpha)^{-1})\nonumber\\
& &\text{Gam}(a_0,b_0)\text{Gam}(c_0,d_0)\nonumber\\
&=&\prod\limits_{n=1}^N\frac{1}{\sqrt{2\pi\tau^{-1}}}
 \exp{\left(-\frac{(y[n]-\langle\mathbf{x}[n,\cdot],\mathbf{w}\rangle)^2}{2\tau^{-1}}\right)}\nonumber\\
& &\prod\limits_{k=1}^K\frac{1}{\sqrt{2\pi(\tau\alpha)^{-1}}}
   \exp{\left(-\frac{w[k]^2}{2(\tau\alpha)^{-1}}\right)}\nonumber\\
& &\text{Gam}(\tau|a_0,b_0)\text{Gam}(\alpha|c_0,d_0)
\label{eq:jointPDFRewritten}
\end{eqnarray}
\end{center}

\noindent The first equality comes from applying Bayes rule, the second one
from substituting Eq.~\ref{eq:likelihoodFunction}, Eq.~\ref{eq:priorWTau}, and
Eq.~\ref{eq:priorAlpha} into the first equality, the third one from the fact
that the Gaussian distributions in the second equality are independent, and the
last one from the definition of a Gaussian distribution. 

\begin{center}
\begin{eqnarray}
\argmax\limits_\mathbf{w}J(\mathbf{w})&=&\argmax\limits_\mathbf{w}\log P(\mathbf{w}|\mathbf{y},\tau,\alpha,\Phi)\nonumber\\
&=&\argmax\limits_\mathbf{w}\log P(\mathbf{y},\mathbf{w},\tau,\alpha,\Phi)\nonumber\\
&=&\argmax\limits_\mathbf{w}\left[\sum\limits_{n=1}^N-\frac{(y[n]-\langle\mathbf{x}[n,\cdot],\mathbf{w}]\rangle)^2}{2\tau^{-1}}+\sum\limits_{k=1}^K-\frac{w[k]^2}{2(\tau\alpha)^{-1}}\right]\nonumber\\
&=&\argmin\limits_\mathbf{w}\left[\sum\limits_{n=1}^N\frac{(y[n]-\langle\mathbf{x}[n,\cdot],\mathbf{w}]\rangle)^2}{2\tau^{-1}}+\sum\limits_{k=1}^K\frac{w[k]^2}{2(\tau\alpha)^{-1}}\right]\nonumber\\
&=&\argmin\limits_\mathbf{w}\left[\sum\limits_{n=1}^N(y[n]-\langle\mathbf{x}[n,\cdot],\mathbf{w}]\rangle)^2+\sum\limits_{k=1}^K\frac{w[k]^2}{\alpha^{-1}}\right]\nonumber\\
&=&\argmin\limits_\mathbf{w}\left[\text{MSE}(\mathbf{w})+\alpha||\mathbf{w}||^2\right]\nonumber\\
&=&\argmin\limits_\mathbf{w}\text{RMSE}(\mathbf{w})\nonumber
\end{eqnarray}
\end{center}

\noindent The first equality comes from the definition of $J(\mathbf{w})$ in
Eq.~\ref{eq:logPosterior}, 
the second one from the fact that, by Bayes rule,
$P(\mathbf{w}|\mathbf{y},\tau,\alpha,\Phi)$ equals
$P(\mathbf{y},\mathbf{w},\tau,\alpha,\Phi)$ times a factor independent of
$\mathbf{w}$ and thus
$\argmax P(\mathbf{w}|\mathbf{y},\tau,\alpha,\Phi)=\argmax P(\mathbf{y},\mathbf{w},\tau,\alpha,\Phi)$,
the third one by discarding the terms not including $\mathbf{w}$ in
Eq.~\ref{eq:jointPDFRewritten} and taking the logarithm,
the fourth one from the fact that the maximum of a function is the minimum of
the negative of that function,
the fifth one from the fact that $\argmin$ of a function does not change by
multiplying this function by a constant,
the sixth one from the definition of $\text{MSE}(\mathbf{w})$ in
Eq.~\ref{eq:mseFunction},
and the final one from the definition of $\text{RMSE}(\mathbf{w})$ in 
Eq.~\ref{eq:ridgeRegressionErrorFunction}.
\end{proof}

\subsubsection{Adjustment for multiple comparisons in correlation tests}
\label{sec:multipleComparisons}

For a family of hypothesis $H_1$ vs.\ $H'_1$, $H_2$ vs.\ $H'_2$, \ldots, $H_k$
vs.\ $H'_k$, we aim to control the Familywise Error Rate (FWE) defined as
FWE=P(Reject at least one $H_i\vert$all $H_i$ are true). For this we define
adjusted p-values, $\tilde{p}_i,\;i=1,\ldots,k$, and we reject $H_i$ at
FWE=$\alpha$ if and only if $\tilde{p}_i\leq\alpha$. We denote observed values
with lowercase and the corresponding random variables with uppercase.  The
definition of the adjusted p-value, $\tilde{p}_i$, should guarantee that
FWE=$\alpha$. That is, $\alpha$=FWE=P(Reject at least one $H_i$ at level
$\alpha\vert$all $H_i$ are true)=P($\min_{1\leq j\leq
k}\tilde{P}_j\leq\alpha\vert$all $H_i$ are true).  We define
$\tilde{p}_i=P(\min_{1\leq j\leq k}P_j\leq p_i\vert$all $H_i$ are true), where
$p_i$ is the observed non-adjusted p-value. When the p-values $P_j$ are
independent it can be shown \citep[][p.~47]{westfallAndYoung93} that the
previous definition leads to an exact multiple comparison method (i.e.,
FWE=$\alpha$). 

In our tests of correlation, the hypothesis $H_i$ ($H'_i$) states that the
correlation coefficient between two sequences $\{a_1, \ldots, a_n\}$ and
$\{b_1, \ldots, b_n\}$ corresponding to the ith cluster, \gls{standardModality}, and
\gls{attendedModality} is zero (different from zero).
To compute the adjusted p-value $\tilde{p}_i$ we use a resampling procedure. We first
generate nResamples=5,000 samples of $\min_{1\leq j\leq k}P_j$ under the null
hypothesis that all $H_i$ are true and then estimate $\tilde{p}_i$ as the proportion of samples
smaller or equal than the observed non-adjusted p-value $p_i$.
To generate a sample of $\min_{1\leq j\leq k}P_j$ under the null hypothesis,
for each cluster, \gls{standardModality}, and \gls{attendedModality} we: (1)
obtain the sequences $\{a_1, \ldots, a_n\}$ and $\{b_1, \ldots, b_n\}$ to test
for significant correlation (e.g., $a_i$=``correlation coefficient between
models'
decodings and \glspl{SFPD} for a subject $i$'' and $b_i$=``error rate for
subject $i$'' in Figure~\ref{fig:correlationsWithBehavior}a), (2) shuffle the
sequence $\{b_1, \ldots, b_n\}$ (to be under the null hypothesis), (3) compute
the p-value of the correlation coefficient between the $\{a_1, \ldots, a_n\}$ and
the shuffled $\{b_1, \ldots, b_n\}$. The sample of $\min_{1\leq j\leq k}P_j$
under the null hypothesis is then the minimum of all p-values generated in step
(3) across all clusters, \glspl{standardModality}, and
\glspl{attendedModality}.

We adjusted for multiple comparison correlation tests between (a) models'
decodings and \glspl{SFPD} (Figure~\ref{fig:exampleSingleTrialAnalysis}f,
colored dots in Figure~\ref{fig:clusters},
Table~\ref{table:nModelsSignCorPredictionsVsSFPDs}), (b) models' decoding
powers and  subjects' behavioral measures
(Figure~\ref{fig:correlationsWithBehavior}, daggers in
Figure~\ref{fig:clusters}, Tables~\ref{table:statsErrorRates}
and~\ref{table:statsMeanRTs}), and (c) subjects' peak \gls{ITC} values and
subjects' behavioral measures (Figure~\ref{fig:clustersITC},
Tables~\ref{table:peakITCStatsErrorRates} and~\ref{table:peakITCStatsMeanRTs}).
We begun by defining the family of hypothesis $H_1$ vs. $H'_1$, \ldots, $H_k$
vs. $H'_k$ in the definition of the FWE, see above.
In (a) we used 19 families of hypothesis, one per subject. The family for
subject $s$ contained hypothesis concerning correlation coefficients between
model decodings and \glspl{SFPD}, of all models for subject $s$ across
clusters, \glspl{standardModality}, and \glspl{attendedModality}, that were
significantly different from the intercept only model (p\textless0.01;
likelihood-ratio permutation test, Section~\ref{sec:additionalStatInfo}).
The null ith hypothesis for subject $s$ was $H^s_i$=\{the correlation
coefficient between models' decodings and \glspl{SFPD} for subject $s$ and the
ith combination of cluster, \gls{standardModality}, and \gls{attendedModality}
is zero\}. The mean and standard deviation of the number of pairs of hypothesis
in a family of hypothesis for a subject was $29.37\pm 12.64$.
In (b) we used one family of hypothesis for error rates and another one for
mean reaction times. Each of these families contained hypothesis regarding the
correlation between  models' decoding powers and subjects' behavioral measures
(error rates and mean reaction times) across all clusters,
\glspl{standardModality}, and \glspl{attendedModality}.  Each of these
correlations was evaluated across all models corresponding to a cluster,
\gls{standardModality}, and \gls{attendedModality}, as in
Figure~\ref{fig:correlationsWithBehavior}.  The ith null hypothesis for
behavioral measure $b$ was $H^b_i$=\{the correlation coefficient between
models' decoding accuracies and subjects' behavioral measure b for the ith
combination of cluster, \gls{standardModality}, and \gls{attendedModality}
equals is zero\}. Each family contained 56 pairs of hypothesis (14 clusters
$\times$ 2 \glspl{standardModality} $\times$ 2 \glspl{attendedModality}).
The multiple comparisons procedure in (c) was as that in (b) but using
\gls{ITC} values at the corresponding peak \gls{ITC} frequencies
(Section~\ref{sec:itcAndPeakITCFrequency}) instead of models' decoding powers.

\subsubsection{Calculation of robust correlation coefficients and corresponding
p-values}
\label{sec:robustCorrelationCoefficientsAndPValues}

We used skipped measures of correlation~\citep{wilcox12} to characterize the
association between pairs of variables in a way that was resistant to the
presence of outliers. The calculation of these measures, and the estimation
of their bootstrap confidence intervals, followed the procedures described in
\citep{wilcox12,pernetEtAl13}.

Skipped correlations are obtained by checking for the presence of outliers,
removing them, and applying some correlation coefficient to the remaining
data~\cite[][Chapter 9]{wilcox12}. We identified and removed bivariate outliers. Bivariate outliers were found using a projection method.
The underlying idea behind projection methods is that, if a mutlivariate data
point is an outlier, then it should be an outliers along some one-dimensional
projection of all data points. For each data point $\mathbf{X}_i$, the
projection method used in this manuscript orthogonally projected all data
points onto the line joining a robust estimate of the center of the data
cloud (given by the medians of the correlated variables) and $\mathbf{X}_i$,
and then found the outliers of this one-dimensional projection using the
boxplot rule \citep[which is based on interquartiles
distances,][]{friggeEtAl89}. A data point was declared outlier if it was an
outlier in any one-dimensional projection.
Bivariate outliers
were selected using the function \texttt{outpro} freely available from the web site of Dr.\ Rand Wilcox at
\url{http://dornsife.usc.edu/labs/rwilcox/software/}.

To estimate the association between single-trial decoding errors and
\glspl{SFPD} (e.g., Figure~\ref{fig:exampleSingleTrialAnalysis}f), we
used the skipped Pearson correlation coefficient (i.e., after outlier removal,
the remaining data points were correlated using the Pearson product moment
correlation coefficient; function \texttt{pcor} from Dr.\ Wilcox's website).
The use of this correlation coefficient requires that the marginals of the
data are approximately normal, as it was the case in the previous data. 

The averaged behavioral data tended to be bimodally distributed, with a group
of subjects displaying better performance and another group showing worse
performance. Thus, the averaged behavioral data was not approximately normal,
and it was not possible to use the skipped Pearson correlation coefficient to
asses its association with models' decodings. Instead, we quantified this
association using the skipped Spearman correlation coefficient (i.e., after
outlier removal, the remaining data points were correlated using the Spearman
rank correlation coefficient; function \texttt{spear} from Dr.\ Wilcox's
website).  The use of this correlation coefficient does not assume that the
distribution of the correlated variables is bivariate normal, but requires that
their association be monotonic, as it was the case for our averaged data.

To compute p-values corresponding to hypothesis tests of robust correlation
coefficients we used permutation tests following the procedure given in
\citet{pernetEtAl13}.  Given a dataset of two sequences to correlate, we (1)
removed bivariate outliers from this dataset, (2) computed the correlation
coefficient between the sequences in the outliers-removed dataset (using
functions \texttt{pcor} and \texttt{spear} from Dr.~Wilcox's website for Pearson
and Spearman correlation coefficients, respectively), (3) computed
nResamples=5000 correlation coefficients between the first sequence and a
random permutation of the second sequence in the outliers-removed dataset
(again using functions \texttt{pcor} or \texttt{spear}), and (4) we calculated a
two-sided p-value as the proportion of correlation coefficients computed in
step (3) whose absolute value was larger than that of the correlation
coefficient computed in step (2).

\subsubsection{ANOVA procedure}
\label{sec:anova}

The ANOVA procedures described below started with a model with a large number
of independent factors. We selected a subset of these factors using a stepwise
procedure based on the Akaike Information Criterion, as implemented in the
function \texttt{stepAIC} of the package \texttt{MASS} of
\texttt{R}~\citep{r12}. We then tested for significant differences in the mean
of the dependent variable across the selected subset of independent factors. We
report the F-value and p-value of all significant effects. For the significant
main effects we also report the p-value of posthoc tests for differences
across the levels of the corresponding factors, and the sign of these
differences.

\vspace{.5\baselineskip}
\noindent\textbf{Attentional effects in ANOVAs}

\noindent Attentional effects in the \gls{SFP} effect on \gls{ITPC} could
appear in ANOVAs as significant differences between models estimated from
attended and unattended stimuli of the same or different
\glspl{standardModality} (e.g., a significant difference between models
estimated from epochs aligned to the presentation time of visual standards
with attention directed to the visual or auditory modality, or a significant
difference between models estimated with attention directed to the visual
modality from epochs aligned to the presentation of visual or auditory
standards). Here we only study attentional effects in ANOVAs by comparing
models estimated from attended and unattended stimuli of the same
\gls{standardModality}. We do so because differences in models estimated from
stimuli of different \glspl{standardModality} could be due to differences in
brain responses to stimuli of distinct modalities and unrelated to attention.

As mentioned in Section~\ref{sec:characterizedEpochs}, for any
\gls{standardModality}, the mean number of epochs aligned to the presentation
of attended stimuli was larger than the mean number of epochs aligned to the
presentation of unattended stimuli. This difference invalidates the comparison
between models estimated from attended and unattended stimuli of the same
\gls{standardModality}, since any discrepancy between these models could
simply be due to differences in the amount of data used to estimate them.
Thus, attentional effects in the \gls{SFP} effect on \gls{ITPC} were
investigated using ANOVAs on datasets where the number of epochs aligned to
attended and unattended stimuli had been equalized, as described in
Section~\ref{sec:characterizedEpochs}.

\vspace{.5\baselineskip}
\noindent\textbf{Non-repeated-measures ANOVA}

\noindent We performed non-repeated-measures ANOVAs using the function
\texttt{aov} of the package \texttt{stats} of \texttt{R}~\citep{r12}.
We removed from the analysis those trials for which the absolute value of the
standardized residuals was greater than the 95\% quantile of the normal
distribution. 
When a significant main effect was observed in the ANOVA we created a set of
confidence intervals on the differences between the means of the levels of the
corresponding
factor with 95\% family-wise probability of coverage using the Tukey's `Honest
Significance Difference' method, as implemented in the function
\texttt{TukeyHSD} of the package \texttt{stats} of \texttt{R}~\citep{r12}.
For the pairs of levels with a significant difference in the mean of the
dependent variable we reported the p-value, corrected for multiple comparison,
of the test of the hypothesis that the previous difference is distinct from
zero. We also reported the sign of this difference extracted from the sign of
the corresponding confidence interval.

\vspace{.5\baselineskip}
\noindent\textbf{Repeated-measures ANOVA}

\noindent We performed repeated-measures ANOVAs by estimating mixed-effects
models [Pinheiro and Bates, 2000] with the function \texttt{lme} of package
nlme of \texttt{R}~\citep{r12}.
Trials for which the absolute value of the normalized residuals (i.e.,
standardized residuals pre-multiplied by the inverse square-root factor of the
estimated error correlation matrix) was greater than the 95\% quantile of the
normal distribution were removed from the analysis.
When a significant main effect was observed in the ANOVA, we
performed posthoc hypothesis tests using the function \texttt{glht} from the
package \texttt{multcomp} of \texttt{R}~\citep{r12}. This function controls
for multiple comparisons by using multivariate t, or Gaussian, distributions
to compute p-values associated with multiple hypothesis~\citep{bretzEtAl10}.

\subsubsection{Further information on ANOVAs}
\label{sec:anovaFurtherInfo}

\vspace{.5\baselineskip} \noindent\textbf{Maximum
SFPD}~(Section~\ref{sec:selectionMaxSFPD}):
We performed an initial repeated-measures ANOVA (see \emph{Repeated-measures
ANOVA} in Section~\ref{sec:anova}) using data from all models
significantly different from the intercept-only model (p\textless0.01,
likelihood-ratio permutation test, Section~\ref{sec:additionalStatInfo}) with
the natural logarithm of the maximum \gls{SFPD} as dependent variable (we
selected this transformation of the dependent variable to normalize the
distribution of the ANOVA residuals). The
input to the stepwise selection procedure was a model with
\gls{standardModality}, \gls{attendedModality}, their interaction, and cluster
of \glspl{IC} as independent variables. This procedure  selected a model with
\gls{standardModality}, \gls{attendedModality}, and their interaction as
independent variables. An ANOVA on the selected model showed significant main
effects for \gls{attendedModality} (F(1,615)=15.15, p=0.0001) and for the
interaction between \gls{standardModality} and \gls{attendedModality}
(F(1,515)=0.0122, p=0.0122). Since this interaction indicates an attentional
effect, we repeated the previous analysis in a dataset with equalized number
of epochs aligned to the presentation of attended an unattended standards (see
\emph{Attentional effect in ANOVAs} in Section~\ref{sec:anova}).  With the
equalized dataset, the model selection procedure selected a model with only
\gls{attendedModality} as independent factor, and an ANOVA on this model
showed a significant main effect of this factor (F(1,579)=8.45, p=0.0038).  A
posthoc analysis indicated that the maximum \gls{SFPD} was shorter for the
visual than the auditory \gls{attendedModality} (z=2.91, p=0.0018).

\vspace{.5\baselineskip} \noindent\textbf{Absolute value of the correlation
coefficient between models' decoding accuracies and subjects averaged
behavioral measures}~(Section~\ref{sec:correlationsWithBehavior}):
We performed a non-repeated-measures ANOVA (see \emph{Non-repeated-measures
ANOVA} in Section~\ref{sec:anova}) using data from all models
significantly different from the intercept-only model (p\textless0.01,
likelihood-ratio permutation test, Section~\ref{sec:additionalStatInfo}) with
the absolute value of the correlation coefficient between models' decoding
accuracies and subjects mean reaction times or error rates as dependent
variable.  The input to the stepwise selection procedure was a model with
\gls{standardModality}, \gls{attendedModality}, their interaction, cluster of
\glspl{IC}, and type of behavioral measure (i.e., error rate or mean reaction
times) as independent variables. This procedure  selected a model with
\gls{standardModality} and type of behavioral measures as independent
variables. An ANOVA on the selected model showed significant main effects of
\gls{attendedModality} (F(1,109)=8.22, p=0.005) and of the type of behavioral
measure (F(1,109)=7.26, p=0.0082). A posthoc analysis showed that the mean of
the absolute value of the correlation coefficient was larger for the visual
than the auditory \gls{standardModality} (p=0.005; Tukey test) and larger for
error rates than mean reaction times (p=0.0082; Tukey test). 

\vspace{.5\baselineskip}
\noindent\textbf{SFPD}~(Section~\ref{sec:characterizedEpochs}):
We performed a repeated-measures ANOVA using data from all epochs with the
cubic root of the \gls{SFPD} as dependent variable (we selected this
transformation of the dependent variable to normalize the distribution of the
ANOVA residuals).  The input to the stepwise selection procedure was a model
with \gls{standardModality}, \gls{attendedModality}, and their interaction as
independent variables. This procedure  selected a model with all independent
variables.  An ANOVA on the selected model showed a significant main effect
of \gls{attendedModality} (F(1,16264)=234.31, p\textless 0.0001).  A
posthoc analysis showed that the mean \gls{SFPD} was shorter in epochs where
attention was directed to the visual than to the auditory modality (z=15.31;
p\textless 2e-16).

\vspace{.5\baselineskip} \noindent\textbf{Time of Peak
Coefficient}~(Section~\ref{sec:timing}):
We performed a first repeated-measures ANOVA using
data from all models significantly different from the intercept-only model
(p\textless0.01, likelihood-ratio permutation test,
Section~\ref{sec:additionalStatInfo}) with the time of the largest peak as
dependent variable.  The input to the stepwise selection procedure was a model
with \gls{standardModality}, \gls{attendedModality}, group of clusters (as
shown in Figure~\ref{fig:clusters}), and their triple interaction as
independent variables. This procedure  selected all independent variables and
their interactions, but in a subsequent ANOVA only the effect due to the
interaction between \gls{standardModality} and groups of clusters (F(4,
230)=2.9026, p=0.0227) and the main effect of \gls{attendedModality} (F(1,
230)=4.2590, p=0.0402) remained significant. A posthoc analysis on
\gls{attendedModality} showed that the peak occurred earlier for auditory than
visual attention (z=1.842, p=0.0327, left black asterisk in
Figure~\ref{fig:timing}a). To disentangle the interaction between
\gls{standardModality} and group of clusters we performed a second and a third
repeated-measures ANOVA using only models corresponding to visual and auditory
standards, respectively.

The input to the stepwise procedure for the second repeated-measures ANOVA
for models corresponding to visual standards was a
model with peak time as dependent variable and \gls{attendedModality}, group of
clusters, and their interaction as independent variables. This procedure
selected all independent variables and interactions, but a subsequent ANOVA
showed  that only the main effects of \gls{attendedModality} (F(1,135)=11.1939,
p=0.0011) and groups of clusters (F(4,135)=4.8073, p=0.0022) were significant.
A posthoc analysis revealed that the peak of coefficients occurred earlier when
attention was oriented to the visual than auditory modality (z=3.512,
p=0.000455) and in central than occipital brain regions (z=4.163, p=3.14e-5,
red asterisks in Figure~\ref{fig:timing}b).

The input to the stepwise procedure for the third repeated-measures ANOVA
for models corresponding to auditory standards was
the same as that for models corresponding to the visual modality. This
procedure found no significant variable.

From Figure~\ref{fig:timing} it is apparent that the strongest modulations of
the timing of the \gls{SFP} effect on \gls{ITPC} occur in the central group of
clusters. To test this observation statistically, we performed a fourth
repeated-measures ANOVA restricted to models corresponding to the central group
of clusters. The input to the stepwise procedure was a model with peak time as
dependent variable and \gls{standardModality}, \gls{attendedModality}, and
their interaction as independent variables. This procedure selected the two
independent variables and their interaction.  However, the significance of the
interaction term disappeared when we repeated this analysis with a dataset with
equal number of epochs aligned to the presentation of attended and unattended
standards (see \emph{Attentional effect in ANOVAs} in Section~\ref{sec:anova}),
so we did not used the interaction term in the following.  A repeated-measures
ANOVA found significant main effects for \gls{standardModality} (F(1,
43)=14.2173, p=0.0005) and for \gls{attendedModality} (F(1, 43)=4.5061,
p=0.0396). A posthoc analysis showed that the peak coefficient occurred earlier
in models corresponding to auditory standards (z=3.406, p=0.000659; black
asterisks in Figure~\ref{fig:timing}b), and in models estimated with attention
directed to the auditory modality (z=2.178, p=0.029360; black asterisk next to
\emph{Central} cluster group in Figure~\ref{fig:timing}a).

\subsubsection{Identification of peaks in coefficients}
\label{sec:identificationOfCoefsPeaks}

We identified positive and negative peaks in the coefficients of decoding
models significantly different from the intercept-only model (p\textless 0.01,
likelihood-ratio permutation test, Section~\ref{sec:additionalStatInfo}). We
only searched for peaks at coefficients significantly different from zero (95\%
bootstrap CI for regression coefficients,
Section~\ref{sec:additionalStatInfo}), and occurring after 200~ms (to avoid
possible effects from the \gls{warningSignal}).  We filtered continuous subsets
of the coefficients with a finite impulse response filter of type I, order two,
and cutoff frequency of three times the peak \gls{ITC} frequency. We selected
in the filtered coefficients the local maxima larger than zero (local minima
smaller than zero) as the positive (negative) peaks.  From all positive and
negative peaks, and from all continuous subsets of the coefficients, we
selected the peak with maximum absolute value.

\subsubsection{Additional statistical information}
\label{sec:additionalStatInfo}

Most statistical figures presented in this paper are based on the bootstrap
method~\citep{efronAndTibshirani93}.

\noindent\textbf{Crossvalidated decodings} All models decodings were cross
validated using the leave-one-out method with the function \texttt{crossval}
of the package \texttt{bootstrap} of \texttt{R}~\citep{r12}.

\noindent\textbf{95\% bootstrap CIs for regression coefficients}
(Figure~\ref{fig:exampleSingleTrialAnalysis}e) We performed 2,000 ordinary
bootstrap resamples of the trials (with the function \texttt{boot} of the
package \texttt{boot} of \texttt{R}~\citep{r12}), and for each resample we
estimated a set of regression coefficients, as described in
Section~\ref{sec:linearRegressionModel}. Having estimated 2,000 sets of
regression coefficients, we computed 95\% percentile confidence intervals with
the function \texttt{boot.ci} of package \texttt{boot} of
\texttt{R}~\citep{r12}.

\noindent\textbf{95\% bootstrap CI for difference in paired means} of
correlation coefficients obtained from the original versus a surrogate dataset
(Figure~\ref{fig:controls}).  For each pair of correlation coefficients
(obtained by correlating \glspl{SFPD} and their decodings in the
original and surrogate datasets) we subtracted the correlation coefficient from
the original minus the surrogate dataset. Then we performed 2,000 ordinary
resamples (with the function \texttt{boot} of the package \texttt{boot} of
\texttt{R}~\citep{r12}) and computed a 95\% percentile bootstrap confidence
interval with the function \texttt{boot.ci} of package \texttt{boot} of
\texttt{R}~\citep{r12}.

\noindent\textbf{95\% bootstrap CI for difference in averaged DMP} between
trials with the longest and shortest \gls{SFPD}
(Figures~\ref{fig:avgDMPsAndCoefs}a-c). We calculated 2,000 stratified
bootstrap resamples of the 20\% of trials with the shortest and longest
\glspl{SFPD}, using the function \texttt{boot} of the package
\texttt{boot} of \texttt{R}~\citep{r12} with the \texttt{strata} option).  For
each resample, and for each sample point, we subtracted from the mean \gls{DMP}
across the 20\% of trials with longest \gls{SFPD} the mean \gls{DMP}
across the 20\% of trials with shortest \gls{SFP}.  Having estimated 2,000
bootstrap differences in average \gls{DMP} values between trials with the
longest and shortest \gls{SFPD}, we computed a 95\% percentile
bootstrap confidence interval for these differences with the function
\texttt{boot.ci} of package \texttt{boot} of \texttt{R}~\citep{r12}.

\noindent\textbf{Likelihood-ratio permutation test for linear model} We used it
to test the alternative hypothesis, H1, that there is association between the
dependent and independent variables of a linear regression model against the
alternative one, H0, that there is no association. 
The statistic for this test is the logarithm of the ratio between the
likelihood of the data given a full model and that given a null model.
A full model associates all independent variables with the dependent one, as in~(\ref{eq:linearRegressionModel}), while a null model associates only a
constant term with the dependent variable. The likelihood of the data given a
model is derived from (\ref{eq:likelihoodFunction}).
The test was conducted as follows. We first measured the value of the test
statistic in the original dataset, $T_{obs}$. Then, we built the distribution
of the test statistic under H0, by generating R=2000 datasets where there was no
association between the dependent and independent variables (by permuting the
dependent value assigned to independent values) and measuring the test
statistic in these datasets. The one-sided p-value of the test is the
proportion of samples of the test statistic under H0 that are larger than
$T_{obs}$.

\noindent\textbf{Normalized cross correlation}
(Section~\ref{sec:directEvidence}) We normalized the cross correlation (at lag
zero) between two time series in such a way that the normalized correlation of
a time series with itself is one. We defined:

\begin{eqnarray*}
ncor(x,y)=\frac{cor(x,y)}{\sqrt{cor(x,x)}\sqrt{cor(y,y)}}
\end{eqnarray*}

\noindent where $x$ and $y$ are two time series, and $ncor(x,y)$ and $cor(x,y)$
are the normalized and unnormalized correlations, respectively, between these
time series.

\subsection{Supplementary figures}
\label{sec:supplementaryFigures}

Figure~\ref{fig:avgDMPsC04Visualswitchvision} shows average \glspl{DMP} for the
20\% trials with the shortest and longest \glspl{SFPD} in \glspl{IC} from the
left parieto-occipital cluster~04..
Figure~\ref{fig:exampleForeperiodEffectsOnRTAndDetection} plots examples of
deviant foreperiod effects on reaction times and detectability.
Figure~\ref{fig:clusters3D} plots axial, coronal, and sagital views of
all clusters.
Figure~\ref{fig:histDLatenciesAndSFPDs} plots the distribution of deviant
latencies and \glspl{SFPD}.

\begin{sidewaysfigure}
\begin{center}
\includegraphics{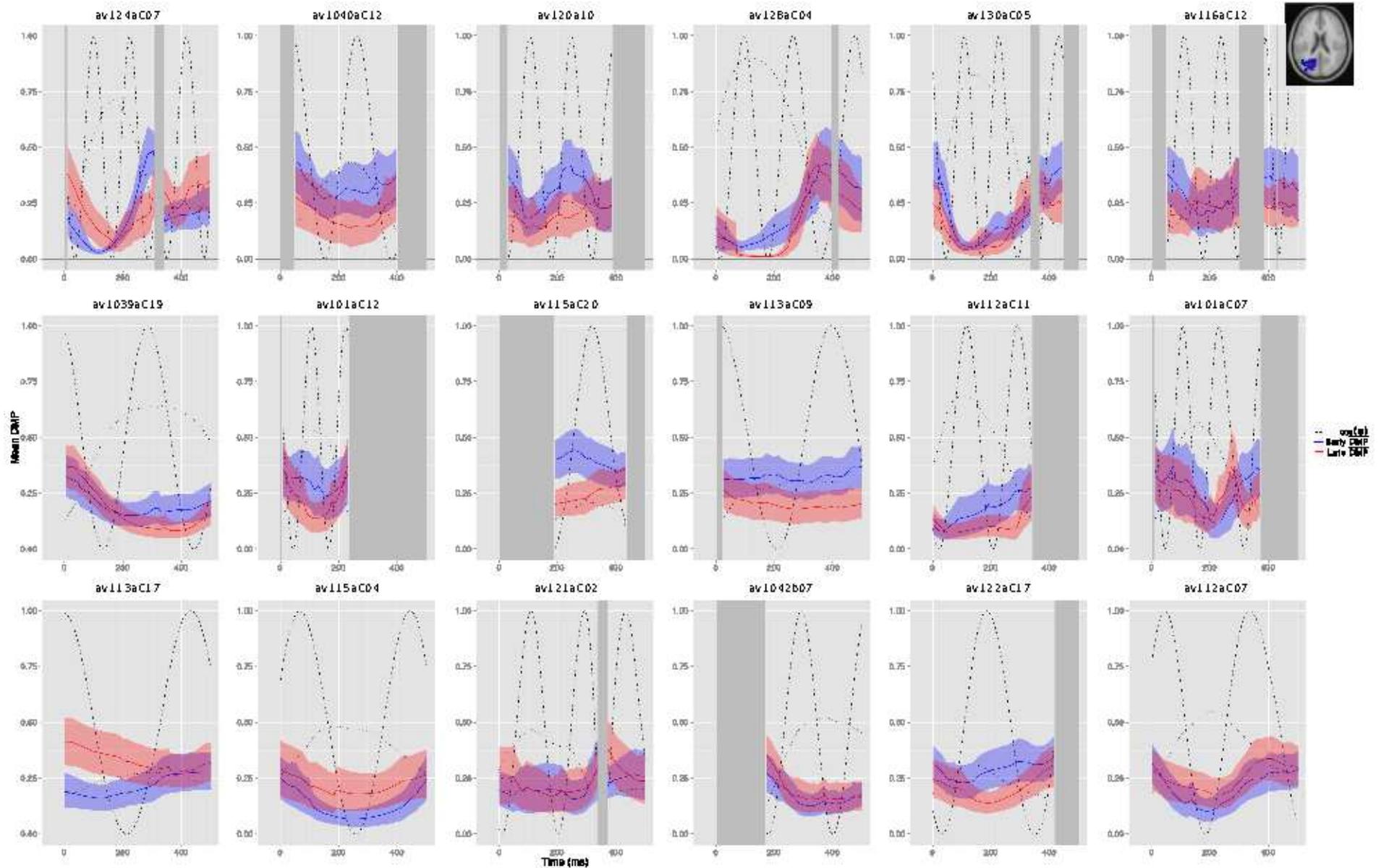}

\caption{Average \glspl{DMP} for the 20\% trials with the shortest (blue
traces) and longest (red traces) \glspl{SFPD} computed from trials from the
left-parieto-occipital cluster~04 (top-right inset) and attended visual stimuli.
Panels are sorted from left to right and top to bottom by decreasing
decoding power of the corresponding model. These averages suggest that the
averaged \gls{DMP} oscillates at a frequency around 1~Hz and that the phase
of these oscillations at time zero is different between trials closest and
furthest away from the \gls{warningSignal}.}

\label{fig:avgDMPsC04Visualswitchvision}
\end{center}
\end{sidewaysfigure}

\begin{figure}
\begin{center}
\includegraphics{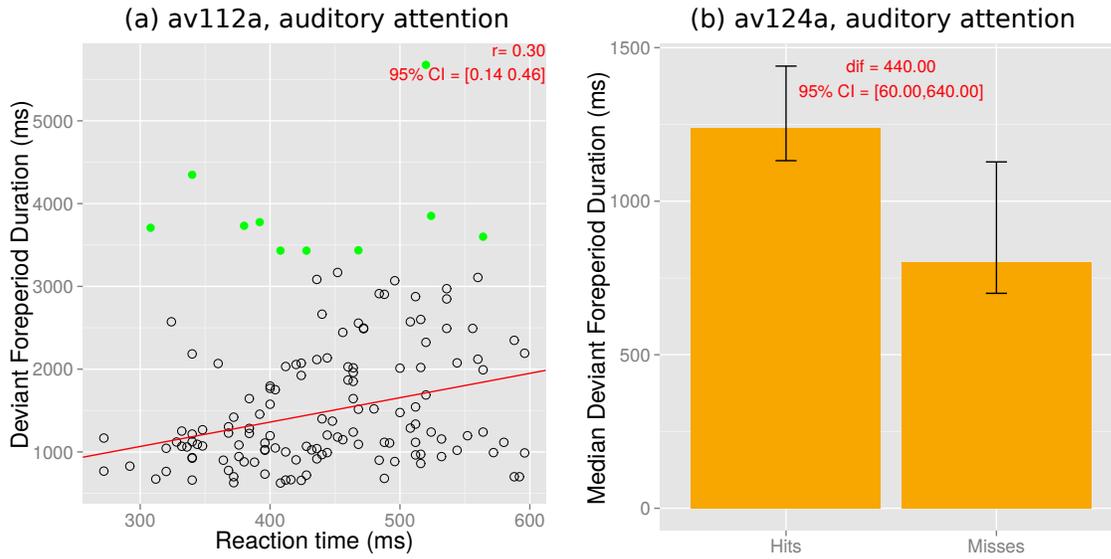}

\caption{Examples of significant deviant foreperiod effects on reaction times
(a) and detectability (b). (a) Deviant foreperiod durations as a function of
subject reaction times, for subject av112a and the auditory
\gls{attendedModality}. Points colored in green indicate outliers detected in
the computation of the robust correlation coefficient
(Section~\ref{sec:robustCorrelationCoefficientsAndPValues}). (b) Median deviant
foreperiod duration for hits and misses, for subject av124a and the auditory
\gls{attendedModality}.}

\label{fig:exampleForeperiodEffectsOnRTAndDetection}
\end{center}
\end{figure}

\begin{figure}
\begin{center}
\includegraphics{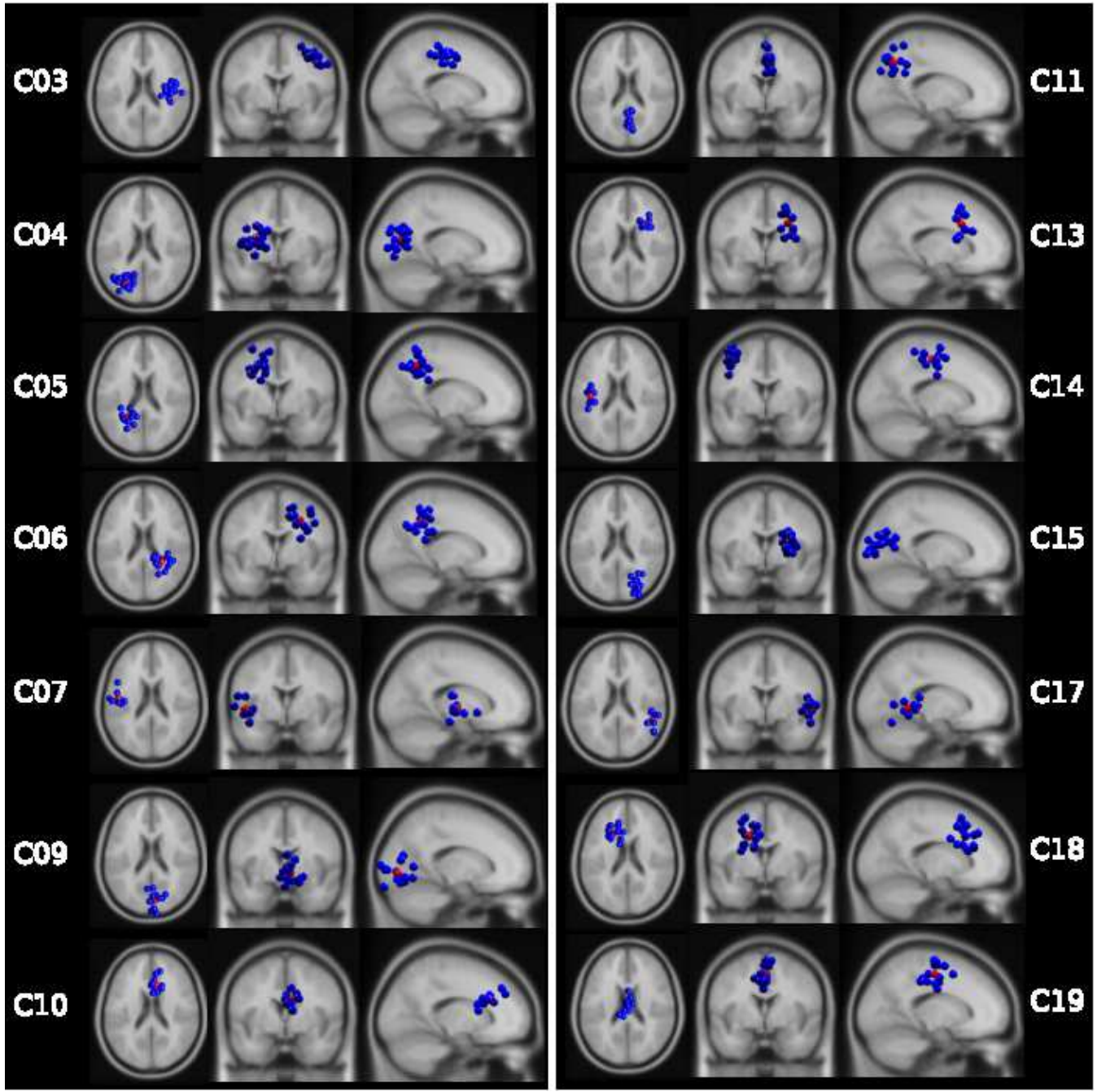}

\caption{Clusters if \glspl{IC}.  Left, center, and right columns correspond to
axial, coronal, and sagital views, respectively.}

\label{fig:clusters3D}
\end{center}
\end{figure}

\begin{figure}
\begin{center}
\includegraphics{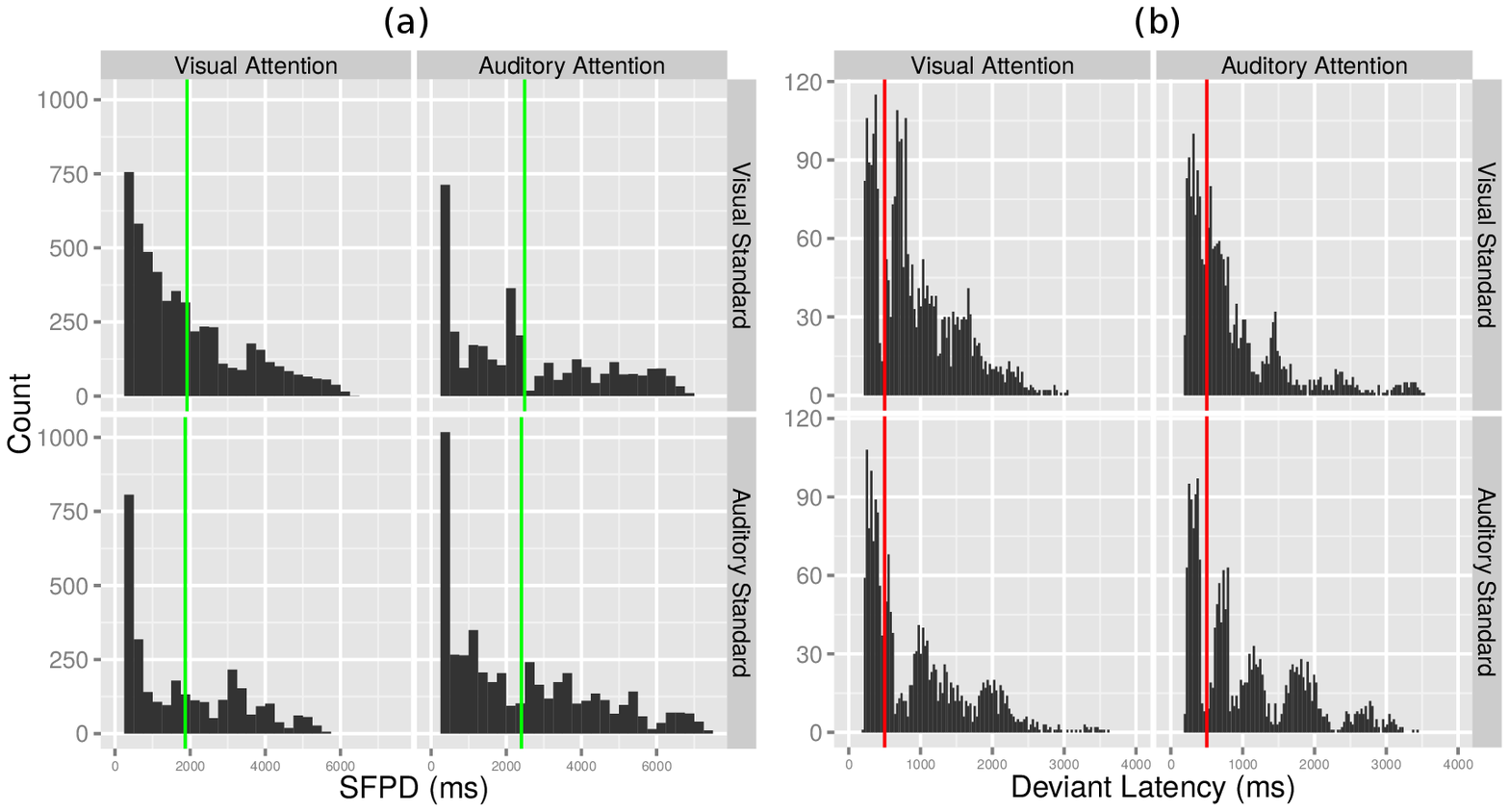}

\caption{Distributions of \glspl{SFPD} (a) and of latencies of deviants (b)
across \glspl{standardModality} (rows) and \glspl{attendedModality} (columns). 
(a) The mean \gls{SFPD} in epochs corresponding to visual attention (1.9~sec,
left panels in (a)) is significantly shorter than that in epochs corresponding
to auditory attention (2.4~sec, right panels in (a)), as for the optimal
maximum \glspl{SFPD} of models (see Section~\ref{sec:selectionMaxSFPD}).
(b) To avoid possible response-related movements artifacts in the EEG, epochs
including deviants with latencies shorter than 500~ms were excluded from the
analysis.  The number of such deviants is given by the sum of counts to the left
of the red vertical lines in (b).}

\label{fig:histDLatenciesAndSFPDs}
\end{center}
\end{figure}

\subsection{Supplementary tables}
\label{sec:supplementaryTables}

Table~\ref{table:clustersInfo} provides the Talairach coordinates and
anatomical labels of the centroids of the clusters in
Figure~\ref{fig:clusters}. For each cluster, \gls{standardModality}, and
\gls{attendedModality}, Table~\ref{table:nModelsSignCorPredictionsVsSFPDs}
gives the number of models with decodings significant correlated with
\glspl{SFPD}.

\begin{table}[ht]

\caption{Information about clusters of \glspl{IC}. The anatomical labels
associated with the Talairach coordinates of the clusters' centroids
were extracted using the Talairach
client~\citep{lancasterEtAl97,lancasterEtAl00}}

\small
\begin{center}
\begin{tabular}{|c|c|c|c|c|c|c|c|c|c|}\hline
\multicolumn{1}{|c}{Cluster} & \multicolumn{1}{|c}{No.} &
\multicolumn{1}{|c}{No.} & \multicolumn{3}{|c}{Talairach} &
\multicolumn{1}{|c}{Hemisphere} & \multicolumn{1}{|c}{Lobe} &
\multicolumn{1}{|c}{Gyrus} & \multicolumn{1}{|c|}{Brodmann}\\

\multicolumn{1}{|c}{} & \multicolumn{1}{|c}{Subjects} & 
\multicolumn{1}{|c}{ICs} & \multicolumn{1}{|c}{X} & \multicolumn{1}{|c}{Y} & \multicolumn{1}{|c}{Z} &
\multicolumn{1}{|c}{} & \multicolumn{1}{|c}{} &
\multicolumn{1}{|c}{} & \multicolumn{1}{|c|}{Area}\\ \hline \hline

% 2 & 19 & 64 & 43 & -39 & 34 & Right & Parietal Lobe & Supramarginal Gyrus & 40\\ \hline
3 & 12 & 17 & 45 & -13 & 43 & Right & Frontal Lobe & Precentral Gyrus & 4\\
\hline
4 & 14 & 18 & -19 & -74 & 9 & Left & Occipital Lobe & Cuneus & 17\\ \hline
5 & 11 & 14 & -14 & -42 & 33 & Left & Limbic Lobe & Cingulate Gyrus & 31\\
\hline
6 & 12 & 17 & 43 & -39 & 34 & Right & Parietal Lobe & Supramarginal Gyrus &
40\\ \hline
7 & 8 & 9 & -38 & -18 & 1 & Left & Sub-lobar & Claustrum & *\\ \hline
8 & 3 & 4 & 56 & -1 & 12 & Right & Frontal Lobe & Precentral Gyrus & 6\\ \hline
9 & 12 & 14 & 3 & -62 & 7 & Right & Limbic Lobe & Posterior Cingulate & 30\\ \hline
10 & 9 & 10 & 6 & 17 & 18 & Right & Limbic Lobe & Anterior Cingulate & 33\\ \hline
11 & 12 & 13 & 8 & -46 & 28 & Right & Parietal Lobe & Precuneus & 31\\ \hline
12 & 3 & 4 & 7 & 18 & -50 & \multicolumn{4}{|c|}{No Gray Matter found} \\ \hline
13 & 11 & 12 & 30 & 10 & 34 & Right & Frontal Lobe & Middle Frontal Gyrus & 8\\
\hline
14 & 10 & 13 & -35 & -25 & 53 & Left & Frontal Lobe & Precentral Gyrus & 4\\ \hline
15 & 16 & 20 & 23 & -87 & 10 & Right & Occipital Lobe & Middle Occipital Gyrus
& 18\\ \hline
16 & 4 & 5 & 14 & 45 & -3 & Right & Limbic Lobe & Anterior Cingulate & 32\\ \hline
17 & 7 & 9 & 48 & -33 & 4 & Right & Temporal Lobe & Superior Temporal Gyrus & 22\\ \hline
18 & 12 & 15 & -6 & 21 & 31 & Left & Limbic Lobe & Cingulate Gyrus & 32\\
\hline
19 & 12 & 14 & 2 & -7 & 42 & Right & Limbic Lobe & Cingulate Gyrus & 24\\
\hline
\end{tabular}

\end{center}
\normalsize

\label{table:clustersInfo}
\end{table}

\begin{table}

\caption{Number (n) and proportion (\%) of models with significant correlations
(adj\_p\textless0.05) between models' decodings and experimental \glspl{SFPD}. Cells highlighted
in color corresponds to clusters, standard modalities, and attended modalities
with more than 40\% of significant correlations}

\begin{center}
\begin{tabular}{|c|c|c c|c c|}\hline
\multicolumn{1}{|c}{Cluster} & \multicolumn{1}{|c|}{Standard}  & \multicolumn{2}{c|}{Visual Attention}  & \multicolumn{2}{c|}{Auditory Attention}\\
\multicolumn{1}{|c}{} & \multicolumn{1}{|c|}{Modality}  & \multicolumn{1}{c}{n} & \multicolumn{1}{c|}{\%} & \multicolumn{1}{c}{n} & \multicolumn{1}{c|}{\%}\\ \hline \hline
03 & Visual & 05 & 0.29\% & 06 & 0.35\% \\
 & Auditory & 02 & 0.12\% & 02 & 0.12\% \\
\hline
04 & Visual & \cellcolor{NavyBlue}{08} & \cellcolor{NavyBlue}{0.44\%} & \cellcolor{NavyBlue}{08} & \cellcolor{NavyBlue}{0.44\%} \\
 & Auditory & 02 & 0.11\% & 05 & 0.28\% \\
\hline
05 & Visual & 04 & 0.29\% & 02 & 0.14\% \\
 & Auditory & 01 & 0.07\% & 03 & 0.21\% \\
\hline
06 & Visual & 04 & 0.24\% & 03 & 0.18\% \\
 & Auditory & 05 & 0.29\% & 06 & 0.35\% \\
\hline
07 & Visual & 01 & 0.11\% & 01 & 0.11\% \\
 & Auditory & 00 & 0.00\% & 03 & 0.33\% \\
\hline
09 & Visual & \cellcolor{NavyBlue}{06} & \cellcolor{NavyBlue}{0.43\%} & 05 & 0.36\% \\
 & Auditory & 03 & 0.21\% & 05 & 0.36\% \\
\hline
10 & Visual & 01 & 0.10\% & 02 & 0.20\% \\
 & Auditory & 00 & 0.00\% & \cellcolor{NavyBlue}{05} & \cellcolor{NavyBlue}{0.50\%} \\
\hline
11 & Visual & \cellcolor{NavyBlue}{06} & \cellcolor{NavyBlue}{0.46\%} & 03 & 0.23\% \\
 & Auditory & 03 & 0.23\% & 04 & 0.31\% \\
\hline
13 & Visual & 03 & 0.25\% & 02 & 0.17\% \\
 & Auditory & 03 & 0.25\% & 03 & 0.25\% \\
\hline
14 & Visual & 04 & 0.31\% & 05 & 0.38\% \\
 & Auditory & 04 & 0.31\% & 02 & 0.15\% \\
\hline
15 & Visual & \cellcolor{NavyBlue}{11} & \cellcolor{NavyBlue}{0.55\%} & 07 & 0.35\% \\
 & Auditory & 01 & 0.05\% & 07 & 0.35\% \\
\hline
17 & Visual & 02 & 0.22\% & 02 & 0.22\% \\
 & Auditory & 03 & 0.33\% & 03 & 0.33\% \\
\hline
18 & Visual & 01 & 0.07\% & 03 & 0.20\% \\
 & Auditory & \cellcolor{NavyBlue}{06} & \cellcolor{NavyBlue}{0.40\%} & 03 & 0.20\% \\
\hline
19 & Visual & 04 & 0.29\% & 02 & 0.14\% \\
 & Auditory & 03 & 0.21\% & 04 & 0.29\% \\
\hline
\end{tabular}

\end{center}
\label{table:nModelsSignCorPredictionsVsSFPDs}
\end{table}

\begin{table}[ht]

\caption{Correlations between the decoding accuracy of models and subjects'
detection error rates. The decoding accuracy of a model is quantified with the
correlation coefficient between decodings of the model and experimental
\glspl{SFPD}. Each cell shows the correlation coefficient, r, and
p-values unadjusted, p, and adjusted, adj\_p, for multiple comparisons. Blue
cells highlight correlations with p\textless0.05 and the red cell corresponds
to adj\_p=0.05.}

\begin{center}
\begin{tabular}{|c|c|c c c|c c c|}\hline
\multicolumn{1}{|c}{Cluster} & \multicolumn{1}{|c|}{Standard}  & \multicolumn{3}{c|}{Visual Attention}  & \multicolumn{3}{c|}{Auditory Attention}\\
\multicolumn{1}{|c}{} & \multicolumn{1}{|c|}{Modality}  & \multicolumn{1}{c}{r} & \multicolumn{1}{c}{p} & \multicolumn{1}{c|}{adj\_p} & \multicolumn{1}{c}{r} & \multicolumn{1}{c}{p} & \multicolumn{1}{c|}{adj\_p}\\ \hline \hline
03 & Visual & 0.29 & 0.3054 & 1.00 & 0.30 & 0.2714 & 1.00 \\
 & Auditory & -0.07 & 0.8022 & 1.00 & -0.20 & 0.4766 & 1.00 \\
\hline
04 & Visual & \cellcolor{NavyBlue}{-0.64} & \cellcolor{NavyBlue}{0.0048} & \cellcolor{NavyBlue}{0.38} & 0.16 & 0.5212 & 1.00 \\
 & Auditory & 0.06 & 0.8360 & 1.00 & -0.12 & 0.6572 & 1.00 \\
\hline
05 & Visual & -0.50 & 0.0916 & 1.00 & \cellcolor{NavyBlue}{-0.70} & \cellcolor{NavyBlue}{0.0162} & \cellcolor{NavyBlue}{0.77} \\
 & Auditory & -0.22 & 0.4846 & 1.00 & -0.10 & 0.7612 & 1.00 \\
\hline
06 & Visual & \cellcolor{NavyBlue}{-0.54} & \cellcolor{NavyBlue}{0.0336} & \cellcolor{NavyBlue}{0.95} & -0.33 & 0.2068 & 1.00 \\
 & Auditory & 0.15 & 0.5650 & 1.00 & -0.23 & 0.3694 & 1.00 \\
\hline
07 & Visual & \cellcolor{NavyBlue}{-0.74} & \cellcolor{NavyBlue}{0.0444} & \cellcolor{NavyBlue}{0.98} & 0.36 & 0.3490 & 1.00 \\
 & Auditory & -0.17 & 0.6548 & 1.00 & -0.59 & 0.1016 & 1.00 \\
\hline
09 & Visual & -0.20 & 0.4902 & 1.00 & -0.43 & 0.1456 & 1.00 \\
 & Auditory & -0.34 & 0.2412 & 1.00 & -0.25 & 0.3942 & 1.00 \\
\hline
10 & Visual & -0.29 & 0.4274 & 1.00 & 0.49 & 0.1706 & 1.00 \\
 & Auditory & -0.18 & 0.6250 & 1.00 & -0.05 & 0.8794 & 1.00 \\
\hline
11 & Visual & 0.47 & 0.1218 & 1.00 & -0.29 & 0.3104 & 1.00 \\
 & Auditory & -0.43 & 0.1388 & 1.00 & -0.34 & 0.2586 & 1.00 \\
\hline
13 & Visual & -0.19 & 0.5626 & 1.00 & 0.26 & 0.4390 & 1.00 \\
 & Auditory & -0.23 & 0.5264 & 1.00 & 0.11 & 0.7402 & 1.00 \\
\hline
14 & Visual & -0.20 & 0.4964 & 1.00 & 0.34 & 0.2938 & 1.00 \\
 & Auditory & 0.36 & 0.2404 & 1.00 & -0.21 & 0.5538 & 1.00 \\
\hline
15 & Visual & 0.23 & 0.3648 & 1.00 & -0.08 & 0.7376 & 1.00 \\
 & Auditory & -0.38 & 0.1298 & 1.00 & -0.09 & 0.7108 & 1.00 \\
\hline
17 & Visual & 0.06 & 0.8684 & 1.00 & -0.29 & 0.4408 & 1.00 \\
 & Auditory & 0.11 & 0.7900 & 1.00 & -0.08 & 0.8194 & 1.00 \\
\hline
18 & Visual & -0.10 & 0.7612 & 1.00 & -0.11 & 0.7054 & 1.00 \\
 & Auditory & 0.07 & 0.8038 & 1.00 & -0.25 & 0.3922 & 1.00 \\
\hline
19 & Visual & -0.07 & 0.8098 & 1.00 & \cellcolor{Red}{-0.91} & \cellcolor{Red}{0.0004} & \cellcolor{Red}{0.05} \\
 & Auditory & \cellcolor{NavyBlue}{-0.62} & \cellcolor{NavyBlue}{0.0216} & \cellcolor{NavyBlue}{0.85} & -0.07 & 0.8002 & 1.00 \\
\hline
\end{tabular}

\end{center}
\label{table:statsErrorRates}
\end{table}

\begin{table}[ht]

\caption{Correlations between the decoding accuracy of models and subjects' mean
reaction times. Same format as Table~\ref{table:statsErrorRates}.}

\begin{center}
\begin{tabular}{|c|c|c c c|c c c|}\hline
\multicolumn{1}{|c}{Cluster} & \multicolumn{1}{|c|}{Standard}  & \multicolumn{3}{c|}{Visual Attention}  & \multicolumn{3}{c|}{Auditory Attention}\\
\multicolumn{1}{|c}{} & \multicolumn{1}{|c|}{Modality}  & \multicolumn{1}{c}{r} & \multicolumn{1}{c}{p} & \multicolumn{1}{c|}{adj\_p} & \multicolumn{1}{c}{r} & \multicolumn{1}{c}{p} & \multicolumn{1}{c|}{adj\_p}\\ \hline \hline
03 & Visual & -0.30 & 0.2574 & 1.00 & 0.02 & 0.9392 & 1.00 \\
 & Auditory & -0.25 & 0.3218 & 1.00 & 0.16 & 0.5410 & 1.00 \\
\hline
04 & Visual & -0.15 & 0.6102 & 1.00 & 0.22 & 0.3654 & 1.00 \\
 & Auditory & -0.04 & 0.8712 & 1.00 & 0.14 & 0.5598 & 1.00 \\
\hline
05 & Visual & 0.13 & 0.6412 & 1.00 & \cellcolor{NavyBlue}{-0.57} & \cellcolor{NavyBlue}{0.0318} & \cellcolor{NavyBlue}{0.84} \\
 & Auditory & 0.09 & 0.7596 & 1.00 & -0.24 & 0.4150 & 1.00 \\
\hline
06 & Visual & -0.15 & 0.5712 & 1.00 & -0.14 & 0.5892 & 1.00 \\
 & Auditory & 0.48 & 0.0534 & 0.96 & -0.16 & 0.5414 & 1.00 \\
\hline
07 & Visual & 0.20 & 0.5940 & 1.00 & -0.12 & 0.7456 & 1.00 \\
 & Auditory & 0.11 & 0.7720 & 1.00 & -0.03 & 0.9420 & 1.00 \\
\hline
09 & Visual & -0.01 & 0.9686 & 1.00 & 0.09 & 0.7458 & 1.00 \\
 & Auditory & 0.13 & 0.6486 & 1.00 & 0.11 & 0.6886 & 1.00 \\
\hline
10 & Visual & -0.40 & 0.2706 & 1.00 & 0.54 & 0.1098 & 1.00 \\
 & Auditory & 0.36 & 0.3294 & 1.00 & -0.12 & 0.7308 & 1.00 \\
\hline
11 & Visual & 0.13 & 0.6674 & 1.00 & -0.02 & 0.9430 & 1.00 \\
 & Auditory & -0.45 & 0.1306 & 1.00 & 0.14 & 0.6354 & 1.00 \\
\hline
13 & Visual & 0.12 & 0.6910 & 1.00 & 0.58 & 0.0524 & 0.95 \\
 & Auditory & -0.01 & 0.9682 & 1.00 & 0.35 & 0.2622 & 1.00 \\
\hline
14 & Visual & -0.34 & 0.2598 & 1.00 & -0.03 & 0.9210 & 1.00 \\
 & Auditory & -0.25 & 0.4040 & 1.00 & -0.07 & 0.7996 & 1.00 \\
\hline
15 & Visual & -0.04 & 0.8536 & 1.00 & -0.00 & 0.9890 & 1.00 \\
 & Auditory & -0.07 & 0.7526 & 1.00 & -0.06 & 0.8010 & 1.00 \\
\hline
17 & Visual & -0.23 & 0.5416 & 1.00 & -0.07 & 0.8538 & 1.00 \\
 & Auditory & 0.03 & 0.9172 & 1.00 & -0.10 & 0.7754 & 1.00 \\
\hline
18 & Visual & \cellcolor{NavyBlue}{-0.56} & \cellcolor{NavyBlue}{0.0270} & \cellcolor{NavyBlue}{0.80} & 0.14 & 0.6256 & 1.00 \\
 & Auditory & 0.01 & 0.9728 & 1.00 & -0.13 & 0.6502 & 1.00 \\
\hline
19 & Visual & \cellcolor{NavyBlue}{-0.57} & \cellcolor{NavyBlue}{0.0336} & \cellcolor{NavyBlue}{0.86} & -0.35 & 0.2044 & 1.00 \\
 & Auditory & -0.16 & 0.5780 & 1.00 & 0.18 & 0.5258 & 1.00 \\
\hline
\end{tabular}

\end{center}
\label{table:statsMeanRTs}
\end{table}

\subsection{Code and sample data}

To facilitate the application of the single-trial decoding methods described
in this manuscript to other EEG studies, we provide the \texttt{R}~\citep{r12}
code implementing these methods, the EEG data from the left
parieto-occipital cluster~04, and instruction on how to use this code to
analyze this data at
\url{http://sccn.ucsd.edu/~rapela/avshift/codeAndSampleData/}

\end{document}